\documentclass[12pt]{article}

\usepackage[pagewise,mathlines]{lineno}
\usepackage{mathrsfs}
\usepackage{amsmath, amsthm}
\usepackage{lipsum}
\usepackage{amssymb,url}
\usepackage{bbm}
\usepackage{chngcntr}
\usepackage{mathtools, nccmath}
\usepackage{algorithm}
\usepackage{algpseudocode}
\usepackage{multirow}
\usepackage{array}
\usepackage{comment}
\usepackage{enumitem}
\usepackage{bm}
\usepackage{booktabs}
\usepackage{makecell}
\usepackage[toc,page]{appendix}
\usepackage{enumitem,kantlipsum}
\usepackage[font=small,labelfont=bf]{caption}
\newcommand\norm[1]{\left\lVert#1\right\rVert}

\usepackage{booktabs}
\usepackage{dsfont}

\usepackage{graphicx}
\usepackage{subcaption}
\usepackage{booktabs}
\usepackage{float}
\allowdisplaybreaks

\newtheorem{assumption}{Assumption}
\newtheorem{Theorem}{Theorem}
\newtheorem{Proposition}{Proposition}
\newtheorem{Lemma}{Lemma}

\newtheorem{Example}{Example}

\newtheorem{Corollary}{Corollary}
\newtheorem{remark}{Remark}

\newtheorem{Definition}{Definition}

\newcommand{\indep}{\perp\!\!\!\perp}

\usepackage{hyperref}

\usepackage{titlesec}
\titleformat{\section}{\normalfont\large\bfseries}{\thesection}{1em}{}
\titleformat{\subsection}{\normalfont\normalsize\bfseries}{\thesubsection}{1em}{}

\usepackage{natbib}
\usepackage{bibunits}
\defaultbibliographystyle{apalike}
\defaultbibliography{reference}
\newcommand{\dd}{\mathop{}\!\mathrm d}

\newcommand*\patchAmsMathEnvironmentForLineno[1]{%
  \expandafter\let\csname old#1\expandafter\endcsname\csname #1\endcsname
  \expandafter\let\csname oldend#1\expandafter\endcsname\csname end#1\endcsname
  \renewenvironment{#1}%
     {\linenomath\csname old#1\endcsname}%
     {\csname oldend#1\endcsname\endlinenomath}}%
\newcommand*\patchBothAmsMathEnvironmentsForLineno[1]{%
  \patchAmsMathEnvironmentForLineno{#1}%
  \patchAmsMathEnvironmentForLineno{#1*}}%
\AtBeginDocument{%
  \patchBothAmsMathEnvironmentsForLineno{equation}%
  \patchBothAmsMathEnvironmentsForLineno{align}%
  \patchBothAmsMathEnvironmentsForLineno{flalign}%
  \patchBothAmsMathEnvironmentsForLineno{alignat}%
  \patchBothAmsMathEnvironmentsForLineno{gather}%
  \patchBothAmsMathEnvironmentsForLineno{multline}%
}

\usepackage{setspace}
\usepackage{etoolbox}
\appto\normalsize{%
  \setlength{\abovedisplayskip}{4pt plus 2pt minus 2pt}%
  \setlength{\belowdisplayskip}{4pt plus 2pt minus 2pt}%
  \setlength{\abovedisplayshortskip}{1pt plus 1pt}%
  \setlength{\belowdisplayshortskip}{2pt plus 1pt minus 1pt}%
}
\begin{document}

\def\spacingset#1{\renewcommand{\baselinestretch}{#1}\small\normalsize}

\begin{bibunit}

\spacingset{1}

\newcommand{\anon}{1}

\if1\anon
{
  \title{\bf Omitted Variable Bias in Difference-in-Differences Designs}
  \author{
    Juejue Wang \\
    Department of Statistics, University of Washington \\
    and \\
    Pedro H.C. Sant'Anna \\
    Department of Economics, Emory University \\
    and \\
    Victor Chernozhukov \\
    Department of Economics, MIT \\
    and \\
    Carlos Cinelli \\
    Department of Statistics, University of Washington}
  \date{}
  \maketitle
} \fi

\if0\anon
{
  \title{\bf Omitted Variable Bias in Difference-in-Differences Designs}
  \author{
    }
  \date{}
  \maketitle
} \fi

\begin{abstract}
We study the omitted variable bias (OVB) problem in canonical difference-in-differences (DiD) designs when unobserved confounding induces departures from the parallel trends assumption. Our results provide a novel characterization of the OVB formula for the average treatment effect on the treated (ATT), which is of independent interest. We show how the ATT bias is mainly governed by the strength of confounding in the treatment assignment mechanism and provide alternative ways of quantifying this strength, such as (i) changes in the average odds of treatment among the treated, (ii) confounding imbalance between treated and control units, or (iii) variation explained in treatment odds among the untreated. 
Building on these results, we offer sensitivity statistics for routine reporting, describing the minimum strength of confounding required to overturn the conclusions of a DiD study, as well as formal bounds on the strength of confounders based on comparisons to observed covariates or pre-trends. Finally, we provide flexible and efficient statistical inference methods for the bounds on ATT, which can leverage modern machine learning algorithms for estimation. We demonstrate the utility of our approach in an empirical example that estimates the effects of minimum wage on teen employment.
\end{abstract}

\vfill

\newpage
\spacingset{1.73}

\section{Introduction}

Difference-in-differences (DiD) has become the most widely used research design for causal inference with observational data in economics and other empirical sciences \citep{goldsmith2024tracking}. Its appeal is easy to understand. The method is simple to implement, makes modest demands on the data (it requires only a treated and control group, observed before and after treatment takes place), and it even allows for some types of treatment selection. In its canonical form, DiD identifies the average treatment effect on the treated (ATT) under a parallel trends assumption (PTA), which states that, in the absence of treatment, the average outcomes of treated and control units would have evolved in parallel over time---possibly only after conditioning on a set of observed pre-treatment covariates. As with any method for causal inference, however, the reliability of DiD rests on the plausibility of its identifying assumption. Because the PTA concerns counterfactual trajectories that are never observed for treated units after treatment, it is fundamentally untestable, and investigators marshal what evidence they can to defend its plausibility in the context of their investigations.

By far the most common way to gather such evidence is to perform placebo tests using pre-treatment information. Researchers typically check whether the outcomes of treated and control groups evolve in parallel before treatment, and read the absence of differential pre-treatment trends (often known as ``pre-trends'') as evidence for post-treatment parallel trends. Another common practice is to check whether treated and control groups are balanced on observed pre-treatment covariates that are also thought to be determinants of changes in the untreated potential outcome. While null findings on either of these tests are consistent with the parallel trends assumption, they are not dispositive. Pre-trends speak only to what happened before treatment, and parallel trends may well hold ex-ante and fail ex-post. Similarly, balance on observed covariates cannot rule out \emph{unobserved} confounders that induce violations of parallel trends. And in many settings, pre-trends do differ across groups, and observed covariates are imbalanced. Yet researchers would still like to learn something about the ATT, even if there is evidence that the parallel trends assumption does not hold exactly.

To this end, sensitivity analyses have been proposed that use observed pre-treatment deviations from parallel trends to bound the plausible magnitude of post-treatment violations (e.g., \citealp{rambachan2023more}). This framework has been widely adopted, and is now recommended as a standard sensitivity analysis tool in the DiD literature \citep{roth2023s, baker2025difference}. While a step forward from naively assuming that parallel trends holds exactly, this approach still suffers from shortcomings. In applied work, violations of parallel trends are typically attributed to \emph{unobserved confounders} that \emph{jointly} shape the \emph{selection} of units \emph{into treatment} and their untreated \emph{outcome trends} \citep{callaway2021difference, roth2023s, baker2025difference}. Yet pre-trend extrapolation offers no framework for reasoning about these forces, and \emph{why} parallel trends might fail in the first place. It is informative only insofar as pre-treatment dynamics are a reliable proxy for the post-treatment counterfactual, an assumption that it cannot itself adjudicate. One immediate symptom of this deficiency is that the framework cannot be applied to the simplest, canonical two-period DiD design, where pre-trends are not available. How can we translate such confounding into formal statements about deviations from parallel trends and the ATT bias? How strong would it have to be to overturn a given conclusion in a DiD study? 

In this paper, we develop a suite of sensitivity analysis tools for DiD that allows one to easily answer such questions.  
We first derive a novel characterization of the omitted variable bias (OVB) formula for the ATT. While motivated by DiD, this result is important in its own right and of independent interest---we elaborate on this point in the related literature below. We show how the bias due to violations of parallel trends admits a formal decomposition into the product of a scaling factor estimable from the data, and three interpretable bias factors that researchers already informally invoke: (i)~the strength of unobserved confounders in shaping selection into treatment, (ii)~their strength in shaping the evolution of untreated potential outcomes, and (iii)~the alignment between these two channels. Among these, we further discuss how selection into treatment is the most important factor for bounding this bias (since the other two factors are each upper bounded by one), and we offer a rich set of equivalent ways that practitioners can use to reason about it, such as how unmeasured confounding increases the average odds of treatment among the treated, widens the covariate imbalance between treated and control units, or explains variation in treatment odds among the untreated. 

Next, we derive (extreme) robustness values for DiD. These are sensitivity statistics that characterize the minimum strength of unobserved confounding needed to overturn a given conclusion about the ATT \citep{cinelli2020making,cinelli2025omitted,chernozhukov2022long}. Routine reporting of these quantities, alongside point estimates and standard errors, provides a quick and simple way to communicate how robust DiD findings are to violations of the parallel trends assumption. We also offer formal bounds on the ATT  based on plausibility judgments about how the strength of unobserved confounders compares against the explanatory power of observed covariates. These bounds connect naturally to balance checks already performed by applied researchers, who compare observed covariates across treated and control groups as informal evidence for parallel trends. Our framework gives these comparisons a direct role in sensitivity analysis, allowing the observed imbalance on measured covariates to discipline judgments about the unobserved imbalance one is willing to entertain. When multiple pre-treatment periods are available, the same logic applies to pre-treatment dynamics, allowing our framework to be used in tandem with current pre-trend extrapolation approaches. 

Finally, we provide flexible and efficient methods for statistical inference using debiased machine learning (DML), which allows the use of modern machine learning algorithms for estimation---though we note our approach can also be used with standard parametric regressions. We demonstrate the utility of our approach in an empirical application revisiting the analysis of the effect of minimum wage on teen employment in \citet{callaway2021difference}. Open-source software for R implements the methods discussed in this paper.

\paragraph{Related literature.}

Our work contributes to the growing literature on sensitivity analysis for DiD, where the dominant approach is the pre-trend extrapolation framework of \cite{rambachan2023more}. Pre-trend extrapolation is best understood as a way of benchmarking the plausible magnitude of post-treatment violations against observed pre-treatment dynamics. However, it cannot explain how such violations arise in the first place. Our framework fills this gap by decomposing the bias into interpretable factors tied to selection into treatment and untreated outcome evolution, and in doing so, it opens the door to a broader set of benchmarking strategies. Observed covariates, for instance, can also be used to benchmark the plausibility of unobserved confounding, further connecting sensitivity analysis to the balance tests that applied researchers already perform as a matter of course. See Section~\ref{sec:benchmarking} for further discussion.

Our work is also related to the broader literature on sensitivity analysis for the ATT. In particular, our results are most closely related to \citet{chernozhukov2022long}, who develop a general OVB framework for linear functionals of the conditional expectation of the outcome, encompassing the ATT as a special case; see \cite{bach2025sensitivity} for a direct application of \citet{chernozhukov2022long} to DiD designs using this unconditional parameterization. Though we build on \citet{chernozhukov2022long}, our OVB formula for the ATT differs from theirs because we exploit a structural feature of the ATT---namely, that the bias admits a parameterization conditional solely on the untreated units. This parameterization has several desirable properties, such as requiring plausibility judgments only on the local parameters that directly drive the bias. We further extend \citet{chernozhukov2022long} by providing alternative ways of quantifying selection strength in terms of treatment odds and covariate imbalance, both for our preferred conditional parameterization and for the original unconditional parameterization. 

Our results are also closely related to the variance-based sensitivity analysis proposed in \citet{huang2025variance} for weighting estimators, which bounds the bias of the ATT by restricting an $R^2$ measure of the balancing weights. Our approach differs from theirs in four main ways. First, we use a non-centered parameterization of selection strength, which naturally avoids a zero-denominator issue acknowledged in \citet{huang2025variance}. Second, we further decompose the bias into three components---the strength of confounding in treatment selection, its strength in untreated outcome trends, and the alignment between the two. We show that \citet{huang2025variance}'s parameterization bundles the trend and alignment components together, so that their upper bound on the bias implies setting these two components to one. Third, their benchmarking analysis uses observed covariates in a way that is less conservative than ours. Fourth, we use debiased machine learning for statistical inference.
We further extend the derivations in \citet{huang2025variance} to obtain results analogous to ours in their centered parameterization. 

Thus, a distinct contribution of this paper is not only a novel OVB formula for the ATT, but a complete account of the different parameterization choices available for this bias, along with the theoretical and practical consequences of these choices. See Section~\ref{sec:parameterization} and Appendix~\ref{app:all-comparison} for further details.

\paragraph{Notation.} $R^2_{Y\sim X}$ denotes the (possibly uncentered) $R^2$
from the orthogonal linear projection of a scalar $Y$ onto a random vector
$X$, and $\eta^2_{Y\sim X}$ its non-parametric analogue, with $E[Y\mid X]$ in
place of the best linear predictor. The partial $R^2$ of $Y$ with $U$ given
$X$ is defined as $R^2_{Y\sim U\mid X}:=(R^2_{Y\sim U+X}-R^2_{Y\sim X})/(1-R^2_{Y\sim
X})$, with an analogous definition for $\eta^2_{Y\sim U\mid X}$. We use $\dd L/\dd P$ to denote the Radon--Nikodym
derivative of $L$ with respect to $P$; for distributions $P \ll Q$,
$\chi^2(P\|Q):=\int(\dd P/\dd Q-1)^2\,\dd Q$ denotes the chi-squared divergence of
$P$ from $Q$. Finally, $P_{V\mid W=w}$ denotes the distribution of $V$ given
$W=w$, and we use $E_n[f(X_i)]:=\frac{1}{n}\sum_{i=1}^n f(X_i)$ for empirical
expectations.

\section{Background and running example}
\label{sec:background}

In this section we introduce the running example used throughout the paper. This example serves three main purposes: (i)~it illustrates the standard use of pre-trends and covariates as a tool for assessing the plausibility of parallel trends; (ii)~it shows how our framework already refines this standard analysis through a richer characterization of \emph{observed} biases; and (iii)~it sets the stage for the formal sensitivity analysis to \emph{unobserved} confounders that we develop in the rest of the paper. 

\subsection{The effect of minimum wage on teen employment}

Whether increases in minimum wages reduce teen employment has been actively debated in labor economics for decades, with credible evidence on both sides \citep{neumark2022myth, dube2024minimum}. Here we revisit the analysis of \citet{callaway2021difference}, who study this question using a difference-in-differences design.   

Between 2001 and 2007, the federal minimum wage remained at \$5.15, but several states raised their own minimums above this floor at different points in time. \citet{callaway2021difference} exploit this variation in treatment timing to estimate group-time average treatment effects for each cohort of states, defined by when they first raised their minimum wage above the federal level.  To fix ideas, here we focus on the last cohort, the states that raised their minimum wage in 2007. The treated group consists of the 584 counties in those states, and the control group consists of the 1,377 counties in states that never raised their minimum wage above the federal floor before the end of 2007.

An investigator interested in estimating the ATT might start the analysis by positing that, absent the policy change, teen employment in treated and control counties would have evolved in parallel over time. Under this parallel trends assumption, the ATT is identified by the usual canonical  two-by-two DiD estimand. The results are shown in panel (a) of Figure~\ref{fig:nt_ATT_multi_main}. 
The treatment effect of interest is the estimate for 2007, and we find that increases in minimum wages resulted in a 2.77\% reduction in teen employment for that group. This is a sizable, policy-relevant negative effect, but how credible is this estimate? 

\begin{figure}[t]
    \centering
    \begin{subfigure}{0.49\textwidth}
        \centering
        \includegraphics[width=\textwidth]{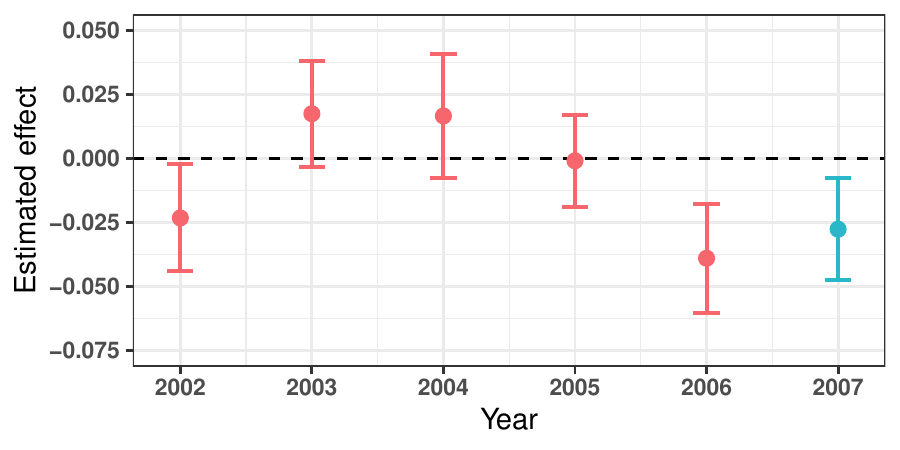}
        \caption{{\scriptsize Unconditional Parallel Trends, never-treated as control}}
    \end{subfigure}%
    \hspace{0.015\textwidth}%
    \begin{subfigure}{0.49\textwidth}
        \centering
        \includegraphics[width=\textwidth]{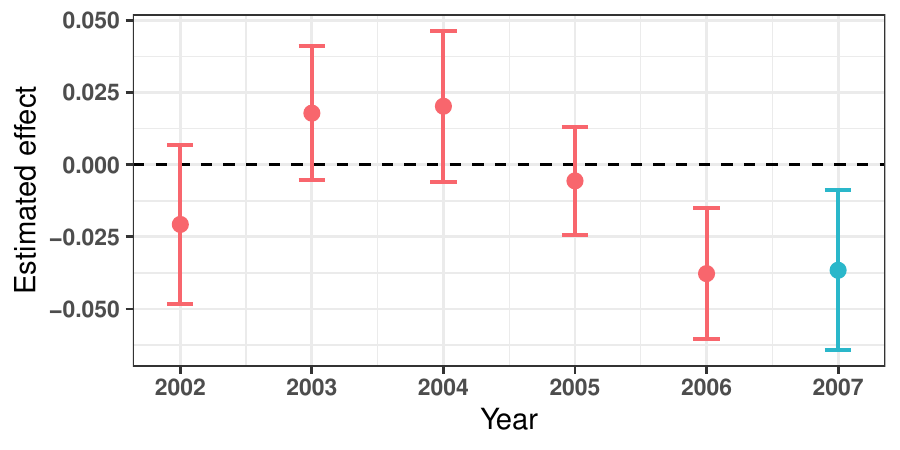}
        \caption{{\scriptsize Conditional Parallel Trends, never-treated as control}}
    \end{subfigure}
    \caption{{\footnotesize 
     Effect of the minimum wage on teen employment. Estimates use DML with a canonical DiD design and the preceding year as base period. Red and blue lines give point estimates with uniform 95\% confidence bands for pre-treatment periods and treatment effects, respectively.}}
    \label{fig:nt_ATT_multi_main}
\end{figure}

Beyond the effect of interest, Figure~\ref{fig:nt_ATT_multi_main}(a) also shows pre-treatment estimates (``pre-trends'') from 2002 to 2006, which are now commonly used to assess the plausibility of the parallel trends assumption. If the untreated trends of both groups were indeed evolving in parallel over time, these pre-trends should be equal to zero, barring sampling variation. 
Yet, the 2006 pre-trend is  -3.9\%, which is not only large and statistically significant, but also of the same order of magnitude as the treatment effect estimate itself. This suggests that treated and control counties were already on diverging trends before any policy changes.

What might explain such deviations from parallel trends? A natural concern is that employment in counties with \emph{different characteristics} may not evolve in parallel over time. To probe this, another common approach is to check whether the treated and control groups are \emph{balanced} in terms of covariates that could plausibly affect employment trajectories. Table~\ref{tab:balance_main} in Appendix~\ref{app:additional-results} reports standardized differences in means for several such characteristics. We find that treated counties are more populous, whiter, more educated and have lower poverty rates than control counties; they are also concentrated in the Midwest and West, with the South substantially underrepresented. Most differences exceed the 0.25 threshold, conventionally regarded as problematic \citep{imbens2015causal,baker2025difference}. These are the kind of imbalances that could account for parallel trends violations, and they motivate adjusting for observed covariates before drawing any conclusions about the policy effect.

To alleviate such concerns, panel (b) of Figure~\ref{fig:nt_ATT_multi_main} reports DiD estimates under a \emph{conditional} parallel trends assumption, which requires only that counties with the \emph{same observed characteristics} would have evolved in parallel over time.  These estimates account for observed covariates nonparametrically, using random forests and debiased machine learning as we will discuss in Section~\ref{sec:est-and-inf}. The treatment effect estimate is now -3.66\%, somewhat larger in magnitude than the unconditional estimate, and points to the same conclusions. However, the 2006 pre-trend estimate also remains essentially unchanged. In other words, though adjusting for observed covariates has not changed our treatment effect estimate, it has not resolved the pre-trend problem either. This persistent pre-trend deviation suggests confounding is acting through factors not captured by the observed data. 

At this point, the standard recommendation in the DiD literature is to proceed with a sensitivity analysis of the type of \citet{rambachan2023more}. It proceeds as follows: take the conditional 2006 deviation of -3.8\%,
assume the same deviation persists in 2007, and subtract it from the conditional estimate of -3.66\%. Since 3.8\% is larger than 3.66\%, this is evidently enough bias to overturn our original conclusion. This approach is straightforward and an improvement over simply assuming parallel trends and ignoring the problem. But it is also somewhat unsatisfactory as it leaves several questions unanswered. Why have the observed covariates barely changed these estimates, despite their substantial imbalance? Should we be looking at standardized mean differences or something else? What kind of unobserved confounding would have to be at work to produce these biases? And is pre-trend extrapolation the only way to gauge the plausible magnitude of such confounders? 

As a preview of how the results of this paper can help analysts answer these questions,
we show that the difference between the conditional and unconditional estimates in fact admits the following exact decomposition:
\[
\underbrace{(-0.0277)}_{\text{Unconditional Estimate}} - \underbrace{(-0.0366)}_{\text{Conditional Estimate}} =~ \underbrace{0.0089}_{\text{Bias}} ~= \underbrace{0.194}_{\text{Alignment}} \times \underbrace{0.117}_{\text{Trend}} \times \underbrace{2.548}_{\text{Selection}} \times \underbrace{0.154}_{\text{Scale}}.
\vspace{.2cm}
\]
Table~\ref{tab:obs_strength} in Appendix~\ref{app:additional-results} reproduces the above decomposition using all covariates, and also presents analogous decompositions obtained by adjusting for each observed covariate separately. Almost every applied DiD paper conditions on covariates, yet the standard narrative about why they matter rarely extends beyond covariate imbalance.  Decompositions such as the one above provide an exact post-mortem of how covariates affect (or not) the results. 

First, the ``selection'' component confirms that indeed there is a sizable distributional imbalance between treated and control groups, corresponding to a chi-squared divergence of $2.548^2 \approx 6.5$---note the relevant metric of distributional imbalance here is not the commonly used standardized mean difference, though in some cases the two may be close to each other, as we will discuss. Second, one of the reasons why this imbalance does not translate into a substantial difference in estimates in our running example is because these covariates explain only $0.117^2\approx 1.4\%$ of the variation in trends among the control group (the ``trend'' component). Third, for bias to arise, covariates must affect both channels jointly---their effect on the outcome trend must correlate with their effect on selection. This correlation, the ``alignment'' component, equals only 0.194 in our example, further attenuating the bias. As with the conditional estimate, all these components can also be estimated non-parametrically using debiased machine learning with the tools we provide in this paper. We argue that exercises such as the one above should be a routine part of DiD studies.

But more importantly, this same decomposition used to explain how observed covariates \emph{did} change our estimates can also be used to contemplate how unobserved confounders \emph{would} have changed our estimates. For example, it explains what unobserved confounders need to look like, in terms of their strength in predicting outcome trends and treatment selection, in order to induce any certain amount of bias---including biases derived from pre-trend extrapolation.
Crucially, it also opens up different ways of gauging the plausible magnitude of confounding. For instance, if an investigator has grounds to claim that no plausible unobserved confounder is more strongly imbalanced than the observed covariate region---or, equivalently, in terms of how it changes the odds of belonging to the treated group---this is also sufficient to bound the bias. We develop these tools formally in the rest of the paper, beginning with the bias decomposition itself in the next section.

\section{Omitted variable bias in difference-in-differences designs}
\label{sec:ovb}

\subsection{Problem setup}
\label{sec:setup}

We begin by introducing the omitted variable bias problem in the canonical two-by-two DiD setup. Let $t \in \{1,2\}$ index the pre- and post-treatment periods, respectively. Let $D\in \{0,1\}$ denote treatment group assignment, with $D=1$ for treated and $D=0$ for control. For each period $t$, $Y_t(1)$ and $Y_t(0)$ denote the treated and untreated potential outcomes at time $t$, respectively. We assume the observed outcome at time $t$ satisfies
\begin{equation}
\label{eq:consistency}
Y_t = D Y_t(1) + (1 - D) Y_t(0).    
\end{equation}
Equation~(\ref{eq:consistency}) is often referred to as the consistency assumption, or the stable unit treatment value assumption (SUTVA). 

The causal parameter of interest is the average treatment effect on the treated (ATT) at the post-treatment period, defined~as 
\begin{equation}
\label{eq:att}
\text{ATT} := E[Y_2(1) \mid D=1] - E[Y_2(0) \mid D=1].    
\end{equation}
Now let $\Delta Y:= Y_2 - Y_1$ denote the outcome evolution, $\Delta Y(0) := Y_2(0) - Y_1(0)$ the untreated potential outcome evolution, $X$ denote a vector of observed covariates and $U$ denote a vector of \emph{unobserved} covariates. Suppose that, conditionally on both $X$ and $U$, the outcomes of treated and control units would have evolved in parallel over time in the absence of treatment,~i.e.,
\begin{equation}
\label{eq:cpta}
E[\Delta Y(0) \mid D=1, X, U] = E[\Delta Y(0) \mid D=0, X, U].
\end{equation}
Under this conditional parallel trends assumption---along with the assumptions of no anticipation and other regularity conditions\footnote{See Appendix~\ref{app:did-assump} to recall the usual standard DiD assumptions.}, such as weak overlap below---the ATT is identified by the following well-known \emph{difference-in-differences estimand},
\begin{equation}
\label{eq:long}
\theta:= E[~E[\Delta Y \mid D=1, X,U] - E[\Delta Y \mid D=0,X,U] ~\mid~ D=1~].
\end{equation}
We refer to $\theta$ as the ``long'' DiD estimand because it depends both on the observed covariates $X$ and the unobserved confounders $U$. The starting point of our analysis is that the investigator has determined that she is interested in the DiD estimand  $\theta$ as in (\ref{eq:long}). 

\begin{remark}
Different formulations of the DiD assumptions can lead to the same statistical estimand (\ref{eq:long}), after appropriately redefining assumptions and variables. See, for example, \cite{callaway2021difference} for identification results involving group-time average treatment effects with multiple time periods, which invoke various different choices of control groups, parallel trends and no-anticipation assumptions.
Another example is the identification of the ATT under an unconfoundedness assumption, instead of parallel trends---for that case, one simply replaces $\Delta Y$ with $Y$.  
Our results apply to any such estimand that can be expressed in the generic form (\ref{eq:long}). 
\end{remark}

We further impose the following regularity conditions.

\begin{assumption}[Regularity Conditions]
\label{asmp:regularity}
We assume the outcome evolution is square-integrable $E[\Delta Y^2] < \infty$ and weak overlap  $E\left[P(D=1)^{-1}P(D=0|X,U)^{-1} \mid D=1 \right] < \infty$.
\end{assumption}

In practice, however, the unobserved confounders $U$ are not available to the researcher, so we cannot estimate the ``long" DiD parameter $\theta$. Instead, we are forced to estimate its ``short'' version, omitting $U$ from the analysis,
\begin{equation}
\label{eq:short}
\theta_s:= E[~E[\Delta Y \mid D=1, X] - E[\Delta Y \mid D=0,X] ~\mid~ D=1~].    
\end{equation}
Note that the parallel trends assumption may not hold if we condition on $X$ alone, and our estimates may thus suffer from omitted variable bias. In fact, under standard DiD assumptions, the omitted variable bias due to the omission of $U$ exactly equals the average deviation from parallel trends, among the treated, from conditioning on $X$ alone, i.e.,
\[\theta - \theta_{s}
= -
E\left[~E[\Delta Y(0) \mid X, D=1] - E[\Delta Y(0)  \mid X, D=0] ~\mid~ \, D=1 ~\right].\]
Consequently, judgments about deviations from parallel trends can be translated into judgments about omitted confounding. Conversely, the presence of omitted confounders provides a natural way to explain deviations from parallel trends.

Our goal is to characterize the omitted variable bias---or, equivalently, the deviations from parallel trends---in terms of the strength of the omitted variable $U$, capturing the distinct channels through which $U$ confounds our estimand. 

\subsection{The omitted variable bias formula}

Let $g_0$ and $g_{0s}$ be the long and short regression functions of outcome trend for the untreated,
\[g_0:= E[\Delta Y|D=0, X,U], \qquad g_{0s}:= E[\Delta Y|D=0, X].\]
The long and short DiD estimands can thus be compactly written as
\[\theta = E[\Delta Y -g_0\mid D=1], \qquad \theta_s = E[\Delta Y - g_{0s}\mid D=1].\]
The average treated trend $E[\Delta Y|D=1 ]$ is the same in both the long and short parameters, and it cancels out from the bias. The OVB is therefore fully governed by errors in the extrapolation of the untreated trend from the control group to the treated group. That is,
\[\theta - \theta_s = -(\theta_0 - \theta_{0s}),\]
where,
\begin{equation}
\label{eq:theta0}
\theta_0 :=E[g_0|D=1], \qquad \text{and}, \qquad \theta_{0s}:=E[g_{0s}|D=1].    
\end{equation}
In order to better understand the nature of this extrapolation, consider the marginal odds, as well as the short and long conditional odds of treatment,
\[
O:=\frac{P(D=1)}{P(D=0)},
\qquad
O_{XU}:=\frac{P(D=1\mid X,U)}{P(D=0\mid X,U)},
\qquad
O_X:=\frac{P(D=1\mid X)}{P(D=0\mid X)}.
\]
A key step for our approach is the following lemma, which rewrites $\theta_0$ and $\theta_{0s}$ as inner products of the conditional regression functions and rebalancing weights given by the conditional to marginal odds ratio, under the control population. 

\begin{Lemma}
\label{lemma:cond_RR_main} 
Under Assumption~\ref{asmp:regularity}, $\theta_0$ and $\theta_{0s}$ in (\ref{eq:theta0}) can be written as
\[\theta_0 = E 
\left[g_0 \left(\frac{O_{XU}}{O}\right)\mid D=0 \right], \quad \theta_{0s} = E\left[g_{0s} \left(\frac{O_{X}}{O}\right)\mid D=0\right].\]
Moreover, $g_{0s} = E[g_0 \mid D=0,X]$ and $O_{X} = E[O_{XU}\mid D=0,X]$.
\end{Lemma}
The fact that $g_{0s} = E[g_0 \mid D=0,X]$ follows directly from the law of iterated expectations. The property $O_{X} = E[O_{XU}\mid D=0,X]$ can be verified by applying Bayes' rule. Using this lemma, we can derive the following characterization of the OVB for the ATT. It shows that deviations from parallel trends can be exactly characterized by the ability of omitted confounders to explain treatment selection---as parameterized by the odds of treatment---and trend variation, jointly.
\begin{Theorem} [OVB for the ATT] 
\label{thm:main_npm}
Consider the long and short estimands $\theta$ and $\theta_{s}$ in (\ref{eq:long}) and (\ref{eq:short}). Under the regularity conditions of Assumption~\ref{asmp:regularity}, the OVB is given by the (negative of the) covariance of errors induced~by~the omission of $U$ both in the regression function and in the odds ratio, among the untreated units, 
\begin{align*}
\theta - \theta_s = - \operatorname{Cov}\left(g_0 - g_{0s},~ \frac{O_{XU}}{O}- \frac{O_{X}}{O} ~\middle|~ D=0 \right).
\end{align*}
Moreover, this bias can be further expressed in terms of $R^2$ measures:
\[\theta - \theta_s = -\underbrace{\rho_0}_{alignment} \underbrace{C_{0\Delta Y}}_{trend} \underbrace{C_{0D}}_{selection}
\underbrace{S_0}_{scale}, \]
where,
\[\rho_0 := \operatorname{Cor}\left(g_0 - g_{0s}, O_{XU} - O_X \mid D=0\right),\quad
C^2_{0\Delta Y} := \eta^2_{\Delta Y\sim U\mid X,D=0},\quad
C^2_{0D}:= \frac{1-R^2_{O_{XU}\sim O_{X}|D=0}}{R^2_{O_{XU}\sim O_{X}|D=0}},\]
and  
\[S^2_0 := \sigma^2_{0s} \nu^2_{0s}, \qquad 
\sigma^2_{0s}:= E[\operatorname{Var}(\Delta Y \mid X,D=0)\mid D=0], \qquad 
\nu^2_{0s}:= E\left[\left(\frac{O_X}{O}\right)^2 \middle| D=0\right].\]
\end{Theorem} 

Theorem~\ref{thm:main_npm} rewrites the average deviation from parallel trends in terms that more conveniently rely on scale-free partial $R^2$ measures of association characterizing the strength of unobserved confounding. The bias decomposes into the product of three non-identifiable bias factors, $\rho_0$, $C_{0\Delta Y}$, $C_{0D}$---which need to be restricted by plausibility judgments on the strength of confounders---and an identifiable scale factor, $S_0$, which is estimable from the observed data. Note that bias arises only when all bias factors are non-zero. 

The ``trend'' component $C^2_{0\Delta Y}$ measures the strength of confounding in explaining trend variation. Formally, this is measured by $\eta^2_{\Delta Y \sim U|X,D=0}$, which is the \emph{non-parametric} partial $R^2$ of $\Delta Y$ with $U$, given $X$, among the untreated. This quantifies how much residual variance of the trend in the control group is explained by the omitted confounders, after taking into account what is already explained by~$X$. The less confounders can explain trend variation in the control group, the less the potential of such confounders to induce deviations from parallel trends. Note this component can always be left unconstrained by the investigator, as it is naturally upper bounded by 1.

The ``selection'' component $C^2_{0D}$ measures the strength of confounding in explaining selection into treatment. Formally, this is driven by $1-R^2_{O_{XU}\sim O_{X}|D=0}$, i.e., the share of variation in the long treatment odds of the control group which cannot be explained by observed covariates alone. While this latter quantity is also an $R^2$ which can be upper bounded by 1, note the selection term $C^2_{0D}$ is given by the ratio, $(1-R^2_{O_{XU}\sim O_{X}|D=0})/R^2_{O_{XU}\sim O_{X}|D=0}$, which can be arbitrarily large. Therefore, the treatment selection component must always be restricted. To aid  plausibility judgments regarding $C^2_{0D}$, we provide alternative characterizations of treatment selection in the next section below.

Together, the two sensitivity parameters, $\eta^2_{\Delta Y \sim U|X,D=0}$ and $1-R^2_{O_{XU}\sim O_{X}|D=0}$, characterize the strength of confounding in explaining trends and treatment odds.  However, for bias to arise, it is not sufficient for confounders to create errors in the trend and selection components; these errors must be systematically aligned.  This is captured by the ``alignment'' component $\rho_0$, defined as the correlation between these two errors in the control arm.  Similarly to the trend component, in the worst case, the alignment component can also be left unconstrained by the investigator, as its magnitude is upper bounded by 1.

Finally, the identifiable scale factor $S_0$ translates these three scale-free measures of the strength of confounding back into the original units of the ATT. It is characterized by the strength of \emph{observed covariates} in explaining variation in outcome trend and selection into treatment. The first component, $\sigma_{0s}^2$, measures the leftover variation of the trend in the control group after taking into account the part explained by $X$. Note the more variation $X$ explains of the outcome trend, the less room there is for unobserved confounding to create bias. Moving to $\nu_{0s}^2$, we see that the opposite relationship holds with respect to the strength of covariates in explaining  treatment assignment. This term measures how well observed covariates separate treated from control units. The stronger this separation, the less identifying variation is being used to estimate the ATT, and the more leverage unobserved confounding has to drive the bias; $\nu_{0s}^2$ is also directly interpretable as the imbalance of observed covariates between treated and control units, as measured by the chi-squared divergence, as we show next.

\subsection{Making sense of treatment selection strength} 
\label{sec:treatment-selection}

As we have seen, the OVB formula for the ATT is particularly sensitive to selection strength, both unobserved (via $C_{0D}^2$) and observed (via $\nu^2_{0s}$). It is therefore useful to develop alternative ways of understanding and reasoning about these quantities. In this section, we provide such characterizations in terms of increases in average treatment odds and covariate imbalance, helping researchers specify plausible upper bounds on the unobserved selection strength $C^2_{0D}$ using the scale most familiar to them in their own applications.

For example, in genetics, psychology, economics and other quantitative social sciences, researchers routinely think in terms of variance explained and related $R^2$ measures \citep{cohen2013statistical,Imbens03,oster_unobservable_2019,cinelli2020making,cinelli2025omitted,chernozhukov2022long}. In medicine and epidemiology, odds and odds ratios are commonly used to characterize selection into treatment \citep{DingVanderweele16}. Covariate imbalance---often summarized by the standardized mean difference, itself a common effect size, \citep{cohen2013statistical}---is a widely used balance diagnostic in observational studies and randomized trials \citep{stuart2010matching}. Each community has developed intuition for what counts as ``large'' or ``small'' on its own scale. Our goal is to formally connect these measures, allowing researchers to approach the OVB problem from different perspectives. The following lemma provides the basic identities linking treatment odds, changes of measure, and covariate imbalance.

\begin{Lemma}
\label{thm:main_connection}
Under the weak overlap condition of Assumption~\ref{asmp:regularity}, the following holds:
\begin{enumerate}[label=(\roman*), leftmargin=*, itemsep=2pt, topsep=3pt]
\item \textbf{Odds ratio as change of measure:}
\(\displaystyle O_{XU}/O
= \dd P_{X,U|D=1}/\dd P_{X,U|D=0}\).

\item \textbf{Control-to-treated odds identity:}
\(\displaystyle E[O_{XU}^2\mid X,D=0]
=O_XE[O_{XU}\mid X,D=1]\).

\item \textbf{Imbalance as average odds:}
\(\displaystyle \chi^2(P_{X,U|D=1}\|P_{X,U|D=0})
=E[O_{XU}/O\mid D=1]-1\).
\end{enumerate}
\end{Lemma}

Applying these identities yields equivalent expressions for the selection components $C^2_{0D}$ and $\nu^2_{0s}$ in terms of increase in average treatment odds and $\chi^2$-divergence.

\begin{Corollary}[Alternative characterizations of $C_{0D}^2$ and $\nu^2_{0s}$]
\label{thm:alt-selection}
Under the weak overlap condition of Assumption~\ref{asmp:regularity}, $C^2_{0D}$ and $\nu^2_{0s}$ of Theorem~\ref{thm:main_npm} admit the additional interpretations:
\begin{enumerate}
\item Increase in average treatment odds among the treated:
{
\begin{align*}
C_{0D}^2 = \frac{E[O_{XU}\mid D=1] - E[O_X\mid D=1]}{E[O_X\mid D=1]}, \qquad 
\nu^2_{0s} = \frac{E[O_X|D=1]}{O}.
\end{align*}
}

\item Covariate imbalance between treated and control units:
{
\begin{align*}
C_{0D}^2 = \frac{\chi^2(P_{X,U|D=1}\|P_{X,U|D=0}) - \chi^2(P_{X|D=1}\|P_{X|D=0})}{\chi^2(P_{X|D=1}\|P_{X|D=0}) + 1}, \qquad 
\nu^2_{0s}=\chi^2(P_{X|D=1}\|P_{X|D=0})+1.
\end{align*}
        
}
\end{enumerate}
\end{Corollary}

The odds characterization gives a direct way to translate substantive restrictions on the propensity score into bounds on $C_{0D}^2$. For example, if one believes that the inclusion of $U$ at most doubles the average odds of treatment, among treated units, then Corollary~\ref{thm:alt-selection} implies $C_{0D}^2 \le 1$. This relationship also clarifies the connection between the selection component $C^2_{0D}$ of the OVB formula and sensitivity models based on worst-case bounds on odds ratios. The marginal sensitivity model of \citet{tan2006distributional}, for instance, assumes
\[
\Lambda^{-1} \leq O_{XU}/O_X \leq \Lambda
\]
almost surely, for some $\Lambda >1$. Such a restriction implies the sharp bound (see Corollary~\ref{cor:msm-parameterization}),
\[
C_{0D}^2 \leq (\Lambda-1)^2/\Lambda.
\]
Thus, if one believes $U$ cannot double (or halve) the odds of treatment not only on average, but uniformly, this implies the tighter restriction $C^2_{0D} \leq 1/2$. This also makes it clear that worst-case restrictions on odds ratios are sufficient, but not necessary, to bound the bias.

Beyond treatment odds, Corollary~\ref{thm:alt-selection} also provides a covariate balance interpretation of $C_{0D}^2$. It shows that $C_{0D}^2$ captures the additional imbalance introduced by $U$, relative to the imbalance already captured by $X$. Checking covariate balance is standard practice in applied research, often through standardized mean differences (SMD) \citep{imbens2015causal,baker2025difference}. As the OVB formula shows, however, the imbalance measure directly relevant for quantifying the bias of the ATT is the $\chi^2$-divergence, which summarizes the full distributional difference between treated and control groups.

\subsection{On the choice of parameterization of the ATT bias}
\label{sec:parameterization}

Theorem~\ref{thm:main_npm} provides a novel parameterization of the omitted variable bias of the ATT. This choice of parameterization is not innocuous. It determines the quantities about which researchers must make plausibility judgments, as well as the allowable values that such sensitivity parameters can take. In this section, we present some key differences of our result from prior OVB formulas for the ATT derived in \citet{chernozhukov2022long} and \citet{huang2025variance}.  A complete analysis is deferred to Appendix~\ref{app:all-comparison}.

The result of \citet{chernozhukov2022long} applied to the ATT can be written as
\begin{align}
\theta-\theta_s
&=\rho\,C_{\Delta Y}\,C_D\,S,
\label{eq:lls}
\end{align}
where
\(
C_{\Delta Y}^2:=\eta^2_{\Delta Y\sim U\mid X,D},
C_D^2:=\frac{E[O_{XU}]-E[O_X]}{E[O_X]},
S^2:=E\!\left[\operatorname{Var}(\Delta Y\mid X,D)\right]\frac{E[O_X]}{P(D=1)^2},
\) and
\(
\rho:=
\operatorname{Cor}\left(
E[\Delta Y\mid D,X,U]-E[\Delta Y\mid D,X],\;
-(1-D)(O_{XU}-O_X)
\right).
\)

Note that the sensitivity parameters in (\ref{eq:lls}) are defined in terms of the full population,  even though Theorem~\ref{thm:main_npm} shows that the bias for the ATT can be localized solely on the untreated units. Therefore, the parameterization in (\ref{eq:lls}) asks users to make plausibility judgments on quantities that are not directly relevant for the bias, such as $\eta^2_{\Delta Y\sim U\mid X,D}$ instead of $\eta^2_{\Delta Y\sim U\mid X,D=0}$. 
This fact has consequences beyond interpretation. In particular, these unconditional trend and alignment bias factors obey the following restriction.

\begin{Proposition}
\label{prop:lls-restriction}
The trend and alignment components in equation~\eqref{eq:lls} satisfy
\[
\rho^2 C_{\Delta Y}^2
=
\frac{P(D=0)\,\sigma_{0s}^2}
{E[\operatorname{Var}(\Delta Y\mid X,D)]}\,
\rho_0^2\,C_{0\Delta Y}^2
\;\leq\;
\frac{P(D=0)\,\sigma_{0s}^2}
{E[\operatorname{Var}(\Delta Y\mid X,D)]}.
\]
\end{Proposition}

Thus, whenever the residual outcome variance among treated units is positive, the upper bound in Proposition~\ref{prop:lls-restriction} is strictly smaller than one, meaning that the trend and alignment components cannot vary freely. For example, applying the bounds $|\rho|\leq1$ and $C_{\Delta Y}\leq1$ directly to (\ref{eq:lls}) gives
\(
|\theta-\theta_s|\leq C_D\,S,
\)
whereas Proposition~\ref{prop:lls-restriction} implies the smaller bound
\(
|\theta-\theta_s|
\leq
\sqrt{
\frac{P(D=0)\,\sigma_{0s}^2}
{E[\operatorname{Var}(\Delta Y\mid X,D)]}
}
\,C_D\,S
=
C_{0D}\,S_0,
\)
which is exactly the bound implied by our OVB decomposition in Theorem~\ref{thm:main_npm}.

As for the characterization of \citet{huang2025variance}, in Appendix~\ref{app:hp-comparison} we show that it admits the following representation:
\begin{align}
\theta-\theta_s
&=-\rho_{w0}\,C_{w0D}\,S_{w0},
\label{eq:hp}
\end{align}
where
\(
\rho_{w0}:=\operatorname{Cor}(\Delta Y,O_{XU}-O_X\mid D=0),
C_{w0D}^2:=\frac{E[O_{XU}]-E[O_X]}{E[O_X]-O},
\)
and
\(
S_{w0}^2
:=
\operatorname{Var}(\Delta Y\mid D=0)\,
\frac{P(D=0)\{E[O_X]-O\}}{P(D=1)^2}.
\)  For sensitivity analysis, \citet{huang2025variance} propose bounding $|\rho_{w0}|$ by the estimable quantity $\bar\rho_{w0}:= \sqrt{
1-\operatorname{Cor}^2(O_X,\Delta Y\mid D=0)
}.$

The first noticeable difference with respect to our result in Theorem~\ref{thm:main_npm} concerns the selection components. If $X$ does not predict treatment, which includes the canonical 2x2 DiD setup without covariates, $C_{w0D}$ has a zero denominator, and the decomposition (\ref{eq:hp}) is not well defined. If $X$ only weakly predicts treatment, (\ref{eq:hp}) is well defined, but $C_{w0D}$ can grow arbitrarily large, while the estimation of $S_{w0}^2$ becomes difficult.

The second difference concerns the trend and alignment components. The parameterization in equation~\eqref{eq:hp} does not contain a separate sensitivity parameter measuring the strength of confounding with the untreated outcome trend. Instead, trend and alignment components are bundled into a single correlation $\rho_{w0}$, and separate plausibility judgments about these parameters are not possible. Furthermore, the parameterization obeys the following restriction.

\begin{Proposition} 
\label{prop:hp-correlation}
When well defined, $\rho_{w0}$ in equation~\eqref{eq:hp} satisfies

\[
|\rho_{w0}|
\leq
\sqrt{
\frac{\sigma_{0s}^2}
{\operatorname{Var}(\Delta Y\mid D=0)}
}
\leq
\bar\rho_{w0}.
\]
The first equality holds iff $\rho_0^2C_{0\Delta Y}^2=1$, and the second iff $g_{0s}$ is an affine function of~$O_X$ among the untreated.
\end{Proposition}

Consequently, their reported upper bound is larger than necessary. When expressed in our parameterization, it can be attained only when the trend and alignment components reach their worst-case values and, among untreated units, the conditional outcome trend can be expressed as an affine function of the treatment odds. 

The consequences of these different parameterizations are especially transparent under the marginal sensitivity model (MSM) of \citet{tan2006distributional} introduced in Section~\ref{sec:treatment-selection}. The same restriction on treatment odds yields the following sharp bounds.

\begin{Corollary}
\label{cor:msm-parameterization}
Under the MSM, the selection parameters satisfy
\[
C_D^2\leq\left(\frac{E[O_X]-p}{E[O_X]}\right)\frac{(\Lambda-1)^2}{\Lambda},\qquad C_{0D}^2\leq\frac{(\Lambda-1)^2}{\Lambda},\qquad C_{w0D}^2\leq\left(\frac{E[O_X]-p}{E[O_X]-O}\right)\frac{(\Lambda-1)^2}{\Lambda}.
\]
All three bounds are sharp, with the last applying when $C_{w0D}^2$ is well defined.
\end{Corollary}

Thus, under the same MSM, the sharp bound for $C_{0D}^2$ is invariant to the observed treatment assignment, whereas the sharp bounds for the alternative parameterizations depend on it, with that for $C_{w0D}^2$ diverging and becoming uninformative as observed imbalance vanishes.

\section{Benchmarking the strength of unobserved confounding} 
\label{sec:benchmarking}

The main difficulty in sensitivity analysis is specifying plausible strengths for unobserved confounding. While in some cases researchers may not be able to  make absolute judgments about the magnitude of an omitted confounder, they may still have grounds to make judgments of its relative importance. For example, a researcher may be willing to argue that any unobserved confounder is no stronger than some observed covariate, or that post-treatment confounding is no stronger than the confounding revealed by pre-treatment periods. In this section, we develop two benchmarking procedures that formalize these comparisons. 

\subsection{Comparing unobserved confounders against observed covariates}
\label{sec:bench-covariates}

Our first analysis compares the gains in explanatory power from unobserved confounders with those from key observed covariates. Let $X_j$ denote the benchmark covariates against which we wish to calibrate the strength of $U$, and let $X_{-j}$ denote the remaining covariates, so that $X = (X_j, X_{-j})$. We define the following relative strength parameters:
\begin{align*}
k_{0\Delta Y,j} := \frac{ \eta^2_{\Delta Y \sim U,X_j \mid X_{-j},D=0} -  \eta^2_{\Delta Y \sim X_j \mid X_{-j},D=0}}{\eta^2_{\Delta Y \sim X_j \mid X_{-j},D=0}},
\quad 
k_{0D,j} := \frac{R^2_{O_{X} \sim O_{X_{-j}}\mid D=0} - R^2_{O_{XU} \sim O_{X_{-j}}\mid D=0}}{1 - R^2_{O_X \sim O_{X_{-j}}\mid D=0}}.
\end{align*}
The parameter $k_{0\Delta Y, j}$ measures how much additional variation in the untreated outcome evolution is explained by $U$, relative to the observed gains in variation explained by $X_j$. For example,  $k_{0\Delta Y, j} =1$ implies that further adding $U$ to the trend regression would lead to similar additive gains in explanatory power to those observed by adding $X_j$. The parameter $k_{0D, j}$ uses the same logic to measure the relative gain in explanatory power with the treatment odds. As discussed in Section~\ref{sec:treatment-selection}, $k_{0D, j}$ admits alternative interpretations in terms of the relative increase in the \textit{average treatment odds} or the additional \textit{imbalance between treated and control groups} attributable to $U$, relative to that attributable to $X_j$.

These measures allow us to re-express the absolute strength of $U$ in the OVB formula of Theorem~\ref{thm:main_npm} in terms of its relative strength as compared to that of $X_j$:
\begin{align*}
\eta^2_{\Delta Y \sim U \mid X,D=0} = k_{0\Delta Y,j} 
\left(\frac{\eta^2_{\Delta Y \sim X_j \mid X_{-j},D=0}}{1 - \eta^2_{\Delta Y \sim X_j \mid X_{-j},D=0}}\right), 
\quad 
1 - R^2_{O_{XU} \sim O_X\mid D=0} &= k_{0D,j}
\left(\frac{1 - R^2_{O_X \sim O_{X_{-j}}\mid D=0}}{R^2_{O_X \sim O_{X_{-j}}\mid D=0}}\right).
\end{align*}
Therefore, plausibility judgments on the \emph{relative importance} of $U$ compared to $X_j$, both in explaining outcome evolution, and treatment odds, can be leveraged to bound the bias. 

Finally, one may benchmark plausible values for the correlation of errors $\rho_0$. Notice that $\rho_0$ is not a measure of explanatory power of the confounders. Rather, it measures how systematically $U$ shifts the odds of treatment and outcome evolution in the same direction. This is partly connected to the functional form of confounding. For example, if $U$ enters the treatment and trend equations with different functional forms, this will generally attenuate the bias. To calibrate this empirically, we may use as a reference the observed correlation between the errors in the outcome regression and treatment odds induced by $X_j$, as measured by
\[
\rho_{0,j}:=\operatorname{Cor}\!\left(g_{0s} - g_{0s,-j},\; O_X - O_{X_{-j}} \,\big|\, D=0\right).
\]
Benchmarking the bias factors $\rho_0$, $C_{0\Delta Y}$ and $C_{0D}$ separately allows researchers to entertain richer confounding scenarios. For example, one may posit that an unobserved confounder is comparable to $X_j$ in terms of explaining treatment selection and outcome evolution, while remaining agnostic about alignment. Further details can be found in Appendix~\ref{app:bench-covariates}.

\subsection{Comparing post-treatment bias against pre-treatment bias}
\label{sec:bench-pre-trend}

We now connect the OVB framework to pre-trend extrapolation approaches such as \citet{rambachan2023more}. Existing methods typically bound the post-treatment bias by extrapolating deviations from parallel trends observed before treatment. Our decomposition shows that these extrapolation assumptions can instead be interpreted as restrictions on the underlying confounding mechanism itself, and also suggests natural alternative ways of both interpreting and augmenting them. We consider the simplest case of one additional pre-treatment period, and leave extensions to multiple time periods to future work.

\subsubsection*{Quantifying pre-treatment bias}

Consider three time periods $t \in \{0, 1, 2\}$, with treatment occurring after
$t = 1$. To estimate the (placebo) effect in the pre-treatment period, we take
$t=0$ as the reference and estimate the effect at $t=1$. Let
$\Delta Y^{\text{pre}} := Y_1 - Y_0$ denote the observed outcome evolution prior
to treatment, and let $(X^{\text{pre}}, U^{\text{pre}})$ denote the covariates
that would make parallel trends hold in the pre-treatment period. Throughout, we
write $\theta^{\text{pre}}$, $\theta_s^{\text{pre}}\neq 0$, $\rho_0^{\text{pre}}$,
$C^{\text{pre}}_{0\Delta Y}$, $C^{\text{pre}}_{0D}$, and
$S_0^{\text{pre}} = \sigma_{0s}^{\text{pre}}\nu_{0s}^{\text{pre}}$ for the
estimands, bias factors, and scale factor of Section~\ref{sec:ovb}, computed with
$(\Delta Y^{\text{pre}}, X^{\text{pre}}, U^{\text{pre}})$ in place of
$(\Delta Y, X, U)$, with $D$ unchanged throughout.

Since no treatment has yet occurred, $\theta^{\text{pre}} = 0$, and the
\emph{pre-treatment bias} from omitting confounders reduces to
$\theta_s^{\text{pre}}$, the (average) deviation from parallel trends before
treatment. Applying Theorem~\ref{thm:main_npm} to the pre-treatment period gives
\[
|\rho^{\text{pre}}_0|\, C^{\text{pre}}_{0\Delta Y}\, C^{\text{pre}}_{0D}\,
S_0^{\text{pre}} = |\theta_s^{\text{pre}}|,
\]
so $\theta_s^{\text{pre}}$ does not identify the pre-treatment bias factors
individually, but constrains their product. We allow $(X^{\text{pre}},
U^{\text{pre}})$ to differ from $(X,U)$, but in many cases, such as in our
empirical application, they coincide. When that happens, the treatment selection component is then immediately transportable across periods, since
$C^{2,\text{pre}}_{0D}=C^{2}_{0D}$ and $\nu_{0s}^{2,\text{pre}} = \nu_{0s}^{2}$.

Therefore, two natural strategies emerge for extrapolating the pre-treatment bias to the post-treatment period in this setting. One is to directly transport the bias magnitude, as is currently done in pre-trend extrapolation approaches; another is to transport only (some of) the bias factors, while re-estimating the scaling factor $S_0$ from post-treatment data. 

\subsubsection*{Extrapolating the bias magnitude}

As in \cite{rambachan2023more}, researchers may be willing to assume that the post-treatment bias is no larger than $k$ times the bias observed in the pre-treatment period. This corresponds to the assumption that $|\theta - \theta_s| \leq k|\theta_s^{\text{pre}}|$, and yields the following bounds
\begin{align*}
    \theta_{\pm} = \theta_s \pm k |\theta_s^{\text{pre}}|,
\end{align*}
where both $\theta_s$ and $\theta_s^{\text{pre}}$ are estimable from the data. Mapping this to unobserved confounders, this approach is equivalent to imposing a constraint on the full product of bias factors and the scale factor, i.e., 
\[
|\rho_0|\, C_{0\Delta Y}\, C_{0D}\, S_0 \leq k |\rho^{\text{pre}}_0| C^{\text{pre}}_{0\Delta Y} C_{0D}^{\text{pre}}\, S_0^{\text{pre}} = k |\theta_s^{\text{pre}}|.
\]

This connection also clarifies how different values of $k$ could arise; for example, consider $|\theta-\theta_s|=k|\theta_s^{\text{pre}}|$. Writing
$k_\rho := |\rho_0|/|\rho_0^{\text{pre}}|$,
$k_{\Delta Y} := C_{0\Delta Y}/C_{0\Delta Y}^{\text{pre}}$,
$k_D := C_{0D}/C_{0D}^{\text{pre}}$, and
$k_S := S_0/S_0^{\text{pre}}$, we then have $k = k_\rho k_{\Delta Y} k_D k_S$, so the multiplier $k$ collects the relative changes in each component of the bias across periods (and, in fact, $k_S$ is estimable).

\subsubsection*{Extrapolating the bias factors}

The previous considerations suggest another natural approach for extrapolating the bias. Rather than transporting the full magnitude of the pre-treatment bias, we can transport only the bias factors associated with the unobserved confounding mechanism, while allowing the observed scale component to change across periods. This yields the following bounds
\begin{align*}
\theta_{\pm} &= \theta_s \pm k \left(\frac{S_0}{S_0^{\text{pre}}}\right)  |\theta_s^{\text{pre}}| ,
\end{align*}
where $\theta_s$, $\theta_s^{\text{pre}}$, $S_0^{\text{pre}}$, and $S_0$ are estimable from the data, and $k = k_{\rho} k_{\Delta Y}  k_{D}$. 

Regardless of the choice of pre-trend extrapolation, the OVB decomposition of Section~\ref{sec:ovb}, together with benchmarking against observed covariates in Section~\ref{sec:bench-covariates}, allows researchers to clearly  translate these assumptions into statements about the underlying strength of confounding. For example, one can ask what kind of unobserved confounder would be required to generate the extrapolated post-treatment bias, in terms of its strength in treatment selection, outcome trends, and the correlation between the two, and whether such strength of confounding is comparable in magnitude to that of observed covariates.

\section{Estimation and inference} 
\label{sec:est-and-inf}

The previous sections characterize the OVB formula and the associated sensitivity parameters at the population level. We now discuss estimation and inference in finite samples. In what follows, we assume we have an i.i.d. sample from the observed data distribution. Appendix~\ref{app:simulations} reports extensive simulation exercises to verify the properties of our estimators and confidence intervals.

\subsection{Inference under direct restrictions on confounding}

Given plausibility judgments on the magnitude of sensitivity parameters $|\rho_0|$, $C_{0\Delta Y}$, $C_{0D}$, we have the following  bounds on $\theta$:
\begin{align}
    \theta_{\pm} &= \theta_s \pm |\rho_0|\,C_{0\Delta Y}\,C_{0D}\,S_0, \quad\text{with}\quad S_0^2 = \sigma_{0s}^2\,\nu_{0s}^2, \nonumber
\end{align}
where $\theta_s$ and $S_0^2$ are estimable from the observed data. 

We estimate these components using debiased machine learning (DML), which combines Neyman orthogonal scores with cross-fitting to enable the use of machine learning methods for nuisance estimation \citep{chernozhukov2018double,chernozhukov2022long}. Specifically, let $Z := (\Delta Y,D,X)$, $p:=P(D=1)$ and $\pi:=P(D=1\mid X)$.  Estimation proceeds via Algorithm~\ref{alg:DML}, with orthogonal scores
\begin{align}
\psi_{\theta_s}(Z; g_{0s}, \pi, p) &= \left(\frac{D}{p} - \frac{(1-D)}{(1-p)}\frac{O_{X}}{O}\right)(\Delta Y - g_{0s}) - \frac{D\theta_s}{p}, 
\label{eq:psi-theta} \\
\psi_{\sigma_{0s}^2}(Z;g_{0s},p) &= \frac{1-D}{1-p}\left((\Delta Y - g_{0s})^2 - \sigma_{0s}^2\right), \label{eq:psi-sigma} \\
\psi_{\nu_{0s}^2}(Z;\pi,p) 
&= 2\frac{D}{p}\left(\frac{O_{X}}{O} - \nu_{0s}^2\right) - \frac{1-D}{1-p}\left(\left(\frac{O_{X}}{O}\right)^2 - \nu_{0s}^2\right). 
\label{eq:psi-nu}
\end{align}
We denote the resulting estimates as
\begin{align*}
\widehat{\theta}_s := \text{DML}(\psi_{\theta_s}), \qquad
\widehat{\sigma}_{0s}^2 := \text{DML}(\psi_{\sigma_{0s}^2})\qquad
\widehat{\nu}_{0s}^2 := \text{DML}(\psi_{\nu_{0s}^2}).
\end{align*}

\begin{algorithm}[t]
\small
\caption{DML for estimable components: DML($\psi_{\beta}$)}
\begin{algorithmic}
\State \textbf{Input:} The Neyman orthogonal score $\psi_{\beta}(Z;\eta)$, sample $\{Z_i\}_{i=1}^n$, and number of folds $L$.
\State \textbf{Sample Splitting:} Randomly partition the sample into $L$ folds of approximately equal size. Denote each fold by $I_l$ and its complement by $I_l^c$, for $l = 1,\dots,L$.
    \For{$l = 1,\dots,L$}
        \State Estimate the nuisance parameters using observations in $I_l^c$. Denote the estimates by $\widehat{\eta}_l$.
    \EndFor
\State \textbf{Estimate:} Construct the estimator $\widehat{\beta}$ as the solution of $0 = \frac{1}{n}\sum_{l=1}^L\sum_{i \in I_l}\psi_{\widehat{\beta}}(Z_i;\widehat{\eta}_l)$.
\State \textbf{Return:} The estimate $\widehat{\beta}$ and the estimated scores $\psi_{\widehat{\beta}}(Z_i;\widehat{\eta}_l)$ for each $i \in I_l$ and each $l$.
\end{algorithmic}
\label{alg:DML}
\end{algorithm}

Let $\beta \in \{\theta_s, \sigma^2_{0s},\nu^2_{0s} \}$ denote a generic target parameter and $\eta = (g_{0s}, \pi, p)$ the vector of nuisance parameters. The scores (\ref{eq:psi-theta})-(\ref{eq:psi-nu}) admit the representation $\psi_{\beta}(Z;\eta) = \psi_\beta^a(Z;\eta)\,\beta + \psi_\beta^b(Z;\eta)$.  An application of the results in \citet{chernozhukov2018double} for linear score functions yields the following lemma, which establishes the asymptotic linearity and normality of the DML estimators, and characterizes their influence functions.

\begin{Lemma}[Asymptotic linearity of the DML estimators]
\label{lemma:Properties of the DML Bound Estimators}
Suppose that Assumptions 3.1 and 3.2 from \citet{chernozhukov2018double} hold for each of the scores $\psi_{\beta}$ above, and for the estimators of the nuisance parameters $\widehat{\eta}_l$. Then, the DML estimators $\widehat{\beta}$ are asymptotically linear and Gaussian with the following properties
    \begin{align}
        \sqrt{n}(\widehat{\beta} - \beta) &= \frac{1}{\sqrt{n}}\sum_{i=1}^n\underbrace{\frac{\psi_{\beta}(Z_i;\eta)}{-E[\psi_\beta^a(Z;\eta)]}}_{=: \varphi^0_{\beta}(Z_i)} + o_p(1) \overset{d}{\rightarrow} N(0,\sigma_{\varphi_{\beta}}^2), \text{ with } 
        \sigma_{\varphi_{\beta}}^2 := E[(\varphi^0_{\beta}(Z))^2], \nonumber 
    \end{align}
    where $\varphi^0_{\beta}(\cdot)$ is the influence function. Following Theorem 3.2 in \citet{chernozhukov2018double}, we denote the estimate of the influence function as $\widehat{\varphi}_{\widehat{\beta}}(\cdot) := -\psi_{\widehat{\beta}}(\cdot;\widehat{\eta}_l)\big/E_n[\psi^a(Z;\widehat{\eta}_l)]$ and use $\widehat{\sigma}_{\varphi_{\beta}}^2 := E_n\left[(\widehat{\varphi}_{\widehat{\beta}}(Z))^2\right]$ as the estimator for $\sigma^2_{\varphi_\beta}$.
\end{Lemma}

\begin{remark}
Note that the DML estimator for $\nu^2_{0s}$ can be used to estimate the chi-squared divergence of observed covariates, since $\nu^2_{0s} -1 = \chi^2(P_{X|D=1}\|P_{X|D=0})$---see Corollary~\ref{thm:alt-selection}.
\end{remark}

The plug-in estimator for the bounds thus takes the following form,
\begin{equation*}
\widehat{\theta}_{\pm} := \widehat{\theta}_s \pm |\rho_0|\,C_{0\Delta Y}\,C_{0D}\,\sqrt{\widehat{\sigma}_{0s}^2\,\widehat{\nu}_{0s}^2}.
\end{equation*}
Since $\widehat{\theta}_{\pm}$ is a smooth function of asymptotically linear estimators, we can obtain confidence intervals for the bounds by applying the delta method.

\begin{Theorem}[Confidence intervals for the bounds] \label{thm: inference}
    Under the Assumptions of Lemma \ref{lemma:Properties of the DML Bound Estimators}, the estimator $\widehat{\theta}_{\pm}$ is asymptotically linear and Gaussian with the following properties 
    \begin{align}
        \sqrt{n}(\widehat{\theta}_{\pm} - \theta_{\pm}) &= \frac{1}{\sqrt{n}}\sum_{i=1}^n\varphi^0_{\theta_{\pm}}(Z_i) + o_p(1) \overset{d}{\rightarrow} N(0, \sigma^2_{\varphi_{\theta_{\pm}}}), \text{ with } \sigma^2_{\varphi_{\theta_{\pm}}} = E[(\varphi^0_{\theta_{\pm}}(Z))^2], \nonumber \\
        \text{where } \varphi^0_{\theta_{\pm}}(Z) &= \varphi_{\theta_s}^0(Z) \pm |\rho_0|C_{0\Delta Y}C_{0D} \times \underbrace{\frac{1}{2S_0}\left(\sigma_{0s}^2\varphi^0_{\nu_{0s}^2}(Z) + \nu_{0s}^2\varphi^0_{\sigma_{0s}^2}(Z)\right)}_{=:\varphi_{S_0}^0(Z)}. \nonumber 
    \end{align}
    Let $\Phi$ denote the CDF of the standard normal and $\alpha$ denote the significance level. The confidence interval
\begin{align}
        [l, u] &= \left[\widehat{\theta}_{-} - \Phi^{-1}(1 - \alpha)\sqrt{\frac{E[(\varphi^0_{\theta_{-}}(Z))^2]}{n}}, ~\widehat{\theta}_{+} + \Phi^{-1}(1 - \alpha)\sqrt{\frac{E[(\varphi^0_{\theta_{+}}(Z))^2]}{n}}\right], \nonumber 
\end{align}
satisfies $P(\theta_{-} \in [l,\infty)) \rightarrow 1-\alpha$ and $P(\theta_{+} \in (-\infty, u]) \rightarrow 1-\alpha$. 
These results continue to hold if we replace $E[(\varphi^0_{\theta_{\pm}}(Z))^2]$ with $E_n[(\widehat{\varphi}_{\widehat{\theta}_{\pm}}(Z))^2]$.
\end{Theorem}

\subsection{Inference under benchmarking restrictions on confounding}
\label{sec:est-inf-bench}

In the benchmarking analysis introduced in Section~\ref{sec:benchmarking}, additional components characterizing the strength of observed confounding are estimable from the data. We now discuss how to account for the uncertainty in estimating these components. We begin with benchmarking against observed covariates. 

Given plausibility judgments on the relative strength parameters $k_{0\Delta Y,j},k_{0D,j}$, and setting $|\rho_0|=|\rho_{0,j}|$, the plug-in estimator of the bias bound now takes the following form:
\begin{align*}
\widehat{\theta}_{\pm} = \widehat{\theta}_s \pm |\widehat{\rho}_{0,j}|\times \sqrt{k_{0\Delta Y,j}\widehat{G}_{0\Delta Y,j} \times \frac{k_{0D,j}\widehat{G}_{0D,j}}{1-k_{0D,j}\widehat{G}_{0D,j}}}\times \sqrt{\widehat{\sigma}_{0s}^2\,\widehat{\nu}_{0s}^2}
\end{align*}
where,
\begin{align*}
\widehat{G}_{0\Delta Y,j} := \frac{\widehat{\sigma}_{0s,-j}^2 - \widehat{\sigma}_{0s}^2}{\widehat{\sigma}_{0s}^2}, \quad
\widehat{G}_{0D,j} := \frac{\widehat{\nu}_{0s}^2 - \widehat{\nu}_{0s,-j}^2}{\widehat{\nu}_{0s,-j}^2}, \ \text{and} \quad
\widehat{\rho}_{0,j} = \frac{-(\widehat{\theta}_s - \widehat{\theta}_{s,-j})}{\sqrt{(\widehat{\sigma}_{0s,-j}^2 - \widehat{\sigma}_{0s}^2) \times (\widehat{\nu}_{0s}^2 - \widehat{\nu}_{0s,-j}^2)}}.
\end{align*}
Here  $\theta_{s,-j}$ denotes the ATT estimand conditional on $X_{-j}$ alone, and $\sigma_{0s,-j}^2$ and $\nu_{0s,-j}^2$ are the scaling factors in our OVB decomposition of $\theta_s - \theta_{s,-j}$. The orthogonal scores for these parameters have the same form as (\ref{eq:psi-theta})-(\ref{eq:psi-nu}), with the only difference that $X_{-j}$ is used for estimation in place of $X$. 

We now move to benchmarking against pre-trend bias. For fixed $k$, the first approach yields the following plug-in estimator of the bias bound
\begin{align*}
    \widehat{\theta}_{\pm} &= \widehat{\theta}_s \pm k  |\widehat{\theta}_s^{\text{pre}}|,
\end{align*}
whereas the second approach gives
\begin{align*}
\widehat{\theta}_{\pm} &= \widehat{\theta}_s \pm k  \left(\frac{\widehat{S}_0}{\widehat{S}_0^{\text{pre}}} \right) |\widehat{\theta}_s^{\text{pre}}|, 
\qquad \text{with} 
\quad \widehat{S}_0^2 = \widehat{\sigma}_{0s}^2\,\widehat{\nu}_{0s}^2, 
\qquad \text{and} 
\quad \widehat{S}_0^{2,\text{pre}} = \widehat{\sigma}_{0s}^{2,\text{pre}}\,\widehat{\nu}_{0s}^{2,\text{pre}}. 
\end{align*}
Under the regularity conditions stated in Appendix~\ref{app:bench}, $\widehat{\theta}_{\pm}$ is a smooth function of asymptotically linear estimators whose influence functions were characterized above. Thus, inference for the bounds follows from the delta method. See the Appendix for corresponding influence functions and asymptotic results.

\subsection{Sensitivity statistics for routine reporting}
\label{sec:statistics}

We now introduce two sensitivity statistics that measure the minimum strength of confounding required to invalidate the conclusions of a DiD study. These statistics require no assumptions about the strength of unobserved confounding; rather, they communicate what one must be prepared to believe in order to rule out confounding that would be problematic. They serve as quick summaries of the overall robustness of the ATT estimates against systematic biases \citep{cinelli2020making, cinelli2025omitted,chernozhukov2022long}.

For $\bar R^2_{\Delta Y},\bar R^2_D\in[0,1]$, let
\(
\text{CI}^{\max}_{1-\alpha,\bar R^2_{\Delta Y},\bar R^2_D}(\theta)
\)
denote the widest confidence interval constructed from Theorem~\ref{thm: inference} under the restrictions
\(
|\rho_0|\leq1,~
\eta^2_{\Delta Y\sim U\mid X,D=0}\leq\bar R^2_{\Delta Y},~
1-R^2_{O_{XU}\sim O_X\mid D=0}\leq\bar R^2_D.
\)
The first sensitivity statistic we propose asks: leaving alignment and trend entirely unrestricted, what is the minimum strength of confounding with treatment selection alone that would lead one to not reject $H_0:\theta=\theta^*$? This yields the \emph{extreme robustness value} (XRV),
\[
\text{XRV}_{\theta^*,\alpha}(\theta)
:=
\inf\left\{
\text{XRV}:
\theta^*\in
\text{CI}^{\max}_{1-\alpha,1,\text{XRV}}(\theta)
\right\}.
\]
Any confounding scenario with selection strength below $\text{XRV}_{\theta^*,\alpha}(\theta)$ is logically incapable of overturning the original conclusions, regardless of how much variation such confounders explain of the untreated trend.

Leaving the association of the confounder with the trend completely unrestricted may be too conservative. Thus, the second sensitivity statistic considers the minimum strength of confounding in explaining both the untreated outcome trend and treatment selection jointly. The \emph{robustness value} is
\[
\text{RV}_{\theta^*,\alpha}(\theta)
:=
\inf\left\{
\text{RV}:
\theta^*\in
\text{CI}^{\max}_{1-\alpha,\text{RV},\text{RV}}(\theta)
\right\}.
\]
Any unobserved confounding with trend and selection strength below $\text{RV}_{\theta^*,\alpha}(\theta)$ is incapable of overturning the results of the study.

Appendix~\ref{app:rv} provides the general definition and inferential properties of these quantities. In particular, it shows that $\text{RV}_{\theta^*,\alpha}(\theta)$ and $\text{XRV}_{\theta^*,\alpha}(\theta)$ are asymptotically valid lower confidence bounds for their population counterparts.

\section{Applying the OVB framework to the sensitivity of DiD}
\label{sec:app}

We now return to the minimum wage example of Section~\ref{sec:background},  and show how to deploy the results of Sections~\ref{sec:ovb} to \ref{sec:est-and-inf} to answer: (i) How strong would unobserved confounders have to be to overturn the estimated negative effect in 2007? (ii) How does this required strength compare to that of the observed covariates? and, (iii) What does the 2006 pre-trend imply about the nature of unobserved confounding when read through the OVB decomposition? 

\subsection{Minimal sensitivity reporting}

Table~\ref{tab: MW-results-robustness_main} shows our proposal for minimal sensitivity reporting of DiD estimates.  The first columns report usual estimates of the treatment effect in 2007 under the assumption that parallel trends holds conditionally on the observed covariates alone. In addition to these estimates,  we propose researchers report the extreme robustness value and the robustness value \citep{cinelli2020making,cinelli2025omitted,chernozhukov2022long}. As discussed in Section~\ref{sec:statistics}, these statistics quickly convey how robust the estimated effect is to the presence of omitted variables.

The extreme robustness value ($\text{XRV}_{\theta^*=0,\ \alpha=0.05}$) of $0.2\%$ means that confounders that explain at most 0.2\% of the variation in treatment odds cannot explain away the estimated effect, at the 5\% significance level, even if such confounders were to explain all leftover variation of the outcome evolution. Alternatively, the same number can be interpreted as confounders that induce at most a $\approx0.2\%$ increase in the average odds of treatment among the treated. When the XRV is large, this may provide sufficient evidence for a causal effect in and of itself, by virtue of ruling out confounders that explain a large fraction of treatment assignment. This turns out not to be the case in our application, as it seems difficult to rule out confounders that change the odds of treatment by 0.2\%.
We thus move to examining the minimum \emph{joint} strength that confounders need to have, both in terms of explaining treatment selection and outcome evolution, in order to explain away the results.

\begin{table}[h]
\centering
{\footnotesize
\begin{tabular}{ccccc}
\toprule
\multicolumn{3}{c}{\textbf{Results Under Conditional Parallel Trends}} 
& \multicolumn{2}{c}{\textbf{Robustness Values}} \\
\cmidrule(lr){1-3} \cmidrule(lr){4-5}
\textbf{Short Estimate ($\widehat{\theta}_{s,2007}$)} 
& \textbf{Std. Error} 
& \textbf{Confidence Interval} 
& $\text{RV}_{\theta^*=0,\ \alpha=0.05}$ 
& $\text{XRV}_{\theta^*=0,\ \alpha=0.05}$\\
\midrule
-0.0366 
& 0.0105
& [-0.0572;\;-0.0160] 
& 4.38\% 
& 0.20\% \\
\bottomrule
\end{tabular}
}
\caption{{\footnotesize Minimal Sensitivity Reporting.
}}
\label{tab: MW-results-robustness_main}
\end{table}

The robustness value of $\text{RV}_{\theta^*=0,\ \alpha=0.05}=4.38\%$ indicates that unobserved confounders explaining less than $4.38\%$ of the residual variation in both treatment odds and the outcome evolution of untreated units are not sufficiently strong to render the negative effect insignificant, in the sense of shifting the upper confidence bound to zero, at the $5\%$ significance level. As before, the alternative characterizations of selection, discussed in Section~\ref{sec:treatment-selection}, offer two additional interpretations of this value---it corresponds to a $4.58\%$ relative increase in the observed average treatment odds among treated units. For a single independent binary confounder, this corresponds to a standardized mean difference of about $0.2$---which is not implausible in light of observed imbalances, as we have seen in Section~\ref{sec:background}. 

\subsection{Sensitivity contour plots and benchmarking}

Overall, neither the XRV nor the RV allowed us to quickly rule out confounding of magnitudes that would be problematic. We thus turn to the benchmarking procedure of Section~\ref{sec:bench-covariates}, which compares the gains in explanatory power from unobserved confounders with those from key observed covariates. Table~\ref{tab:benchmark_conf_bound} in the Appendix reports one-sided 95\% confidence bounds implied by an unobserved confounder comparable in strength and alignment to each observed covariate. These bounds account for the uncertainty in estimating the benchmark components, using the results of Section~\ref{sec:est-inf-bench}. We see that latent confounders comparable to observed covariates are not sufficiently strong to explain away the estimated effect. Nor would such confounders be able to reproduce the 2006 pre-trend deviation.  

We now turn to examining the full range of estimates we could have obtained under different confounding scenarios, using sensitivity contour plots \citep{Imbens03,cinelli2020making,cinelli2025omitted,chernozhukov2022long}. Since the estimated effect is negative, we focus on the upper confidence bound, which corresponds to the direction of bias relevant for overturning the original conclusion. Figure~\ref{fig:mw_contour_all_main} shows the results, for two different alignment values, a moderate scenario $|\rho_0| = .3$ and an extreme scenario $|\rho_0|=1$. The horizontal axis measures confounding in treatment selection, whereas the vertical axis measures confounding in the untreated outcome trend, both on the same $R^2$ scale. Each contour line corresponds to confounding scenarios that yield the same upper confidence bound for $\theta$, under those hypothetical strengths, and a fixed chosen alignment value. The black triangle in the lower-left corner corresponds to the upper bound of the confidence interval reported in Table~\ref{tab: MW-results-robustness_main}. The red dashed line represents the threshold scenarios that render the effect insignificant. Any confounding scenario below that line is not capable of overturning the original conclusion. The virtue of the contour plot is that the investigator can examine robustness to postulated confounding of \emph{any} hypothetical strength, by simply reading off the value at the corresponding coordinates. 

To better aid plausibility judgments, the red diamonds show the point estimates of the scenarios implied by unobserved confounders comparable in strength to observed covariates in explaining treatment selection and outcome evolution, with the alignment fixed at the value of each panel. In addition, the blue dashed line represents the scenario implied by extrapolating the 2006 pre-trend deviation, and it shows what kinds of omitted variables reproduce that amount of bias. (Table~\ref{tab:benchmark_pretrend_bounds} in the Appendix shows the same scenarios accounting for sampling uncertainty.) In accordance with the previous results, in a scenario of moderate alignment none of the benchmark points cross the critical threshold, let alone explain the pre-trend. If, however, we postulate a more conservative scenario, then we cannot rule out that adversarial confounders comparable to \textit{region}, median income (\textit{lmedinc}) or race (\textit{white}) would explain away the estimated effect. In that case, the 2006 pre-trend deviation would be rationalized by regional confounding, but would still not be reproducible by confounders comparable to race, poverty, or income.

\begin{figure}[t]
    \centering
    \begin{subfigure}[t]{0.49\textwidth}
        \centering
        \makebox[\linewidth][c]{
        \includegraphics[width=1.4\linewidth]{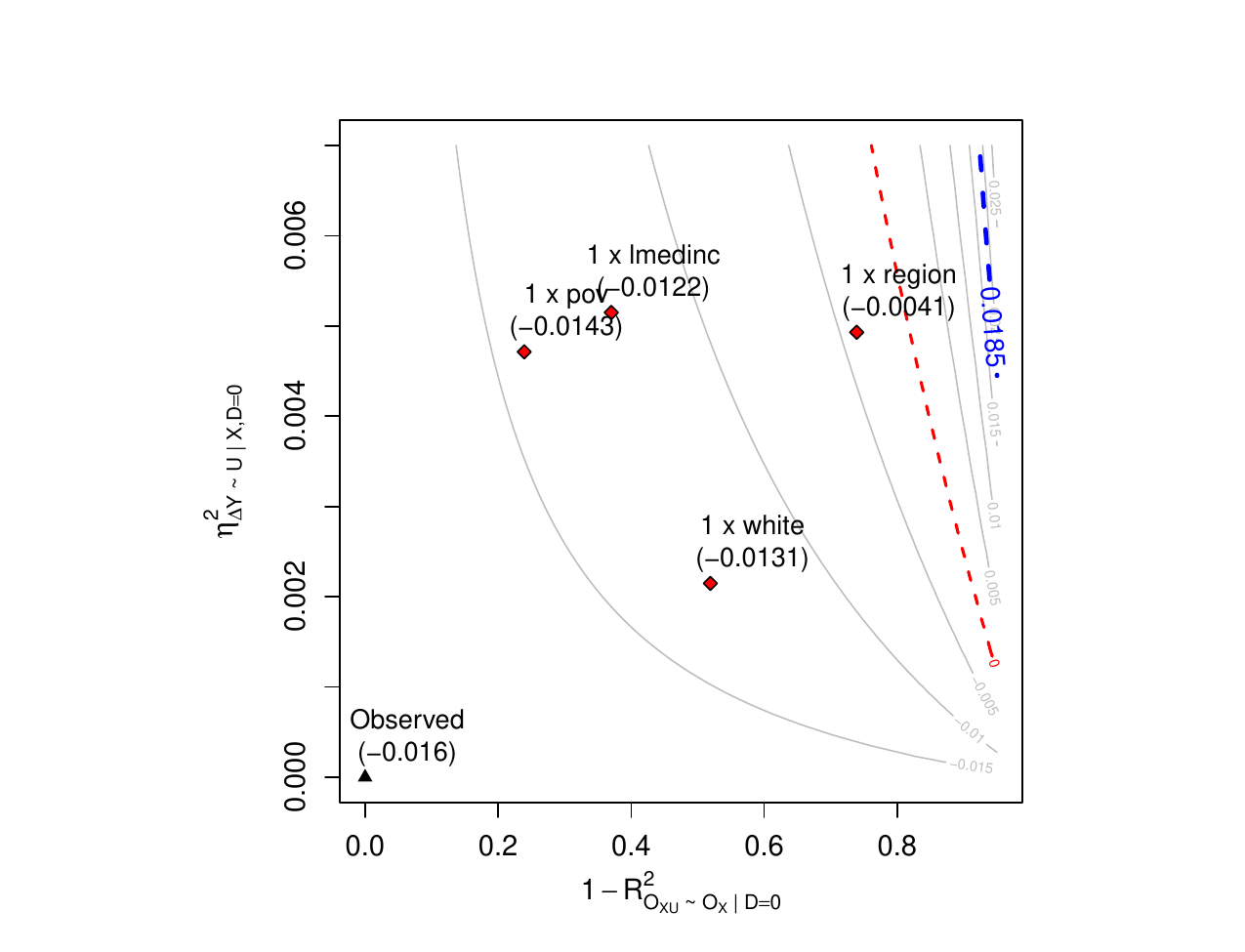}}
        \caption{\scriptsize{Upper limit conf. bound, $|\rho_{0}|=0.3$.}}
    \end{subfigure}\hfill
    \begin{subfigure}[t]{0.49\textwidth}
        \centering
        \makebox[\linewidth][c]{
        \includegraphics[width=1.4\linewidth]{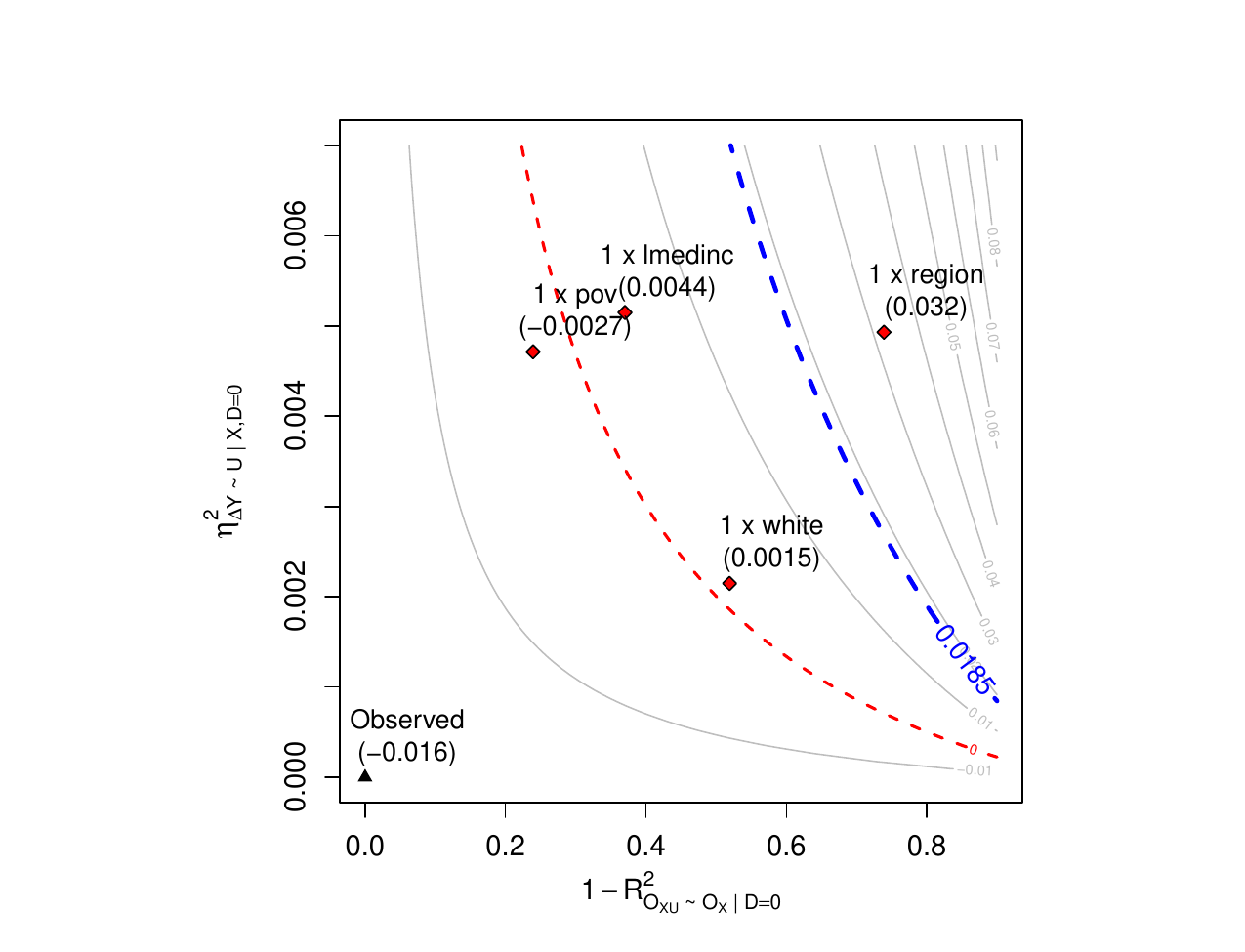}}
        \caption{\scriptsize{Upper limit conf. bound, $|\rho_{0}|=1$.}}
    \end{subfigure}
    \caption{{\footnotesize Sensitivity contour plots for the minimum wage example at a significance level of $\alpha = 0.05$. 
    }}
    \label{fig:mw_contour_all_main}
\end{figure}

Taken together, these results show that the minimum wage estimate is not completely immune to omitted confounding, but also that overturning the original analysis requires confounding of nontrivial magnitude. More importantly, the analysis clarifies exactly what one needs to believe in order to sustain the estimated effect. For example, unobserved confounders would need to be stronger---or at least more adversarial---than observed covariates to explain away the result. Reproducing the 2006 pre-trend deviation is more demanding still. Thus, a critic arguing that the estimate is not credible must articulate what plausible omitted variables remain that are not only comparable in magnitude to regional confounding, but also sufficiently adversarial. Conversely, someone defending the estimate must argue against unobserved confounders of that magnitude. Note also that sensitivity exercises such as the one above may help adjudicate between competing explanations for the 2006 deviation. If confounders of the required magnitude can be ruled out, anticipation becomes more credible---see Appendix~\ref{app:anticipation} for additional results under this scenario. While researchers may not have much confidence in answering these questions, the sensitivity analysis shifts the discussion from a generic concern about whether parallel trends might fail, to a more disciplined 
discussion about the mechanisms required to overturn the conclusion, and the competing hypotheses that could explain pre-trend deviations.

\section{Conclusion}
\label{sec:conclusion}

In this paper we provide an omitted variable bias framework for sensitivity analysis of difference-in-differences designs when unobserved confounding induces violations of parallel trends. We show that the bias in the average treatment effect on the treated admits an exact decomposition into an estimable scale factor and three bias factors, measuring confounding in treatment selection, confounding in untreated outcome evolution, and the alignment between these two channels. We also provide alternative characterizations of the treatment selection component in terms of treatment odds and covariate imbalance. Building on these results, we derive sensitivity statistics for routine reporting, develop benchmarking procedures based on observed covariates and pre-treatment biases, and provide inference methods for the resulting bounds using debiased machine learning.

Possible extensions of our framework include accommodating multiple time periods with staggered treatment adoption. In these cases, the OVB analysis we showed here can be performed period-by-period and aggregated across cohorts and time. Likewise, it is possible to perform richer benchmarking exercises in these settings, including calibration against the full sequence of pre-trends and pre-treatment covariates. Finally, extending the analysis to continuous and multi-valued treatments is also an interesting direction for future work.

\section*{Acknowledgments}
This work was supported in part by the National Science Foundation Grant No. MMS-2417955.

\par
\spacingset{1.4}
\putbib
\end{bibunit}

\newpage

\appendix
\numberwithin{assumption}{section}
\numberwithin{Proposition}{section}
\begin{bibunit}

\setcounter{page}{1}
\counterwithin{equation}{section}
\newcommand{\nindep}{\not\!\perp\!\!\!\perp}
\captionsetup[figure]{font=small,skip=0pt}
\allowdisplaybreaks
\theoremstyle{remark}
\spacingset{1.8}

\begin{center}
{\large \bfseries Supplementary Materials \\ for ``Omitted Variable Bias in Difference-in-Differences Designs''}
\end{center}

\section{Preliminaries}

\subsection{Notation}

Let $Y$ denote the outcome, $D$ the treatment indicator, $X$ the observed covariates, and $U$ the unobserved covariates. To simplify notation, we denote $p = P(D=1)$, $\pi = P(D=1 \mid X)$, $O = p/(1-p)$, and $O_X = \pi/(1-\pi)$. We recall the following standard definitions and results, which will be used in subsequent derivations.

\paragraph{Linear Projection and Residuals.} For square-integrable $Y$ and $X$ with nonsingular $E[XX^\top]$, let $\widehat{Y^X} := X^\top \beta$ denote the linear projection of $Y$ onto $X$, where $\beta := \arg\min_{b} E[(Y - X^\top b)^2] = (E[XX^\top])^{-1} E[XY]$. We define the residual as $Y^{\perp X} := Y - \widehat{Y^X}$, which represents the component of $Y$ orthogonal to the linear span of $X$. 
    
\paragraph{Radon-Nikodym Derivative.} Let $(\mathscr{X}, \mathscr{A})$ be a measurable space on which two $\sigma$-finite measures, $\mu$ and $\nu$, are defined. If $\nu$ is absolutely continuous with respect to $\mu$, denoted $\nu \ll \mu$, then by the Radon-Nikodym theorem, there exists an $\mathscr{A}$-measurable function $f:\mathscr{X} \rightarrow [0,\infty)$ such that for any measurable set $A \in \mathscr{A}$, $\nu(A) = \int_A f \dd \mu$. The function $f$, denoted $f = \dd \nu/ \dd\mu$, is called the Radon-Nikodym derivative of $\nu$ with respect to $\mu$.
    
\paragraph{$\chi^2$-divergence.} Let $P_{X|D=0}$ and $P_{X|D=1}$ denote the distributions of $X$ in the control and treated groups, respectively. Suppose $P_{X|D=1} \ll P_{X|D=0}$. We define the $\chi^2$-divergence of $P_{X|D=1}$ from $P_{X|D=0}$ as $\chi^2(P_{X|D=1}\|P_{X|D=0}) = \int\left(\frac{\dd P_{X|D=1}}{\dd P_{X|D=0}}-1\right)^2 \dd P_{X|D=0}$.
    
\paragraph{Coefficient of Variation (CV).} Let $Q_X$ denote the distribution of $X$. The coefficient of variation, $\text{CV}_X$, is a standardized measure of dispersion of $Q_X$, defined as $\text{CV}_X := \frac{\sigma_X}{\mu_X},$ where $\mu_X$ and $\sigma_X$ denote the mean and standard deviation of $X$, respectively.

\paragraph{Riesz-Frechet Representation Theorem.} Let $H$ be a Hilbert space over $\mathbb{R}$ with an inner product $\langle \cdot, \cdot \rangle$, that is complete w.r.t. the norm $\norm{\cdot}$ generated by this inner product. Let $T$ be a continuous linear functional on $H$. Then there exists a unique $w \in H$ such that for every $g \in H$, we have $T(g) = \langle g,w\rangle$. 

\subsection{Definition and properties of \texorpdfstring{$R^2$}{R-squared}}

We use $R^2_{Y \sim X}$ to denote the (possibly uncentered) $R^2$ of the orthogonal linear projection of a random variable $Y$ on a random vector $X$. That is, writing the orthogonal decomposition $Y = \widehat{Y^X} + Y^{\perp X}$,
\[
R^2_{Y \sim X}:= \frac{E[(\widehat{Y^X})^2]}{E[Y^2]} = 1- \frac{E[(Y^{\perp X})^2]}{E[Y^2]} = \frac{\left(E[(\widehat{Y^X})Y]\right)^2}{E[Y^2]E[(\widehat{Y^X})^2]}.
\]
If $X$ is univariate we also have the equality,
\[
R^2_{Y \sim X}= \frac{E[XY]^2}{E[Y^2]E[X^2]}.
\]
We define the \emph{partial} $R^2$ of $Y$ with $U$ given $X$ as:
\[
R^2_{Y \sim U \mid X} := \frac{R^2_{Y \sim U + X} - R^2_{Y \sim X}}{1- R^2_{Y\sim X}}.
\]
It is easy to show that this equals the regular $R^2$ of $Y^{\perp X}$ with $U^{\perp X}$. Note that $(Y^{\perp X})^{\perp (U^{\perp X})} = Y^{\perp U, X}$, so, 
\[
R^2_{Y\sim U \mid X} = 1 - \frac{E[(Y^{\perp U, X})^2]}{E[(Y^{\perp X})^2]}.
\]

We can use these definitions to generalize the usual (centered) $R^2$.  That is, if we want to recover the centered $R^2$, we simply partial out the constant. Formally, 
\[
R^2_{Y\sim X|1} =\frac{\operatorname{Var}(\widehat{Y^{X1}})}{\operatorname{Var}(Y)} = 1-\frac{\operatorname{Var}({Y}^{\perp X1})}{\operatorname{Var}(Y)} = \operatorname{Cor}^2(\widehat{Y^{X1}}, Y),
\]
where, for univariate $X$, $R^2_{Y\sim X|1} = \operatorname{Cor}^2(Y,X)$. Note that some other works use $R^2_{Y\sim X}$ to denote the centered version (see, e.g., \citealp{cinelli2020making, cinelli2025omitted}). 

Notice that if $X$ is a vector including the constant, then 
\[
R^2_{Y\sim U \mid X} = 1 - \frac{E[(Y^{\perp U, X})^2]}{E[(Y^{\perp X})^2]} = 1 - \frac{\operatorname{Var}(Y^{\perp U, X})}{\operatorname{Var}(Y^{\perp X})},
\]
which is centered. 

We define the nonparametric analogue of the (uncentered) $R^2$ as:
\begin{align*}
    \eta^2_{Y \sim X} := \frac{E[(E[Y \mid X])^2]}{E[Y^2]},
\end{align*}
which is equivalent to the linear (uncentered) $R^2$ from the projection of $Y$ on $E[Y \mid X]$. We define the nonparametric \textit{partial} $R^2$ of $Y$ with $U$ given $X$ as:
\begin{align*}
    \eta^2_{Y \sim U \mid X} := \frac{\eta^2_{Y \sim (X,U)} - \eta^2_{Y \sim X}}{1 - \eta^2_{Y \sim X}},
\end{align*}
which quantifies how much of the residual variation in $Y$ is explained by $U$, after accounting for $X$.

Note that the nonparametric partial $R^2$ is automatically centered, since $E[Y - E[Y \mid X]] = E[E[Y|X,U] - E[Y|X]] = 0$ by the tower property of conditional expectations. 

\begin{Proposition}[Nonparametric and linear $R^2$]
\par\noindent
   \begin{enumerate}
       \item[(i)] $\eta^2_{Y \sim X} = 1 - \frac{E\left[(Y - E[Y \mid X])^2\right]}{E[Y^2]} = R^2_{Y \sim E[Y \mid X]}$. 
       \item[(ii)] $ \eta^2_{Y \sim U \mid X} = 1 - \frac{E\left[(Y - E[Y \mid X,U])^2\right]}{E\left[(Y - E[Y \mid X])^2\right]} = R^2_{Y - E[Y \mid X] \sim E[Y \mid X,U] - E[Y \mid X]}$.
   \end{enumerate} 
\end{Proposition}
\begin{proof}
    \begin{align*}
    R^2_{Y \sim E[Y \mid X]} &:= 1 - \frac{E\left[(Y - E[Y \mid X])^2\right]}{E[Y^2]} \\
    &= \frac{E[Y^2] - \left(E[Y^2] + E[(E[Y \mid X])^2] - 2E[YE[Y \mid X]]\right)}{E[Y^2]} \\
    &\overset{\text{LTE}}{=} \frac{E[(E[Y \mid X])^2]}{E[Y^2]} = \eta^2_{Y \sim X}.
\end{align*}
Additionally, 
\begin{align*}
    R^2_{Y - E[Y \mid X] \sim E[Y \mid X,U] - E[Y \mid X]} &:= 1 - \frac{E\left[(Y - E[Y \mid X,U])^2\right]}{E[(Y - E[Y \mid X])^2]} \\
    &= \frac{E[(Y - E[Y \mid X])^2] - E\left[(Y - E[Y \mid X,U])^2\right]}{E[(Y - E[Y \mid X])^2]} \\
    &\overset{\text{LTE}}{=} \frac{E[(E[Y \mid X,U])^2] - E[(E[Y \mid X])^2]}{E[Y^2] - E[(E[Y \mid X])^2]} \\
    &= \frac{\eta^2_{Y \sim (X,U)} - \eta^2_{Y \sim X}}{1 - \eta^2_{Y \sim X}} = \eta^2_{Y \sim U \mid X}.
\end{align*}
\end{proof}

\subsection{Important relationships}

\begin{Proposition}[Relationship between selection odds] \label{prop: among_odds} Under the weak overlap condition of Assumption~\ref{asmp:regularity}, the following identities hold:
    \begin{enumerate}
        \item[(1)] \textbf{Odds ratio as change of measure.} \[\frac{O_{XU}}{O_X} = \frac{\dd P_{U|X,D=1}}{\dd P_{U|X,D=0}}.\]
        \item[(2)] \textbf{Short odds as the projection of long odds.} \[O_X = E[O_{XU} \mid X,D=0].\]
        \item[(3)] \textbf{Control-to-treated expectation identity for odds.} \[E[O^2_{XU} \mid X,D=0] = O_X \times E[O_{XU} \mid X,D=1].\]
        \item[(4)] \textbf{From conditional to marginal odds identity.} \[E[O_{XU}] = p\times (E[O_{XU}\mid D=1] + 1) = (1-p)E[O^2_{XU}\mid D=0] + p.\]
    \end{enumerate}
\end{Proposition}
\begin{proof} We establish each of the four properties in turn.
    \begin{enumerate}
        \item[(1)] By Bayes' rule, the following holds: 
        \begin{align*}
            \frac{O_{XU}}{O_X} 
            &= \frac{P(D=1 \mid X,U)P(U \mid X)}{P(D=1 \mid X)} \times \frac{P(D=0 \mid X)}{P(D=0 \mid X,U)P(U \mid X)} \\
            &= \frac{P(D=1,U \mid X)}{P(D=1 \mid X)} \times \frac{P(D=0 \mid X)}{P(D=0,U \mid X)} \\
            &= \frac{P(U \mid X,D=1)}{P(U \mid X,D=0)} \overset{\text{def.}}{=} \frac{\dd P_{U|X,D=1}}{\dd P_{U|X,D=0}}. 
        \end{align*}
        \item[(2)] We will use (1) to show the equality:
        \begin{align*}
            E\left[\frac{O_{XU}}{O_X}  \mid X,D=0\right] &\overset{(1)}{=} \int \frac{\dd P_{U|X,D=1}(u)}{\dd P_{U|X,D=0}(u)} \times \dd P_{U|X,D=0}(u) = \int \dd P_{U|X,D=1}(u) = 1, \\
            \implies O_X&= E[O_{XU} \mid X,D=0]. 
        \end{align*}
        \item[(3)] We will use (1) to show the equality:
        \begin{align*}
            E\left[\left(\frac{O_{XU}}{O_X}\right)^2  \mid X,D=0\right] &\overset{(1)}{=} \int \left(\frac{\dd P_{U|X,D=1}(u)}{\dd P_{U|X,D=0}(u)}\right)^2 \times \dd P_{U|X,D=0}(u) \\
            &= \int \frac{\dd P_{U|X,D=1}(u)}{\dd P_{U|X,D=0}(u)} \times \dd P_{U|X,D=1}(u) \\
            &\overset{(1)}{=} E\left[\frac{O_{XU}}{O_X}   \mid X,D=1\right], \\
            \implies E[O^2_{XU} \mid X,D=0] &= O_X \times E[O_{XU} \mid X,D=1]. 
        \end{align*}
        \item[(4)] By the tower property, the following holds:
        \begin{align*}
            E[O_{XU}] &\overset{\text{LTE}}{=} E\left[E[O_{XU}\mid D]\right] \\
            &= E[O_{XU}\mid D=1] \times p + E[O_{XU}\mid D=0] \times (1-p)\\
            &\overset{(2)}{=} E[O_{XU}\mid D=1] \times p + O \times (1-p) \\
            &= p\times (E[O_{XU}\mid D=1] + 1) \\
            &\overset{(3)}{=} p + p\times \frac{E[O^2_{XU}\mid D=0]}{O} \\
            &= p + (1-p)E[O^2_{XU}\mid D=0]. 
        \end{align*}
    \end{enumerate}
\end{proof}

\begin{Proposition}[Odds and $\chi^2$-divergence] Under the weak overlap condition of Assumption~\ref{asmp:regularity}, the following identities hold: \label{prop: odds_and_divergence}
    \begin{enumerate}
        \item[(1)] \textbf{$\chi^2$-divergence as squared coefficient of variation.}
        \[
        \chi^2(P_{X,U|D=1}\|P_{X,U|D=0}) = \operatorname{Var}\left(\frac{O_{XU}}{O} \mid D=0\right) = \text{CV}_{O_{XU}|0}^2.
        \]
        \item[(2)] \textbf{$\chi^2$-divergence as average conditional expected odds.} 
        \begin{align*}
        \chi^2(P_{X,U|D=1}\|P_{X,U|D=0}) &= E\left[\frac{O_{XU}}{O} \mid D=1\right]-1, \\
        \chi^2(P_{U|X,D=1}\|P_{U|X,D=0}) &= E\left[\frac{O_{XU}}{O_X}  \mid X,D=1\right]-1.    
        \end{align*}
        \item[(3)] \textbf{Conditional $\chi^2$-weighted decomposition of odds increase.}
        \[
        E[O_{XU}\mid D=1] - E[O_X\mid D=1] = E\left[\chi^2(P_{U|X,D=1}\|P_{U|X,D=0})\times O_X \mid D=1\right].
        \]
        \item[(4)] \textbf{Marginal-conditional decomposition of $\chi^2$-divergence.}
        \[\chi^2(P_{X,U|D=1}\|P_{X,U|D=0}) = \chi^2(P_{X|D=1}\|P_{X|D=0}) + E_{P_{X|D=0}}\left[\left(\frac{\dd P_{X|D=1}}{\dd P_{X|D=0}}\right)^2\chi^2(P_{U|X,D=1}\|P_{U|X,D=0})\right].
        \]
    \end{enumerate}
\end{Proposition}
\begin{proof}
    We establish each of the four properties in turn, using Proposition~\ref{prop: among_odds}. Throughout, we let $Z := (X,U)$. 
    \begin{enumerate}
        \item[(1)]
        \begin{align*}
            \chi^2(P_{Z|1}\|P_{Z|0}) &:= \int \left(\frac{\dd P_{Z|1}(z)}{\dd P_{Z|0}(z)} -1\right)^2 \dd P_{Z|0}(z) \\
            &= \int\left(\frac{\dd P_{Z|1}(z)}{\dd P_{Z|0}(z)}\right)^2 \times \dd P_{Z|0}(z)  + \int \dd P_{Z|0}(z) - 2\int\frac{P_{Z|1}(z)}{P_{Z|0}(z)} \dd P_{Z|0}(z) \\
            &= \int\left(\frac{\dd P_{Z|1}(z)}{\dd P_{Z|0}(z)}\right)^2 \dd P_{Z|0}(z) + 1 - 2 \\
            &=  E\left[\left(\frac{O_Z}{O}\right)^2 \mid D=0\right] -1 \text{ by Proposition~\ref{prop: among_odds} (1),} \\
            &= E\left[\left(\frac{O_Z}{O}\right)^2 \mid D=0\right] - \left(E\left[\frac{O_Z}{O} \mid D=0\right]\right)^2 \text{ by Proposition~\ref{prop: among_odds} (2),} \\
            &= \operatorname{Var}\left(\frac{O_Z}{O} \mid D=0\right) = \frac{\operatorname{Var}(O_Z\mid D=0)}{O^2} \\
            &= \frac{\operatorname{Var}(O_Z\mid D=0)}{E^2[O_Z\mid D=0]}\text{by Proposition~\ref{prop: among_odds} (2),} \\
            &\overset{\text{def.}}{=} \text{CV}_{O_{XU}|0}^2. 
            \end{align*}
        Following similar steps, we can show 
        \begin{align*}
            \chi^2(P_{U|X,D=1}\|P_{U|X,D=0}) = \operatorname{Var}\left(\frac{O_{XU}}{O_X}  \mid X,D=0\right) = \text{CV}_{O_{XU}|X,0}^2.
        \end{align*}
        \item[(2)]
        \begin{align*}
            \chi^2(P_{Z|1}\|P_{Z|0}) &\overset{(1)}{=} \operatorname{Var}\left(\frac{O_Z}{O} \mid D=0\right) = \frac{\operatorname{Var}(O_Z\mid D=0)}{O^2} \\
            &= \frac{E[O^2_Z\mid D=0] - O^2}{O^2} \text{ by Proposition~\ref{prop: among_odds} (2),} \\
            &= E\left[\frac{O_Z}{O} \mid D=1\right] - 1 \text{ by Proposition~\ref{prop: among_odds} (3).}
        \end{align*}
        Following similar steps, we can show
        \begin{align*}
            \chi^2(P_{U|X,D=1}\|P_{U|X,D=0}) = E\left[\frac{O_{XU}}{O_X}  \mid X,D=1\right] - 1. 
        \end{align*}
        \item[(3)]
        \begin{align*}
            &E[O_{XU}\mid D=1] - E[O_X\mid D=1] = E\left[\left(\frac{O_{XU}}{O_X} - 1\right) \times O_X \mid D=1\right] \\
            &\overset{\text{LTE}}{=} E\left[\left(E\left[\frac{O_{XU}}{O_X}   \mid X,D=1\right]-1\right) \times O_X \mid D=1\right] \\
            &\overset{(2)}{=} E\left[\chi^2(P_{U|X,D=1}\|P_{U|X,D=0}) \times O_X\mid D=1\right].
        \end{align*}
        \item[(4)] 
        \begin{align*}
            &\chi^2(P_{X,U|1}\|P_{X,U|0}) - \chi^2(P_{X|1}\|P_{X|0}) \overset{(2)}{=} E\left[\frac{O_{XU}}{O} \mid D=1\right] - E\left[\frac{O_X}{O} \mid D=1\right] \\
            &\overset{(3)}{=} \frac{1}{O}E\left[\chi^2(P_{U|X,D=1}\|P_{U|X,D=0})\times O_X \mid D=1\right] \\
            &= E\left[\chi^2(P_{U|X,D=1}\|P_{U|X,D=0})\times \left(\frac{\dd P_{X|1}}{\dd P_{X|0}}\right) \mid D=1\right]\text{ by Proposition~\ref{prop: among_odds} (1),} \\
            &= \int \chi^2(P_{U|X,D=1}\|P_{U|X,D=0}) \times \left(\frac{\dd P_{X|1}}{\dd P_{X|0}}\right) \dd P_{X|1} \\
            &= \int \chi^2(P_{U|X,D=1}\|P_{U|X,D=0}) \times \left(\frac{\dd P_{X|1}}{\dd P_{X|0}}\right)^2 \dd P_{X|0} \\
            &= E_{P_{X|0}}\left[\left(\frac{\dd P_{X|1}}{\dd P_{X|0}}\right)^2\chi^2(P_{U|X,D=1}\|P_{U|X,D=0})\right]. 
        \end{align*}
    \end{enumerate}
\end{proof}

\subsection{Standard DiD assumptions and identification of ATT}
\label{app:did-assump}

This section presents the standard identifying assumptions for the ATT in the canonical DiD design and shows that, under these assumptions, our target estimand admits a causal interpretation.

\begin{assumption}[Consistency]  \label{assump:consistency}
    Observed outcomes are generated as $Y:= Y(D)$.
\end{assumption}

\begin{assumption}[No anticipation]  \label{assump:No anticip}
    Treatment has no effect on pre-treatment outcomes.
    \[E[Y_1(0)\mid D=1, X,U] = E[Y_1(1)\mid D=1, X,U] \ \ a.s.\]
\end{assumption}

\begin{assumption}[Conditional parallel trends assumption] \label{assump:PTA} Let $\Delta Y(0) := Y_2(0) - Y_1(0)$ denote the untreated potential outcome evolution, and assume that 
    \[E[\Delta Y(0)\mid D=1,X,U] = E[\Delta Y(0)\mid D=0,X,U] \ \ a.s.\]
\end{assumption}

\begin{assumption}[Strong overlap]\label{assump:Overlap}
    There exists $\epsilon > 0$ such that
    \[p := P(D=1) \ge \epsilon, \ \text{and} \ \pi_{XU} := P(D=1\mid X,U) \le 1 - \epsilon, \ \ a.s.\] 
\end{assumption}

\begin{remark}   
Assumption~\ref{assump:Overlap} is stronger than what is needed to derive the OVB formula of Theorem~\ref{thm:main_npm}, in which we impose only the weak overlap condition of Assumption~\ref{asmp:regularity}.
\end{remark}

\begin{Proposition}[DiD identifies the ATT]\label{lemma:Identification}
    Let $\Delta Y:= Y_2 - Y_1$ denote the observed outcome evolution, and let $g_0(X,U) := E[\Delta Y \mid D=0,X,U]$. Under Assumptions \ref{assump:consistency}-\ref{assump:Overlap}, the ATT is identified by \[\theta:= E\left[\Delta Y - g_0(X,U) \mid D=1\right].\]
    \begin{proof}
        By no anticipation, 
        \begin{align*}
            \text{ATT} &=  E\left[Y_2(1) - Y_2(0)\middle|D=1\right] \\
            &= E[Y_2(1) - (Y_1(1) - Y_1(1)) - Y_2(0) + (Y_1(0) - Y_1(0)) \mid D =1] \\
            &= E[\Delta Y(1) \mid D=1] - E[\Delta Y(0) \mid D=1] + E[Y_1(1) - Y_1(0) \mid D=1] \\
            &= E[\Delta Y(1) \mid D=1] - E[\Delta Y(0) \mid D=1].
        \end{align*}
        Then, we have
        \begin{align*}
            \text{ATT} &= E[\Delta Y(1) \mid D=1] - E[\Delta Y(0) \mid D=1] \\
            &= E[\Delta Y \mid D=1] - E[\Delta Y(0) \mid D=1] \text{ by consistency,} \\
            &= E[\Delta Y \mid D=1] - E\left[E[\Delta Y(0) \mid X,U,D=1]\mid D=1\right] \text{ by LTE,} \\
            &= E[\Delta Y \mid D=1] - E\left[E[\Delta Y(0) \mid X,U,D=0]\mid D=1\right] \text{ by conditional PTA,} \\
            &= E[\Delta Y - g_0(X,U) \mid D=1] \\
            &= \theta.\qedhere
        \end{align*}
    \end{proof}
\end{Proposition}

\section{Deferred proofs}

For simplicity, let $\alpha_0 := O_{XU}/O$, and $\alpha_{0s}:= O_X/O$. 

\begin{proof}[Proof of Lemma~\ref{lemma:cond_RR_main}]~\\
    We begin by showing that $\theta_{0s} = E[g_{0s}\alpha_{0s}\mid D=0]$. 
        \begin{align}
            \theta_{0s} &:= E\left[g_{0s}(X) \mid D=1\right] = E\left[\frac{D}{p} \times g_{0s}(X)\right] = E\left[\frac{P(D=1 \mid X)}{p} \times g_{0s}(X)\right] \text{ by LTE,} \nonumber \\
            &= E\left[\left(\frac{P(D=1 \mid X)}{P(D=0 \mid X)} \times \frac{1-p}{p}\right) \times \frac{P(D=0 \mid X)}{1-p} \times g_{0s}(X)\right] \nonumber \\
            &= E\left[\alpha_{0s}(X)g_{0s}(X) \times \frac{P(D=0 \mid X)}{1-p}\right] \nonumber \\
            &= E\left[\alpha_{0s}(X)g_{0s}(X) \times \frac{1-D}{1-p}\right] \text{ by LTE,} \nonumber \\
            &= E[g_{0s}\alpha_{0s}\mid D=0]. \nonumber 
        \end{align}
        The same argument shows that $\theta_0 = E[g_0\alpha_0\mid D=0]$. 
        Next, by Proposition~\ref{prop: among_odds} (2), we have that $\alpha_{0s} = E[\alpha_{0}\mid X,D=0]$ and that $E[\alpha_{0}\mid D=0] = E[\alpha_{0s}\mid D=0] = 1$. 
\end{proof}

    \begin{proof}[Proof of Theorem~\ref{thm:main_npm}]~\\
        We first show that $\theta - \theta_s = -\operatorname{Cov}(g_0 - g_{0s}, \alpha_{0} - \alpha_{0s}\mid D=0)$. 
        \begin{align}
        \theta - \theta_s &= -(\theta_0 - \theta_{0s}) \nonumber \\
        &=E[g_{0s}\alpha_{0s}\mid D=0] - E[g_0\alpha_0\mid D=0] \nonumber \\
        &= E[g_{0s}\alpha_{0s}\mid D=0] - \left[E[(g_{0s}+g_0-g_{0s})(\alpha_{0s} + \alpha_{0} - \alpha_{0s})\mid D=0]\right] \nonumber \\
            &= -\left[E[g_{0s}(\alpha_{0} - \alpha_{0s})\mid D=0] + E[\alpha_{0s}(g_0 - g_{0s})\mid D=0] + E[(g_0 - g_{0s})(\alpha_{0} - \alpha_{0s})\mid D=0]\right] \nonumber \\
            &= -E[(g_0 - g_{0s})(\alpha_{0} - \alpha_{0s})\mid D=0] \nonumber \\
            &= -\operatorname{Cov}(g_0 - g_{0s}, \alpha_{0} - \alpha_{0s}\mid D=0). \nonumber
    \end{align}
    The second-to-last equality uses that given $D=0$, $g_{0s}$ is orthogonal to $(\alpha_{0} - \alpha_{0s})$ and $\alpha_{0s}$ is orthogonal to $(g_0 - g_{0s})$. The last equality uses that $E[\alpha_{0}\mid D=0] = E[\alpha_{0s}\mid D=0] = 1$.

    Next, we decompose this covariance in the following way: 
    \begin{align}
    |\theta - \theta_{s}|^2 &=  \operatorname{Cov}^2(g_0 - g_{0s}, \alpha_{0} - \alpha_{0s}\mid D=0) \nonumber \\
    &= \underbrace{\frac{\operatorname{Cov}^2(g_0 - g_{0s}, \alpha_{0} - \alpha_{0s}\mid D=0)}{\operatorname{Var}(g_0 - g_{0s}\mid D=0)\operatorname{Var}(\alpha_{0} - \alpha_{0s}\mid D=0)}}_{=:\rho_0^2} \times \operatorname{Var}(g_0 - g_{0s}\mid D=0)\operatorname{Var}(\alpha_{0} - \alpha_{0s}\mid D=0) \nonumber \\
    \text{Since } &E[g_0 - g_{0s}\mid D=0] = E[\alpha_{0} - \alpha_{0s}\mid D=0] = 0, \nonumber \\
    &= \rho_0^2E[(g_0-g_{0s})^2\mid D=0]E[(\alpha_{0} - \alpha_{0s})^2\mid D=0] \nonumber \\
    &= \rho_0^2 \times \underbrace{\frac{E[(g_0-g_{0s})^2\mid D=0]}{E[(\Delta Y-g_{0s})^2\mid D=0]}}_{=: C_{0\Delta Y}^2} \times \underbrace{\frac{E[(\alpha_{0} - \alpha_{0s})^2\mid D=0]}{E[\alpha_{0s}^2\mid D=0]}}_{=: C_{0D}^2} \times S_0^2\nonumber \\
    &= \rho_0^2 C_{0\Delta Y}^2C_{0D}^2S_0^2, \nonumber
    \end{align}
    where $S_0^2 := E[(\Delta Y -g_{0s})^2\mid D=0]E[\alpha_{0s}^2\mid D=0]$. 
    
    Now, we show that $C_{0\Delta Y}^2 = R^2_{\Delta Y-g_{0s} \sim g_0-g_{0s}\mid D=0}$:
    \begin{align}
        R^2_{\Delta Y-g_{0s} \sim g_0-g_{0s}\mid D=0} &:= 1 - \frac{E[(\Delta Y-g_0)^2\mid D=0]}{E[(\Delta Y - g_{0s})^2\mid D=0]} \nonumber \\
        &= \frac{E[(\Delta Y - g_0 + g_0 -g_{0s})^2\mid D=0] - E[(\Delta Y - g_{0})^2\mid D=0]}{E[(\Delta Y - g_{0s})^2\mid D=0]} \nonumber \\
        &= \frac{E[(g_0-g_{0s})^2\mid D=0] + 2E[(\Delta Y-g_0)(g_0 - g_{0s})\mid D=0]}{E[(\Delta Y - g_{0s})^2\mid D=0]} \nonumber \\
        &\overset{\text{LTE}}{=} \frac{E[(g_0-g_{0s})^2\mid D=0] + 2E[(g_0-g_0)(g_0 - g_{0s})\mid D=0]}{E[(\Delta Y - g_{0s})^2\mid D=0]} \nonumber \\
        &= C_{0\Delta Y}^2. \nonumber
    \end{align}
    
    Next, we relate $C_{0\Delta Y}^2$ to the nonparametric partial $R^2$ measure:
    \begin{align}
        C_{0\Delta Y}^2 &:= \frac{E[(g_0 - g_{0s})^2\mid D=0]}{E[(\Delta Y-g_{0s})^2\mid D=0]} \nonumber \\
        &\overset{\text{LTE}}{=} \frac{E[g_0^2\mid D=0] - E[g_{0s}^2\mid D=0]}{E[\Delta Y^2\mid D=0] - E[g_{0s}^2\mid D=0]} \nonumber \\
        &= \frac{\frac{E[g_{0}^2\mid D=0]}{E[\Delta Y^2\mid D=0]}-\frac{E[g_{0s}^2\mid D=0]}{E[\Delta Y^2\mid D=0]}}{1 - \frac{E[g_{0s}^2\mid D=0]}{E[\Delta Y^2\mid D=0]}} \nonumber \\
        &\overset{\text{def.}}{=} \eta^2_{\Delta Y\sim U \mid X,D=0}. \nonumber
        \end{align}

        Note that
        \begin{align}
        R^2_{\alpha_0 \sim \alpha_{0s}\mid D=0} &:= 1 - \frac{E[(\alpha_0 - \alpha_{0s})^2\mid D=0]}{E[\alpha_{0}^2\mid D=0]} \nonumber \\
        &= \frac{E[\alpha_{0}^2\mid D=0] - (E[\alpha_{0}^2\mid D=0] + E[\alpha_{0s}^2\mid D=0] - 2E[\alpha_{0}\alpha_{0s}\mid D=0])}{E[\alpha_{0}^2\mid D=0]} \nonumber \\
        &\overset{\text{LTE}}{=} \frac{E[\alpha_{0}^2\mid D=0] - (E[\alpha_{0}^2\mid D=0] + E[\alpha_{0s}^2\mid D=0] - 2E[\alpha_{0s}^2\mid D=0])}{E[\alpha_{0}^2\mid D=0]}  \nonumber \\
        &= \frac{E[\alpha_{0s}^2\mid D=0]}{E[\alpha_{0}^2\mid D=0]}. \nonumber
        \end{align}
        Additionally,
        \begin{align*}
        R^2_{O_{XU}\sim O_X\mid D=0} &:= 1 - \frac{E[(O_{XU} - O_X)^2\mid D=0]}{E[O^2_{XU}\mid D=0]} \nonumber \\
        &= \frac{E[O^2_{XU}\mid D=0] - \left(E[O^2_{XU}\mid D=0] + E[O^2_X\mid D=0] - 2E[O_{XU}O_X\mid D=0]\right)}{E[O^2_{XU}\mid D=0]} \nonumber \\
        &\overset{\text{LTE}}{=}  \frac{E[O^2_{XU}\mid D=0] - \left(E[O^2_{XU}\mid D=0] + E[O^2_X\mid D=0] - 2E[O^2_X\mid D=0]\right)}{E[O^2_{XU}\mid D=0]} \nonumber \\
        &= \frac{E[O^2_X\mid D=0]}{E[O^2_{XU}\mid D=0]} =\frac{E[\alpha_{0s}^2\mid D=0]}{E[\alpha_{0}^2\mid D=0]} \\
        &= R^2_{\alpha_0 \sim \alpha_{0s}\mid D=0}. \nonumber
        \end{align*}

        Accordingly, we express $C_{0D}^2$ in terms of $R^2$ measures:
        \begin{align}
        C_{0D}^2 &= \frac{E[(\alpha_0 - \alpha_{0s})^2\mid D=0]}{E[\alpha_{0s}^2\mid D=0]} = \frac{E[\alpha_{0}^2\mid D=0] + E[\alpha_{0s}^2\mid D=0] - 2E[\alpha_{0} \alpha_{0s}\mid D=0]}{E[\alpha_{0s}^2\mid D=0]} \nonumber \\
        &\overset{\text{LTE}}{=} \frac{E[\alpha_{0}^2\mid D=0] - E[\alpha_{0s}^2\mid D=0]}{E[\alpha_{0s}^2\mid D=0]}\nonumber \\
        &= \frac{1 - R^2_{\alpha_{0}\sim \alpha_{0s}\mid D=0}}{R^2_{\alpha_{0} \sim \alpha_{0s}\mid D=0}}\nonumber \\
        &= \frac{1-R^2_{O_{XU}\sim O_{X}|D=0}}{R^2_{O_{XU}\sim O_{X}|D=0}}. \nonumber
    \end{align}

    Finally, we simplify the scaling factor $S_0^2$ as follows:
    \begin{align}
        \sigma_{0s}^2 &:= E[(\Delta Y - g_{0s})^2\mid D=0] = \operatorname{Var}(\Delta Y - g_{0s}\mid D=0) \nonumber \\
        &\overset{\text{LTV}}{=} \operatorname{Var}\left(E\left[\Delta Y - g_{0s} \mid X,D=0\right] \mid D=0\right) + E\left[\operatorname{Var}\left(\Delta Y - g_{0s} \mid X,D=0\right) \mid D=0\right] \nonumber \\
        &= E[\operatorname{Var}(\Delta Y \mid X,D=0)\mid D=0], \text{ and}\nonumber \\
        \nu_{0s}^2 &:= E[\alpha_{0s}^2\mid D=0] \overset{\text{def.}}{=} E\left[\left(\frac{O_X}{O}\right)^2 \middle| D=0\right], \nonumber \\
        \implies S_0^2 &:= \sigma_{0s}^2\nu_{0s}^2 = E[\operatorname{Var}(\Delta Y \mid X,D=0)\mid D=0] \times E\left[\left(\frac{O_X}{O}\right)^2 \middle| D=0\right]. \nonumber \qedhere
    \end{align}
    \end{proof}

\begin{proof}[Proof of Lemma~\ref{thm:main_connection}]
    \emph{Parts 1} and \emph{2} are immediate consequences of Proposition~\ref{prop: among_odds}, items~(1) and~(3) respectively. \emph{Part~3} follows directly from Proposition~\ref{prop: odds_and_divergence}(2).
\end{proof}

    \begin{proof}[Proof of Corollary~\ref{thm:alt-selection}]~\\
    \emph{Part 1.} 
     By Proposition~\ref{prop: among_odds}(3), we transform $R^2_{O_{XU} \sim O_X\mid D=0}$ into an expression in terms of the average selection odds among the treated.
    \begin{align}
        R^2_{O_{XU} \sim O_X\mid D=0} &= \frac{E[O^2_X\mid D=0]}{E[O^2_{XU}\mid D=0]} = \frac{E[O_X\mid D=1]}{E[O_{XU}\mid D=1]},\nonumber \\
        \implies C_{0D}^2 &= \frac{1 - R^2_{O_{XU} \sim O_X\mid D=0}}{R^2_{O_{XU} \sim O_X\mid D=0}} = \frac{E[O_{XU}\mid D=1] - E[O_X\mid D=1]}{E[O_X\mid D=1]}. \nonumber 
    \end{align}
    Correspondingly, we have
    \begin{align*}
        1 - R^2_{O_{XU} \sim O_X\mid D=0} &= 1 - \frac{E[O_X\mid D=1]}{E[O_{XU}\mid D=1]} = \frac{E[O_{XU}\mid D=1] - E[O_X\mid D=1]}{E[O_{XU}\mid D=1]}, 
    \end{align*}
    and 
    \begin{align*}
        \nu_{0s}^2 &= E\left[\left(\frac{O_X}{O}\right)^2 \middle| D=0\right] = E\left[\frac{O_X}{O} \middle| D=1\right]. 
    \end{align*}

    \noindent
    \emph{Part 2.} By Proposition~\ref{prop: among_odds}(2) and Proposition~\ref{prop: odds_and_divergence}(1), we have that, 
        \begin{align*}
            \chi^2(P_{X|1}\|P_{X|0}) = E\left[\left(\frac{O_X}{O}\right)^2 \middle| D=0\right] - 1, \text{ and } \chi^2(P_{X,U|1}\|P_{X,U|0}) = E\left[\left(\frac{O_{XU}}{O}\right)^2 \middle| D=0\right] - 1. 
        \end{align*}
        Therefore, 
        \begin{align*}
            R^2_{O_{XU} \sim O_X\mid D=0} &= \frac{E[O^2_X\mid D=0]}{E[O^2_{XU}\mid D=0]} \\
            &= \frac{(\chi^2(P_{X|1}\|P_{X|0}) +1) \times O^2}{(\chi^2(P_{X,U|1}\|P_{X,U|0}) +1) \times O^2} \\
            &= \frac{\chi^2(P_{X|1}\|P_{X|0}) +1}{\chi^2(P_{X,U|1}\|P_{X,U|0}) +1}, \\
            \implies C_{0D}^2 &= \frac{1 - R^2_{O_{XU} \sim O_X\mid D=0}}{R^2_{O_{XU} \sim O_X\mid D=0}} = \frac{1- \frac{\chi^2(P_{X|1}\|P_{X|0}) +1}{\chi^2(P_{X,U|1}\|P_{X,U|0}) +1}}{\frac{\chi^2(P_{X|1}\|P_{X|0}) +1}{\chi^2(P_{X,U|1}\|P_{X,U|0}) +1}} \\
            &= \frac{\chi^2(P_{X,U|1}\|P_{X,U|0}) - \chi^2(P_{X|1}\|P_{X|0})}{\chi^2(P_{X|1}\|P_{X|0}) + 1}. 
        \end{align*}
        Correspondingly, we have
        \begin{align*}
            1 - R^2_{O_{XU} \sim O_X\mid D=0} &= 1 - \frac{\chi^2(P_{X|1}\|P_{X|0}) +1}{\chi^2(P_{X,U|1}\|P_{X,U|0}) +1} = \frac{\chi^2(P_{X,U|1}\|P_{X,U|0}) - \chi^2(P_{X|1}\|P_{X|0})}{\chi^2(P_{X,U|1}\|P_{X,U|0}) +1},
        \end{align*}
        and \[\nu_{0s}^2 := E\left[\left(\frac{O_X}{O}\right)^2 \middle| D=0\right]  = \chi^2(P_{X|1}\|P_{X|0}) + 1.\]
    \end{proof}

\begin{proof}[Proof of Proposition~\ref{prop:lls-restriction}]
Multiplying the identity of Proposition~\ref{prop: CB_connection}(i) by equation~\eqref{eq:rho-unc-explicit} and simplifying, the treated-group terms cancel, yielding
\[
\rho^2 C_{\Delta Y}^2 = \frac{P(D=0)\,\sigma_{0s}^2}{E[\operatorname{Var}(\Delta Y \mid X, D)]}\,\rho_0^2\, C_{0\Delta Y}^2.
\]
The inequality follows since $\rho_0^2 \le 1$ and $C_{0\Delta Y}^2 \le 1$. Finally, combining this identity with equation~\eqref{eq:unc-cond-product} yields the equality $\sqrt{P(D=0)\,\sigma_{0s}^2/E[\operatorname{Var}(\Delta Y \mid X, D)]}\; C_D\, S = C_{0D}\, S_0$ stated in the main text; the degenerate case $\rho_0\, C_{0\Delta Y} = 0$ follows directly from the computations leading to \eqref{eq:unc-cond-product}, which do not involve the correlation factors.
\end{proof}

\begin{proof}[Proof of Proposition~\ref{prop:hp-correlation}]
By Proposition~\ref{prop: HS_connection}(ii) and equation~\eqref{eq:R2-DY-g0s},
\[
\rho_{w0}^2 = \rho_0^2\, C_{0\Delta Y}^2\, \frac{\sigma_{0s}^2}{\operatorname{Var}(\Delta Y \mid D=0)}.
\]
The first inequality and its equality condition then follow from $\rho_0^2\, C_{0\Delta Y}^2 \le 1$. For the second inequality, equation~\eqref{eq:maximal-cor} implies $\operatorname{Cor}^2(O_X, \Delta Y \mid D=0) \le \operatorname{Cor}^2(g_{0s}, \Delta Y \mid D=0) = R^2_{\Delta Y \sim g_{0s} \mid 1, D=0}$, which rearranges to $\sigma_{0s}^2/\operatorname{Var}(\Delta Y \mid D=0) \le \bar\rho_{w0}^2$ by \eqref{eq:R2-DY-g0s}. Since $\operatorname{Cov}(O_X, \Delta Y \mid D=0) = \operatorname{Cov}(O_X, g_{0s} \mid D=0)$, equality holds iff the Cauchy--Schwarz inequality for $\operatorname{Cov}(O_X, g_{0s} \mid D=0)$ is an equality, that is, iff $g_{0s}$ is an affine function of $O_X$ almost surely among the untreated.
\end{proof}

\begin{proof}[Proof of Corollary~\ref{cor:msm-parameterization}]
By Theorem~3 of \citet{chernozhukov2022long}, the marginal sensitivity model implies the sharp bound
\[
C_D^2\leq\left(\frac{E[O_X]-p}{E[O_X]}\right)\frac{(\Lambda-1)^2}{\Lambda}.
\]
Proposition~\ref{prop: CB_connection}(ii), Corollary~\ref{thm:alt-selection}, and Proposition~\ref{prop: among_odds}(4) give the exact relationship
\[
C_D^2=\left(\frac{E[O_X]-p}{E[O_X]}\right)C_{0D}^2.
\]
The factor multiplying $C_{0D}^2$ is strictly positive under overlap, so the first two displays imply $C_{0D}^2\leq(\Lambda-1)^2/\Lambda$. Moreover, Proposition~\ref{prop: HS_connection}(i) and Proposition~\ref{prop: among_odds}(4) give, whenever $C_{w0D}^2$ is well defined,
\[
C_{w0D}^2=\left(\frac{E[O_X]-p}{E[O_X]-O}\right)C_{0D}^2,
\]
which yields the stated bound on $C_{w0D}^2$. Because the two relationships above are exact and their factors are fixed by the observed distribution and strictly positive on their respective domains, the same extremal distributions attain the bounds on $C_{0D}^2$ and $C_{w0D}^2$. Hence all three bounds are sharp.
\end{proof}

    \begin{proof}[Proof of Lemma~\ref{lemma:Properties of the DML Bound Estimators} and Theorem~\ref{thm: inference}]~\\
    Lemma~\ref{lemma:Properties of the DML Bound Estimators} is a direct application of Theorems~3.1 and 3.2 in \citet{chernozhukov2018double}. Valid estimation of covariance follows similarly to the proof of Theorem~3.2 in \citet{chernozhukov2018double}. The first claim of Theorem~\ref{thm: inference} then follows by the delta method \citep{van1996weak}, and the stated confidence intervals are justified by standard arguments on asymptotic normality.

    Below, we verify that the scores in \eqref{eq:psi-theta}–\eqref{eq:psi-nu} are Neyman orthogonal, as required by Theorem~3.1 of \citet{chernozhukov2018double}. 
        \begin{itemize}
            \item \textbf{Score for $\theta_s$.}
                \begin{align}
        \intertext{\emph{Moment condition:}}
        &E[\psi_{\theta_s}(Z;g_{0s},\pi,p)] \nonumber \\
        &= E[\Delta Y - g_{0s}\mid D=1] - E[(\Delta Y - g_{0s})\alpha_{0s}\mid D=0] - \theta_s \nonumber \\
        &\overset{\text{LTE}}{=} \underbrace{E[\Delta Y - g_{0s}\mid D=1]}_{\theta_s} - \theta_s\nonumber \\
        &= 0. \nonumber \\
        \intertext{\emph{Orthogonality with respect to $g_{0s}$:}}
        &\partial_{r} \left. \left\{E[\psi_{\theta_s}(Z; g_{0s} + r\underbrace{(g_{0s}^{\prime} - g_{0s})}_h, \pi,p)]\right\} \right|_{r=0} \nonumber \\
        &= \partial_{r} \left. \left\{E\left[\frac{D(\Delta Y - g_{0s}-rh)}{p} - \frac{\pi(1-D)(\Delta Y - g_{0s}-rh)}{p(1-\pi)}\right]\right\} \right|_{r=0} \nonumber \\
        &= \left. E\left[-\frac{Dh}{p} + \frac{\pi(1-D)h}{p(1-\pi)}\right] \right|_{r=0} \nonumber \\
        &= \frac{1}{p}E\left[\frac{\pi-D}{1-\pi}(g_{0s}^{\prime} - g_{0s})\right] \nonumber \\
        &\overset{\text{LTE}}{=} \frac{1}{p}E\left[E\left[\frac{\pi-D}{1-\pi}(g_{0s}^{\prime} - g_{0s})  \mid X\right]\right] \nonumber \\
        &= \frac{1}{p}E\left[E\left[\frac{\pi-\pi}{1-\pi}\right](g_{0s}^{\prime}(X) - g_{0s}(X))\right] = 0. \nonumber \\
        \intertext{\emph{Orthogonality with respect to $\pi$:}}
       &\partial_{r} \left. \left\{E[\psi_{\theta_s}(Z;\pi  + r\underbrace{(\pi^{\prime} - \pi)}_h, g_{0s},p)]\right\} \right|_{r=0} \nonumber \\
        &= \partial_{r} \left. \left\{E\left[ - \frac{(\pi + rh)(1-D)(\Delta Y - g_{0s})}{p(1-\pi - rh)}\right]\right\} \right|_{r=0} \nonumber \\
        &= E\left[-\frac{(1-D)(\Delta Y - g_{0s})h}{p(1 - \pi)^2}\right] \nonumber \\
        &\overset{\text{LTE}}{=} E\left[-(1-D)E\left[\frac{(\Delta Y - g_{0s})h}{p(1-\pi)^2} \mid D\right]\right] \nonumber \\
        &= -E\left[\frac{(\Delta Y - g_{0s})h}{p(1-\pi)^2} \mid D=0\right] \times (1-p) \nonumber \\
        &\overset{\text{LTE}}{=} -E\left[\frac{h}{p(1-\pi)^2}E[\Delta Y - g_{0s} \mid X,D=0] \mid D=0\right] \times (1-p) \nonumber \\
        &= 0.  \nonumber \\
        \intertext{\emph{Orthogonality with respect to $p$:}}
        &\partial_{p} \left\{E[\psi_{\theta_s}(Z;\theta_s, g_{0s},\pi,p)]\right\}  \nonumber \\
        &= E\left[-\frac{D(\Delta Y - g_{0s})}{p^2} + \frac{\pi(1-D)(\Delta Y - g_{0s})}{p^2(1-\pi)} + \frac{D\theta_s}{p^2}\right] \nonumber \\
        &= \frac{1}{p^2}E\left[\frac{(\pi - D)(\Delta Y - g_{0s})}{1-\pi} + D\theta_s\right] \nonumber \\
        &\overset{\text{LTE}}{=} \frac{1}{p^2}\left(E\left[-(\Delta Y - g_{0s}) + \theta_s \mid D=1\right]p + E\left[\frac{\pi(\Delta Y - g_{0s})}{1-\pi} \mid D=0\right] \times (1-p)\right) \nonumber \\
        &\overset{\text{LTE}}{=} 0 + 0 = 0. \nonumber
    \end{align}
        \item \textbf{Score for $\sigma_{0s}^2$.}
    \begin{align}
    \intertext{\emph{Moment condition:}}
    &E[\psi_{\sigma_{0s}^2}(Z;g_{0s},p)] \overset{\text{def.}}{=} \sigma_{0s}^2 - \sigma_{0s}^2 = 0. \nonumber \\
    \intertext{\emph{Orthogonality with respect to $g_{0s}$:}}
    &\partial_{r} \left. \left\{E[\psi_{\sigma_{0s}^2}(Z;g_{0s} + r(g_{0s}^{\prime}-g_{0s}),p)]\right\} \right|_{r=0} \nonumber \\
    &=  \partial_{r} \left. \left\{E\left[\frac{1-D}{1-p}(\Delta Y - g_{0s} - r(g_{0s}^{\prime}-g_{0s}))^2 - \frac{1-D}{1-p}\sigma_{0s}^2\right]\right\} \right|_{r=0}\nonumber \\ 
    &= \left. -2E\left[(\Delta Y - g_{0s} - r(g_{0s}^{\prime}-g_{0s}))(g_{0s}^{\prime}-g_{0s})\mid D=0\right]\right|_{r=0}\nonumber \\ 
    &= -2E\left[(\Delta Y - g_{0s})(g_{0s}^{\prime}-g_{0s})\mid D=0\right]\nonumber \\ 
    &\overset{\text{LTE}}{=} -2E[(g_{0s}^{\prime}-g_{0s})E[(\Delta Y - g_{0s})\mid D=0,X]\mid D=0] = 0. \nonumber \\
    \intertext{\emph{Orthogonality with respect to $p$:}}
    &\partial_{p}E[\psi_{\sigma_{0s}^2}(Z;g_{0s},p)] \nonumber \\
    &= \frac{1}{1-p}\left(E[(\Delta Y - g_{0s})^2\mid D=0] - \sigma_{0s}^2\right) \overset{\text{def.}}{=} 0. \nonumber 
\end{align}
    \item \textbf{Score for $\nu_{0s}^2$.}
    \begin{align}
    \intertext{\emph{Moment condition:}}
     &E[\psi_{\nu_{0s}^2}(Z;\pi,p)] \nonumber \\
     &= E\left[2\frac{D}{p}
        \left(\frac{O_X}{O} - \nu_{0s}^2\right)
        - \frac{1-D}{1-p} \times 
        \left[\left(\frac{O_X}{O}\right)^2 - \nu_{0s}^2\right]\right] \nonumber \\
    &= 2E\left[\alpha_{0s}\mid D=1\right] - 2\nu_{0s}^2 - E[\alpha_{0s}^2\mid D=0] + \nu_{0s}^2 \nonumber \\
    &\overset{\text{def.}}{=} 0. \nonumber \\
     \intertext{\emph{Orthogonality with respect to $\pi$:}}
    &\left. \partial_{r}\left\{E[\psi_{\nu_{0s}^2}(Z;\pi+r(\pi^{\prime} - \pi),p)]\right\}\right|_{r=0}\nonumber \\
    &= \left. \partial_{r}\left\{E\left[ 2\frac{D}{p}\frac{\pi+rh}{1-\pi-rh}\frac{1-p}{p} - \frac{1-D}{1-p}\frac{(\pi+rh)^2}{(1-\pi-rh)^2}\frac{(1-p)^2}{p^2}\right]\right\}\right|_{r=0} \nonumber \\ &= \left. E\left[\frac{2D(1-p)}{p^2}\frac{h}{(1-\pi-rh)^2} - \frac{(1-D)(1-p)}{p^2}\frac{2h(\pi+rh)}{(1-\pi-rh)^3}\right] \right|_{r=0}\nonumber  \\ 
    &=  E\left[\frac{2D(1-p)}{p^2}\frac{h}{(1-\pi)^2} - \frac{(1-D)(1-p)}{p^2}\frac{2h\pi}{(1-\pi)^3}\right] \nonumber \\
    &= \frac{2(1-p)}{p^2} \times E\left[\frac{h(D - \pi)}{(1-\pi)^3}\right] \nonumber \\  
    &\overset{\text{LTE}}{=} \frac{2(1-p)}{p^2} \times E\left[h \times E\left[\frac{D - \pi}{(1-\pi)^3}  \mid X\right]\right] \nonumber \\
    &= \frac{2(1-p)}{p^2} \times E\left[h \times \frac{\pi - \pi}{(1-\pi)^3}\right] = 0.  \nonumber \\
    \intertext{\emph{Orthogonality with respect to $p$:}}
    &\partial_{p}E[\psi_{\nu_{0s}^2}(Z;\pi,p)] \nonumber \\
    &= E\left[\frac{-2D}{p^2}\alpha_{0s} - \frac{2D}{p^3}\frac{\pi}{1-\pi}\right] \nonumber \\
    & \phantom{=} \hspace{2cm} - E\left[\frac{1-D}{(1-p)^2}\alpha_{0s}^2 - \frac{1-D}{1-p}\frac{\pi^2}{(1-\pi)^2}\frac{2(1-p)}{p^3}\right] \nonumber \\
    & \phantom{=} \hspace{3cm} - E\left[\frac{-2D}{p^2} - \frac{1-D}{(1-p)^2}\right]\nu_{0s}^2 \nonumber \\
    &= \frac{-2}{p}\nu_{0s}^2 - \frac{1}{1-p}\nu_{0s}^2 + \frac{2}{p}\nu_{0s}^2 + \frac{1}{1-p}\nu_{0s}^2 \nonumber \\
    & \phantom{=} \hspace{3cm} - E\left[ \frac{2D}{p^3}\frac{\pi}{1-\pi}\right] + E\left[\frac{\pi^2}{(1-\pi)^2}\frac{2(1-D)}{p^3}\right] \nonumber \\
    &= E\left[2\frac{\pi(\pi - D)}{p^3(1-\pi)^2}\right] \overset{\text{LTE}}{=} 0. \nonumber \qedhere
    \end{align}
        \end{itemize}
    \end{proof}

\section{Benchmarking analysis}
\label{app:bench}

\subsection{Benchmarking against observed covariates}
\label{app:bench-covariates}

Our approach follows \citet{cinelli2020making, cinelli2025omitted} and \citet{chernozhukov2022long}, extrapolating the strength of unobserved confounding from its relative strength compared to observed covariates. We begin by introducing notation and then define the benchmarking measures for each bias factor.

\subsubsection{Notation}

Suppose we use $X_j$, a subvector of the observed covariates $X$, as the benchmark covariate vector. Let $X_{-j}$ denote the vector of observed covariates other than $X_j$. Correspondingly, we define $g_{0s,-j} := E[\Delta Y \mid X_{-j},D=0]$, $\alpha_{0s,-j} = O_{X_{-j}}/O$, $\theta_{0s,-j} := E[g_{0s,-j}\mid D=1]$, and $\theta_{s,-j} := E[\Delta Y\mid D=1] - \theta_{0s,-j}$. We define the bias factors that quantify the strength of $X_j$ conditional on $X_{-j}$ as follows:
\begin{align*}
     C_{0\Delta Y,j}^2 &:= \eta^2_{\Delta Y \sim X_j \mid X_{-j},D=0}, \\
     C_{0D,j}^2 &:= \frac{1 - R^2_{O_X \sim O_{X_{-j}}\mid D=0}}{R^2_{O_X \sim O_{X_{-j}}\mid D=0}}\\
     &= \frac{E[O_{X}\mid D=1] - E[O_{X_{-j}}\mid D=1]}{E[O_{X_{-j}}\mid D=1]} \\
     &= \frac{\chi^2(P_{X|1}\|P_{X|0}) - \chi^2(P_{X_{-j}|1}\|P_{X_{-j}|0})}{\chi^2(P_{X_{-j}|1}\|P_{X_{-j}|0}) + 1}, \text{ and}\\
     \rho_{0,j} &:= \operatorname{Cor}(g_{0s} - g_{0s,-j}, O_X - O_{X_{-j}}\mid D=0). 
 \end{align*}
 The corresponding scaling factors are $\sigma_{0s,-j}^2 := E[\operatorname{Var}(\Delta Y \mid X_{-j},D=0)\mid D=0]$ and $\nu_{0s,-j}^2 := E[\alpha^2_{0s,-j}\mid D=0] = E[\alpha_{0s,-j}\mid D=1]$. 

\subsubsection{Relative bounds on \texorpdfstring{$1 - R^2_{O_{XU} \sim O_X | D=0}$}{1 - R2(OXU ~ OX | D=0)}}

We now discuss how to express the strength of unobserved confounders in explaining treatment assignment in terms of their relative strength as compared to $X_j$.

\begin{Definition}[Relative strength with treatment selection] 
\label{def:bench_metric}
    \small{
    \begin{align*}
        k_{0D,j} &:= \frac{R^2_{O_{X} \sim O_{X_{-j}}\mid D=0} - R^2_{O_{XU} \sim O_{X_{-j}}\mid D=0}}{1 - R^2_{O_X \sim O_{X_{-j}}\mid D=0}} \tag*{\text{(Residual $R^2$ in Selection Odds)}} \\
        &= \frac{\frac{E[O_{XU}\mid D=1] - E[O_{X_{-j}}\mid D=1]}{E[O_{XU}\mid D=1]} - \frac{E[O_{X}\mid D=1] - E[O_{X_{-j}}\mid D=1]}{E[O_{X}\mid D=1]}}{\frac{E[O_{X}\mid D=1] - E[O_{X_{-j}}\mid D=1]}{E[O_{X}\mid D=1]}} \tag*{\text{(Average Selection Odds)}} \\
        &= \frac{\frac{\chi^2(P_{X,U|1}\|P_{X,U|0}) - \chi^2(P_{X_{-j}|1}\|P_{X_{-j}|0})}{\chi^2(P_{X,U|1}\|P_{X,U|0})+1} - \frac{\chi^2(P_{X|1}\|P_{X|0}) - \chi^2(P_{X_{-j}|1}\|P_{X_{-j}|0})}{\chi^2(P_{X|1}\|P_{X|0})+1}}{\frac{\chi^2(P_{X|1}\|P_{X|0}) - \chi^2(P_{X_{-j}|1}\|P_{X_{-j}|0})}{\chi^2(P_{X|1}\|P_{X|0})+1}} \tag*{\text{(Distributional Imbalance)}}. 
    \end{align*}}
\end{Definition}

Using this definition, we obtain the following reparameterization of the sensitivity parameter $1 - R^2_{O_{XU} \sim O_X\mid D=0}$.

\begin{Proposition}[Reparameterization of $1 - R^2_{O_{XU} \sim O_X\mid D=0}$ in terms of relative strength]
    \begin{align*}
        1 - R^2_{O_{XU} \sim O_X\mid D=0} &= k_{0D,j} \times C_{0D,j}^2, \text{ and } C_{0D}^2 = \frac{1 - R^2_{O_{XU} \sim O_X\mid D=0}}{R^2_{O_{XU} \sim O_X\mid D=0}},
    \end{align*}
    where $C_{0D,j}^2 := \frac{1 - R^2_{O_X \sim O_{X_{-j}}\mid D=0}}{R^2_{O_X \sim O_{X_{-j}}\mid D=0}}$. 
\end{Proposition}

\begin{proof}
    \begin{align*}
        k_{0D,j} \times C_{0D,j}^2 &= \frac{R^2_{O_{X} \sim O_{X_{-j}}\mid D=0} - R^2_{O_{XU} \sim O_{X_{-j}}\mid D=0}}{1 - R^2_{O_X \sim O_{X_{-j}}\mid D=0}} \times \frac{1 - R^2_{O_X \sim O_{X_{-j}}\mid D=0}}{R^2_{O_X \sim O_{X_{-j}}\mid D=0}} \\
        &= 1 - \frac{R^2_{O_{XU} \sim O_{X_{-j}}\mid D=0}}{R^2_{O_X \sim O_{X_{-j}}\mid D=0}} = 1 - R^2_{O_{XU} \sim O_X\mid D=0}.  \qedhere
    \end{align*}
\end{proof}

\begin{remark}
Note that, since $1 - R^2_{O_{XU} \sim O_X \mid D=0} \le 1$, we must have $k_{0D,j} \le 1/C_{0D,j}^2$. Also, $R^2_{O_X \sim O_{X_{-j}}\mid D=0} = \frac{E\left[O_{X_{-j}}^2 \mid D=0\right]}{E\left[O_{X}^2 \mid D=0\right]} = \frac{\nu_{0s,-j}^2}{\nu_{0s}^2}$, so an alternative expression is
\begin{align*}
        1 - R^2_{O_{XU} \sim O_X\mid D=0} &=  k_{0D,j} \times \frac{\nu_{0s}^2 - \nu_{0s,-j}^2}{\nu_{0s,-j}^2}.
\end{align*}
This latter form is useful for estimation and inference.
\end{remark}

\subsubsection{Relative bounds on \texorpdfstring{$\eta^2_{\Delta Y \sim U \mid X,D=0}$}{eta2}}

We now discuss how to express the strength of unobserved confounders in explaining outcome evolution in terms of their relative strength as compared to $X_j$.

\begin{Definition}[Relative strength with outcome trend]
    \begin{align*}
        k_{0\Delta Y,j} := \frac{ \eta^2_{\Delta Y \sim U,X_j \mid X_{-j},D=0} -  \eta^2_{\Delta Y \sim X_j \mid X_{-j},D=0}}{\eta^2_{\Delta Y \sim X_j \mid X_{-j},D=0}}.
    \end{align*}
\end{Definition}
Using this definition, we obtain the following reparameterization of the sensitivity parameter $\eta^2_{\Delta Y \sim U \mid X,D=0}$.
\begin{Proposition}[Reparameterization of $\eta^2_{\Delta Y \sim U \mid X,D=0}$ in terms of relative strength]
    \begin{align*}
        \eta^2_{\Delta Y \sim U \mid X,D=0} = k_{0\Delta Y,j} \times \frac{C_{0\Delta Y,j}^2}{1 - C_{0\Delta Y,j}^2},
    \end{align*}
    where $C_{0\Delta Y,j}^2 := \eta^2_{\Delta Y \sim X_j \mid X_{-j},D=0}$
\end{Proposition}
\begin{proof}
    \begin{align*}
        k_{0\Delta Y,j} \times \frac{C_{0\Delta Y,j}^2}{1 - C_{0\Delta Y,j}^2} &= \frac{\eta^2_{\Delta Y \sim U,X_j \mid X_{-j},D=0} -  \eta^2_{\Delta Y \sim X_j \mid X_{-j},D=0}}{\eta^2_{\Delta Y \sim X_j \mid X_{-j},D=0}} \times \frac{\eta^2_{\Delta Y \sim X_j \mid X_{-j},D=0}}{1-\eta^2_{\Delta Y \sim X_j \mid X_{-j},D=0}} \\
        &= \frac{\eta^2_{\Delta Y \sim U,X_j \mid X_{-j},D=0} -  \eta^2_{\Delta Y \sim X_j \mid X_{-j},D=0}}{1-\eta^2_{\Delta Y \sim X_j \mid X_{-j},D=0}} \\
        &= \eta^2_{\Delta Y \sim U \mid X,D=0}. \qedhere
    \end{align*}
\end{proof}

\begin{remark}
Note that, since $\eta^2_{\Delta Y \sim U \mid X,D=0}\leq 1$, we must have $k_{0\Delta Y,j} \le \frac{1-C_{0\Delta Y,j}^2}{C_{0\Delta Y,j}^2} $. Also note that
\begin{align*}
\sigma_{0s,-j}^2 - \sigma_{0s}^2 &= E[(\Delta Y - g_{0s,-j})^2\mid D=0] - E[(\Delta Y - g_{0s})^2\mid D=0] \\
&\overset{\text{LTE}}{=} E[(g_{0s} - g_{0s,-j})^2\mid D=0], 
\end{align*}
which leads to the following alternative representation
\begin{align*}
\eta^2_{\Delta Y \sim U \mid X,D=0} &= k_{0\Delta Y,j} \times \frac{\eta^2_{\Delta Y \sim X_j \mid X_{-j},D=0}}{1-\eta^2_{\Delta Y \sim X_j \mid X_{-j},D=0}} \\
&= k_{0\Delta Y,j} \times \frac{E[(g_{0s} - g_{0s,-j})^2\mid D=0]}{E[(\Delta Y - g_{0s,-j})^2\mid D=0] - E[(g_{0s} - g_{0s,-j})^2\mid D=0]} \\
&\overset{\text{LTE}}{=} k_{0\Delta Y,j} \times \frac{E[(g_{0s} - g_{0s,-j})^2\mid D=0]}{E[(\Delta Y - g_{0s})^2\mid D=0]} \\
&= k_{0\Delta Y,j} \times \frac{\sigma_{0s,-j}^2 - \sigma_{0s}^2}{\sigma_{0s}^2}.
\end{align*}
This latter form is useful for estimation and inference.
\end{remark}

\subsubsection{Benchmarking \texorpdfstring{$|\rho_0|$}{|rho0|}} \label{sec:rho-benchmark}

Following \citet{chernozhukov2022long}, we propose extrapolating $\rho_0$ using the observed alignment of $X_j$ as given by $\rho_{0,j} := \operatorname{Cor}(g_{0s} - g_{0s,-j}, O_X - O_{X_{-j}}\mid D=0)$.  Note that the bias decomposition with respect to $X_j$ gives
\begin{align*}
    \theta_s - \theta_{s,-j} 
    = - \rho_{0,j} 
    \sqrt{
    (\sigma^2_{0s,-j}-\sigma^2_{0s})(\nu^2_{0s} - \nu^2_{0s,-j})
    }.
\end{align*}
This leads to the following debiased representation of $\rho_{0,j}$,
\begin{align*}
    \rho_{0,j} = \frac{-(\theta_{s} - \theta_{s,-j}) }{\sqrt{(\sigma^2_{0s,-j}-\sigma^2_{0s})(\nu^2_{0s} - \nu^2_{0s,-j})}}.
\end{align*}
This latter form is useful for estimation and inference, which we discuss below.

\subsubsection{Statistical inference for benchmark components} \label{sec:Statistical Inference of benchmark components}

Statistical inference for our benchmarking analysis follows the procedure introduced in Appendices E.5 and E.6 of \citet{chernozhukov2022long}. Specifically, we define the estimable ``gain" metrics with the debiased representations derived above,
\begin{align*}
    &G_{0D,j} := C_{0D,j}^2 = \frac{\nu_{0s}^2 - \nu_{0s,-j}^2}{\nu_{0s,-j}^2}, \\
    &G_{0\Delta Y,j} := \frac{C_{0\Delta Y,j}^2}{1 - C_{0\Delta Y,j}^2} = \frac{\sigma_{0s,-j}^2 - \sigma_{0s}^2}{\sigma_{0s}^2}, \text{ and }\\
    &\rho_{0,j} = \frac{-(\theta_{s} - \theta_{s,-j})}{\sqrt{(\sigma_{0s,-j}^2 - \sigma_{0s}^2) \times (\nu_{0s}^2 - \nu_{0s,-j}^2)}}. 
\end{align*}
For inference on the benchmarking bounds, we assume $G_{0\Delta Y,j}>0$, $G_{0D,j}>0$, $\rho_{0,j}\neq 0$, and $0<k_{0D,j}G_{0D,j}<1$.
We have shown that $1 - R^2_{O_{XU} \sim O_X\mid D=0} = k_{0D,j} \times G_{0D,j}$ and $\eta^2_{\Delta Y \sim U \mid X,D=0} = k_{0\Delta Y,j} \times G_{0\Delta Y,j}$. Therefore, given the influence functions of $\theta_s$, $\theta_{s,-j}$, $\sigma_{0s}^2$, $\sigma_{0s,-j}^2$, $\nu_{0s}^2$, and $\nu_{0s,-j}^2$ obtained from DML, we apply the delta method to obtain the following debiased influence functions for the ``gain" metrics.
\begin{align*}
    \varphi^0_{G_{0D,j}}(Z) &= \frac{\nu_{0s,-j}^2\varphi^0_{\nu_{0s}^2}(Z) - \nu_{0s}^2\varphi^0_{\nu_{0s,-j}^2}(Z)}{\nu_{0s,-j}^4}, \\
    \varphi^0_{G_{0\Delta Y,j}}(Z) &= \frac{\sigma_{0s}^2\varphi^0_{\sigma_{0s,-j}^2}(Z) - \sigma_{0s,-j}^2\varphi^0_{\sigma_{0s}^2}(Z)}{\sigma_{0s}^4}, \text{ and } \\
    \varphi^0_{\rho_{0,j}}(Z) &= \frac{\varphi^0_{\theta_{s,-j}}(Z) - \varphi^0_{\theta_{s}}(Z)}{(\sigma_{0s,-j}^2 - \sigma_{0s}^2)^{1/2} (\nu_{0s}^2 - \nu_{0s,-j}^2)^{1/2}} - \frac{(\theta_{s,-j}-\theta_s)(\varphi^0_{\sigma_{0s,-j}^2}(Z) - \varphi^0_{\sigma_{0s}^2}(Z))}{2(\sigma_{0s,-j}^2 - \sigma_{0s}^2)^{3/2} (\nu_{0s}^2 - \nu_{0s,-j}^2)^{1/2}} \nonumber \\
    & \phantom{=} \hspace{3.5cm} - \frac{(\theta_{s,-j}-\theta_s)(\varphi^0_{\nu_{0s}^2}(Z) - \varphi^0_{\nu_{0s,-j}^2}(Z))}{2(\sigma_{0s,-j}^2 - \sigma_{0s}^2)^{1/2} (\nu_{0s}^2 - \nu_{0s,-j}^2)^{3/2}}.
\end{align*}
The plug-in estimator of the bias and the corresponding influence function for $\theta_{\pm}$ follow directly from the discussion in Appendix E.6 of \citet{chernozhukov2022long}, after replacing the ``unconditional” components with their corresponding ``conditional” counterparts in our results. 

\subsection{Benchmarking against pre-trends}
\label{app:bench-pre-trend}

Here we let $Z := (\Delta Y,D,X, \Delta Y^{\text{pre}},X^{\text{pre}})$. In this section we provide the deferred influence functions for the pre-trend extrapolation bounds. The first approach yields the following bounds on the target estimand in the post-treatment period:
\begin{align*}
\theta_{\pm} = \theta_s \pm k\cdot |\theta_s^{\text{pre}}|,
\end{align*}
where $\theta_s$ and $\theta_s^{\text{pre}}$ are estimable from the data, and $\theta_s^{\text{pre}} \neq 0$.  Under standard DML conditions, by applying the delta method, the plug-in estimator $\widehat{\theta}_{\pm}$ is asymptotically linear and Gaussian with influence function
\begin{align*}
    \varphi^{0}_{\theta_{\pm}}(Z) := \varphi^{0}_{\theta_s}(Z) \pm k\cdot \text{sign}(\theta^{\text{pre}}_{s})\cdot\varphi^{0}_{\theta^{\text{pre}}_s}(Z).
\end{align*}

As for the second approach, recall it yields the bounds:
\begin{align*}
\theta_{\pm} &= \theta_s \pm k\cdot \left(\frac{S_0}{S_0^{\text{pre}}}\right) \cdot |\theta_s^{\text{pre}}|.
\end{align*}
For this result, we additionally assume $S_0>0$ and $S_0^{\text{pre}}>0$.
The influence function for $S_0$ is:
\begin{align*}
\varphi^{0}_{S_0}(Z) &= 
\frac{1}{2 S_0}\left(\sigma_{0s}^2\, \varphi^0_{\nu_{0s}^2}(Z) + \nu_{0s}^2\, \varphi^0_{\sigma_{0s}^2}(Z)\right),
\end{align*}
with a similar influence function for $S_0^{\text{pre}}$. By the delta method, the plug-in estimator $\widehat{\theta}_{\pm}$ is then asymptotically linear and Gaussian with influence function
\begin{align*}
    \varphi^{0}_{\theta_{\pm}}(Z) = \varphi^{0}_{\theta_s}(Z) \pm k\, |\theta_s^{\text{pre}}|\frac{ S_0}{S_0^{\text{pre}}} \left( \frac{\varphi^{0}_{\theta_s^{\text{pre}}}(Z)}{\theta_s^{\text{pre}}} + \frac{\varphi^{0}_{S_0}(Z)}{S_0} - \frac{\varphi^{0}_{S_0^{\text{pre}}}(Z)}{S_0^{\text{pre}}} \right).
\end{align*}

\section{Sensitivity statistics}
\label{app:rv}

\subsection{Compatible inferences given bounds on sensitivity parameters}

Given fixed values of the sensitivity parameters, Theorem~\ref{thm: inference} determines how inference changes under that particular confounding scenario. Often, however, the analyst wishes to consider all inferences compatible with upper bounds on these parameters. Let
\[
\text{CI}^{\max}_{1-\alpha,\bar\rho^2,\bar R^2_{\Delta Y},\bar R^2_D}(\theta)
\]
denote the interval formed by the most extreme lower and upper confidence limits constructed from Theorem~\ref{thm: inference}. To simplify notation, we denote $\eta^2_{\Delta Y\sim U\mid X,D=0}$ and $1-R^2_{O_{XU}\sim O_X\mid D=0}$ by $\eta^2$ and $1-R^2$, respectively. Formally,
\begin{align*}
&\text{CI}^{\max}_{1-\alpha,\bar\rho^2,\bar R^2_{\Delta Y},\bar R^2_D}(\theta)
= [L^{\min},U^{\max}], \quad\text{where}
\\
&L^{\min}
= \min_{\rho_0^2,\eta^2,R^2}
\left(
\widehat{\theta}_{-}(\rho_0^2,\eta^2,R^2)
- \Phi^{-1}(1-\alpha)
\sqrt{\frac{E[(\varphi^0_{\theta_-}(Z;\rho_0^2,\eta^2,R^2))^2]}{n}}
\right),
\\
&\hspace{1cm}\text{s.t. }\quad
\rho_0^2\leq\bar\rho^2,\qquad
\eta^2\leq\bar R^2_{\Delta Y},\qquad
1-R^2\leq\bar R^2_D,
\\
&U^{\max}
= \max_{\rho_0^2,\eta^2,R^2}
\left(
\widehat{\theta}_{+}(\rho_0^2,\eta^2,R^2)
+ \Phi^{-1}(1-\alpha)
\sqrt{\frac{E[(\varphi^0_{\theta_+}(Z;\rho_0^2,\eta^2,R^2))^2]}{n}}
\right),
\\
&\hspace{1cm}\text{s.t. }\quad
\rho_0^2\leq\bar\rho^2,\qquad
\eta^2\leq\bar R^2_{\Delta Y},\qquad
1-R^2\leq\bar R^2_D.
\end{align*}
Thus, $\text{CI}^{\max}_{1-\alpha,\bar\rho^2,\bar R^2_{\Delta Y},\bar R^2_D}(\theta)$ collects the confidence intervals compatible with the stated restrictions. When $\bar\rho^2=1$, we omit it from the subscript.

\begin{Corollary}
\label{cor:ci-max}
Under the conditions of Theorem~\ref{thm: inference}, the endpoints of $\text{CI}^{\max}_{1-\alpha,\bar\rho^2,\bar R^2_{\Delta Y},\bar R^2_D}(\theta)$ have asymptotic one-sided coverage of at least $1-\alpha$ whenever the stated restrictions hold, and the interval is nested as any of the three upper bounds increases.
\end{Corollary}

\begin{proof}
Let $[l_0,u_0]$ denote the confidence interval from Theorem~\ref{thm: inference} evaluated at the true sensitivity parameters. When the stated restrictions hold, the true sensitivity parameters are among the configurations used to construct $\text{CI}^{\max}_{1-\alpha,\bar\rho^2,\bar R^2_{\Delta Y},\bar R^2_D}(\theta)$. Therefore, $L^{\min}\leq l_0$ and $U^{\max}\geq u_0$. It follows that whenever $l_0$ covers $\theta_-$ from below, $L^{\min}$ does as well; similarly, whenever $u_0$ covers $\theta_+$ from above, $U^{\max}$ does as well. Theorem~\ref{thm: inference} shows that each of these events has asymptotic probability $1-\alpha$, which proves the coverage result.

For nesting, increasing any of the three upper bounds expands the collection of confounding configurations over which the confidence limits are optimized. The minimum lower confidence limit can therefore only decrease, while the maximum upper confidence limit can only increase. This proves the nesting result.
\end{proof}

\subsection{Robustness Values}

\subsubsection{Definitions}

Here we provide a more detailed definition of the robustness value and extreme robustness value. 

\begin{Definition}[Robustness Value (RV)]
For fixed $\theta^*$ and $\alpha\in(0,0.5)$, if the two sided $1-\alpha$ confidence interval from the original analysis under no confounding contains $\theta^*$, we define
\[
\text{RV}_{\theta^*,\alpha}(\theta):=0.
\]
Otherwise, the robustness value is defined as
\begin{align}
\text{RV}_{\theta^*,\alpha}(\theta)
:=
\inf\left\{
\text{RV}:
\theta^*\in
\text{CI}^{\max}_{1-\alpha,\text{RV},\text{RV}}(\theta)
\right\}.
\nonumber
\end{align}
\end{Definition}

\begin{Definition}[Extreme Robustness Value (XRV)]
For fixed $\theta^*$ and $\alpha\in(0,0.5)$, if the two sided $1-\alpha$ confidence interval from the original analysis under no confounding contains $\theta^*$, we define
\[
\text{XRV}_{\theta^*,\alpha}(\theta):=0.
\]
Otherwise, the extreme robustness value is defined as
\begin{align}
\text{XRV}_{\theta^*,\alpha}(\theta)
:=
\inf\left\{
\text{XRV}:
\theta^*\in
\text{CI}^{\max}_{1-\alpha,1,\text{XRV}}(\theta)
\right\}.
\nonumber
\end{align}
\end{Definition}

The (extreme) robustness values of the point estimate are defined by setting the critical value to zero. We denoted them by
$\text{RV}_{\theta^*,\alpha=1}(\theta)$ and
$\text{XRV}_{\theta^*,\alpha=1}(\theta)$.

Notice that the robustness value and extreme robustness value are defined as descriptive sensitivity statistics. They communicate the minimum strength of confounding that one must be willing to postulate for the confidence interval of Theorem~\ref{thm: inference} to include a null value $\theta^*$ of interest. In addition, they also admit the following inferential interpretation.

\subsubsection{Inferential Properties}

We denote the corresponding population values $\text{RV}_{\theta^*}(\theta)$ and $\text{XRV}_{\theta^*}(\theta)$, as the values obtained by replacing the confidence limits in the definitions above with the corresponding population bias bounds. Following arguments similar to those in \citet{cinelli2020making,cinelli2025omitted}, let $f_{\theta^*} := |\theta^* - \theta_s|/S_0$. The population robustness values have the following analytical formulas:
\begin{align*}
\text{RV}_{\theta^*}(\theta) &= \frac{1}{2}\left(\sqrt{f_{\theta^*}^4 + 4f_{\theta^*}^2} - f_{\theta^*}^2\right), \qquad \text{and} \qquad
\text{XRV}_{\theta^*}(\theta) = \frac{f_{\theta^*}^2}{1 + f^2_{\theta^{*}}}.
\end{align*}

The following result shows that the empirical robustness values can be interpreted as lower confidence bounds for their population counterparts.

\begin{Corollary}
\label{cor:rv-lower-confidence}
Let $f_{\theta^*}>0$. Under the conditions of Theorem~\ref{thm: inference},  $\text{RV}_{\theta^*,\alpha}(\theta)$ and $\text{XRV}_{\theta^*,\alpha}(\theta)$ are asymptotically valid $1-\alpha$ lower confidence bounds for their population counterparts.
\end{Corollary}

\begin{proof}
We first consider the RV. Let $r=\text{RV}_{\theta^*}(\theta)$ and suppose that $\theta^*<\theta_s$. By definition, $r$ is the minimum common bound on trend and selection strength for which the lower population bias bound reaches $\theta^*$. Therefore,
\[
\theta^*=\theta_s-S_0\frac{r}{\sqrt{1-r}}.
\]
Equivalently, $\theta^*$ is the lower population endpoint when
\[
(\rho_0^2,\eta^2,1-R^2)=(1,r,r).
\]
Now consider $\text{CI}^{\max}_{1-\alpha,r,r}(\theta)$, which allows these same bounds on trend and selection strength while leaving alignment unrestricted. Corollary~\ref{cor:ci-max} implies that its lower endpoint is no greater than $\theta^*$ with asymptotic probability at least $1-\alpha$. Its upper endpoint exceeds $\theta^*$ with probability tending to one because $\theta^*<\theta_s$. Hence, $\theta^*\in\text{CI}^{\max}_{1-\alpha,r,r}(\theta)$ with asymptotic probability at least $1-\alpha$.

Whenever this event occurs, the definition of $\text{RV}_{\theta^*,\alpha}(\theta)$ implies
\[
\text{RV}_{\theta^*,\alpha}(\theta)
\leq r
=
\text{RV}_{\theta^*}(\theta).
\]
Thus, $\text{RV}_{\theta^*,\alpha}(\theta)$ is an asymptotically valid $1-\alpha$ lower confidence bound for $\text{RV}_{\theta^*}(\theta)$. The case $\theta^*>\theta_s$ follows by applying the same argument to the upper endpoint.

For the XRV, let $r=\text{XRV}_{\theta^*}(\theta)$. Because confounding with the untreated outcome trend is unrestricted, the relevant population configuration is
\[
(\rho_0^2,\eta^2,1-R^2)=(1,1,r).
\]
Applying the same argument to $\text{CI}^{\max}_{1-\alpha,1,r}(\theta)$ gives
\[
\liminf_{n\to\infty}
P\!\left(
\text{XRV}_{\theta^*,\alpha}(\theta)
\leq
\text{XRV}_{\theta^*}(\theta)
\right)
\geq 1-\alpha,
\]
which proves the XRV result.
\end{proof}

Moreover, note that the two coverage events coincide:
\[
\text{RV}_{\theta^*,\alpha}(\theta)
\leq
\text{RV}_{\theta^*}(\theta)
\Longleftrightarrow
\text{XRV}_{\theta^*,\alpha}(\theta)
\leq
\text{XRV}_{\theta^*}(\theta)
\]
Therefore, all simulation results in this appendix report a single common (X)RV empirical coverage.

\subsubsection{Interpretation in terms of increase in average odds or covariate imbalance}

For $0 \le k < 1$, note that
\begin{align*}
    1 - R^2_{O_{XU} \sim O_{X}\mid D=0} = k 
    \iff  C^2_{0D} = \frac{k}{1-k}.
\end{align*}
Then, 
\begin{itemize}
    \item If $k$ is the value of $\text{XRV}_{0,\alpha}$: if unobserved confounders increase treatment odds or covariate imbalance by less than $\frac{k}{1-k}$ then such confounders are not capable of overturning the original results, at the significance level of $\alpha$, \emph{regardless} of how much variation such confounders explain of the outcome trend.
    \item If $k$ is the value of $\text{RV}_{0,\alpha}(\theta)$: unobserved confounders explaining less than $k$ of the residual variation in untreated trend and inducing less than $\frac{k}{1-k}$ increase in average odds or covariate imbalance cannot overturn the conclusions of the study, at the significance level of $\alpha$.
\end{itemize}

\section{Simulations}
\label{app:simulations}

This section discusses the use of Monte Carlo simulations to examine the finite sample properties of the proposed OVB results in the canonical DiD setting. Our data simulation strategy allows us to derive explicit expressions for the bias factors, as well as to identify the correctly specified models for estimating the nuisance parameters. This enables us to attribute the source of bias solely to the omission of covariates, which leads to the violation of the conditional PTA. 

    \subsection{Data generating process}\label{sec:dgp}

    The data are generated according to the following steps:

    \noindent
    \textbf{Step 1:} Given $p \ge \epsilon$, let $D_i \overset{i.i.d.}{\sim} \text{Bernoulli }(p)$. Then, generate the covariates as follows,
        \begin{align}
            &(X_i,U_i)\mid D_i = d \overset{i.i.d.}{\sim} \mathcal{N}(\bm{\mu}_d, \Sigma_d), \text{ with} \nonumber \\
            &\bm{\mu}_d = (\mu_{dX}, \mu_{dU}) \text{ and } \Sigma_d = \text{diag}(\sigma_{dX}^2, \sigma_{dU}^2), \text{ for } d \in \{0,1\}. \nonumber 
        \end{align}

    \noindent
    \textbf{Step 2:} For each unit $i$, calculate the potential outcomes at time period $t \in \{1,2\}$ using 
        \begin{align}
        &Y_{i,t}(0) = \alpha_i + \rho_t + X_i \beta_{xt} + U_i\beta_{ut} +\epsilon_{i,t}, \nonumber \\
        &Y_{i,t}(1) = Y_{i,t}(0) + \mathds{1}\{t \ge 2\} \cdot \theta + (\nu_{i,t} - \epsilon_{i,t}), \nonumber
        \end{align} where $\rho_t = \beta_{xt} = \beta_{ut} = t \text{,} \ \alpha_i\mid D_i=d \overset{i.i.d.}{\sim} \mathcal{N}(2d, 1)$, and $\epsilon_{i,t}, \nu_{i,t} \overset{i.i.d.}{\sim} \mathcal{N}(0,1)$. 

    \noindent
    \textbf{Step 3:} Calculate the observed outcomes using consistency, which leads to the following model for outcome evolution:
    \begin{align}
    \Delta Y_i &= (\Delta\rho + X_i\Delta \beta_{x} + U_i\Delta\beta_{u}) + \theta D_i + \Delta \tilde{\epsilon}_{i}, \nonumber
    \end{align}
    where $\Delta\rho = \Delta \beta_{x} = \Delta\beta_{u} = 1$, and $\Delta \tilde{\epsilon}_{i} = D_i\Delta \nu_i + (1-D_i)\Delta\epsilon_i\overset{i.i.d.}{\sim} \mathcal{N}(0,2)$. 

    The dependence of covariates $(X_i,U_i)$ on $D_i$ reflects selection into treatment. The covariates represent pre-treatment characteristics and are fixed over time.
    
    In this simulation setting, the propensity score model is correctly specified as a logistic regression with quadratic terms in the covariates for both the ``long” and ``short” estimations. From Step 3, the outcome evolution model is also correctly specified as a linear regression for both estimations. The true bias factors admit the following closed-form expressions:
    \begin{align}
        C_{0\Delta Y}^2 &= \frac{\sigma_{0U}^2}{\sigma_{0U}^2 + 2}, \quad C_{0D}^2 = \chi^2(P_{U|1}\|P_{U|0}), \text{ and} \nonumber \\
        \rho_0^2 &= \frac{(\mu_{1U} - \mu_{0U})^2}{\sigma_{0U}^2 \times (\chi^2(P_{X|1}\|P_{X|0})+1)\chi^2(P_{U|1}\|P_{U|0})}, \nonumber 
    \end{align}
    where $\Delta \theta_s := \theta - \theta_s = -(\mu_{1U} - \mu_{0U})$, and 
    \begin{align*}
        \chi^2(P_{U|1}\|P_{U|0}) &= \frac{\sigma_{0U}^2}{\sigma_{1U}\sqrt{2\sigma_{0U}^2 - \sigma_{1U}^2}}\exp\left(\frac{(\mu_{1U} - \mu_{0U})^2}{2\sigma_{0U}^2 - \sigma_{1U}^2}\right) - 1, \\
        \chi^2(P_{X|1}\|P_{X|0}) &= \frac{\sigma_{0X}^2}{\sigma_{1X}\sqrt{2\sigma_{0X}^2 - \sigma_{1X}^2}}\exp\left(\frac{(\mu_{1X} - \mu_{0X})^2}{2\sigma_{0X}^2 - \sigma_{1X}^2}\right) - 1.
    \end{align*}
    These divergences are finite if and only if $2\sigma_{0X}^2 > \sigma_{1X}^2$ and $2\sigma_{0U}^2 > \sigma_{1U}^2$. Intuitively, when the treated distribution has substantially heavier tails, $\pi_{XU}$ approaches one in the tails, undermining overlap and causing the density ratio to grow explosively, leading to infinite divergence.

\begin{remark}
By the Riesz-Frechet representation theorem, the linear functional $\theta_0(g_0)$ is continuous on $L^2(P_{(X,U) \mid D=0})$ if and only if there exists a unique representer $\alpha_0 \in \Gamma$ such that $\theta_0(g_0) = E[g_0\alpha_0 \mid D =0]$. Lemma~\ref{lemma:cond_RR_main} shows that this representer is $\alpha_0 = O_{XU}/O$. Therefore, continuity of $\theta_0$ requires that $E[\alpha_0^2\mid D=0] < \infty$.  Moreover, as shown in Corollary~\ref{thm:alt-selection}, $\chi^2(P_{X,U|1}\|P_{X,U|0}) = E[\alpha_0^2\mid D=0] - 1$, so in our simulation setting, requiring $2\sigma_{0X}^2 > \sigma_{1X}^2$ and $2\sigma_{0U}^2 > \sigma_{1U}^2$ ensures $\chi^2(P_{X,U|1}\|P_{X,U|0}) < \infty$ and equivalently $E[\alpha_0^2\mid D=0] < \infty$. This condition is automatically satisfied under the usual strong overlap assumption (Assumption~\ref{assump:Overlap}), which enforces uniform boundedness of the propensity score and therefore implies square-integrability. Formally,  
\begin{align*}
\alpha_0 &= \frac{O_{XU}}{O} = \frac{\pi_{XU}}{1 - \pi_{XU}} \times \frac{1-p}{p}
\le \left(\frac{1-\epsilon}{\epsilon}\right)^2 \text{ by Assumption~\ref{assump:Overlap},} \\
\implies \quad &\chi^2(P_{X,U|1}\|P_{X,U|0}) + 1 = E[\alpha_0^2 \mid D = 0] < \infty.
\end{align*}
\end{remark}

\subsection{Results} \label{app:sim-results}
We evaluate the performance of our proposed approach from the following perspectives:

\paragraph{Coverage of Confidence Bounds.} We evaluate the coverage of the confidence bound for $\theta$, after plugging in the true bias factors.

We set $\theta = 2$, $\bm{\mu}_0 = (0.3,0.3)$, $\bm{\mu}_1 = (0,0)$, $(\sigma_{0X}^2, \sigma_{0U}^2) = (6,6)$, and $(\sigma_{1X}^2, \sigma_{1U}^2) = (3,3)$. We consider scenarios with $p \in \{0.2,0.5,0.8\}$ and sample sizes $n \in \{500, 2,000, 5,000\}$, at significance levels $\alpha \in \{0.01, 0.05, 0.1\}$. The nuisance parameters are estimated using the correctly specified model with ten-fold cross-fitting.\footnote{To ensure overlap, we set $\sigma_{0\cdot}^2$ to be sufficiently larger than $\sigma_{1\cdot}^2$. We use more than the default five cross-fitting folds to stabilize estimation of $\nu_{0s}^2$.} We report results from $B = 5,000$ repetitions.
        
The results are presented in Table~\ref{tab: sim_coverage} and Figure~\ref{fig: sim_coverage}, which show that empirical coverage closely tracks nominal levels, indicating well-calibrated confidence bounds with stable finite sample performance.
        
\paragraph{Sensitivity Statistics.} Following the same scenarios as in our first coverage experiment, we compare the empirical values of the sensitivity statistics with their population values.
        
Figure~\ref{fig:ss_point} reports the finite sample behavior of $\text{RV}_{\theta^*=0,\alpha=1}$ and $\text{XRV}_{\theta^*=0,\alpha=1}$ when viewed as estimators of their population analogs. We run $5{,}000$ repetitions for various values of sample size $n$ and treatment probability $p$. The black dot represents the mean and the bars the 2.5\% and 97.5\% quantiles. The plot shows that the mean of the empirical (X)RV closely track their population values, with the sampling distribution becoming more concentrated as the sample size increases.

Next we asses the finite sample behavior of $\text{RV}_{\theta^*=0,\alpha}$ and $\text{XRV}_{\theta^*=0,\alpha}$ for  $\alpha \in \{0.01,0.05,0.1\}$ when viwed as a lower limit confidence bound for their population counterparts. That is, Figure~\ref{fig:ss_CI} reports the empirical probabilities
\[
{P}_n\!\left({\text{RV}}_{\theta^*=0,\alpha} \leq \text{RV}_{\theta^*=0}\right)
\quad\text{and}\quad
{P}_n\!\left({\text{XRV}}_{\theta^*=0,\alpha} \leq \text{XRV}_{\theta^*=0}\right).
\]
The red dashed line marks the nominal coverage level $1-\alpha$, and the shaded region shows the expected (95\%) Monte Carlo error under nominal coverage (for $5{,}000$ repetitions). The figure shows that the empirical coverage closely tracks the nominal levels across the simulation scenarios.

\paragraph{Model Misspecification.} In the presence of model misspecification, we evaluate the performance of our proposed approach by following the pipeline of \citet{chernozhukov2018double, chernozhukov2024applied} for first-stage machine learning estimation of nuisance parameters, which we refer to as the ``nonparametric" model. Detailed information on the learners is summarized in Table~\ref{tab:miss_learners}. We then compare these results with those obtained under parametric nuisance estimation, namely linear regression for the outcome evolution model and logistic regression for the propensity score model, which we refer to as the ``parametric” model hereafter.

Specifically, we maintain the DGP described in Section~\ref{sec:dgp}, but fit the short nuisance models using $X^*=\exp(X/2)$ in place of $X$, thereby inducing misspecification of the parametric models. The remaining design follows the first coverage experiment, so the true bias factors remain unchanged. 

Figures~\ref{fig:miss_coverage} and \ref{fig:width_diff_miss} show that confidence bounds obtained from the parametric model tend to be wider and, consequently, exhibit over-coverage relative to those from the nonparametric model. In our simulation, we focus on the upper confidence bound, which corresponds to the direction of bias relevant for overturning the conclusion. Accordingly, Figures~\ref{fig:zoom_miss} and~\ref{fig:upper_miss} show that the upper bounds from the nonparametric model tend to lie closer to the true effect than those from the parametric model. These findings are consistent with our expectations: under model misspecification, the nonparametric model better approximates the true nuisance functions, leading to less biased first-stage estimates and, consequently, tighter confidence bounds with coverage closer to the nominal level and upper bounds closer to the true effect.

In addition, Figure~\ref{fig:zoom_miss} shows that in some repetitions the parametric model produces excessively wide confidence bounds. This occurs because its estimates of $\nu_{0s}^2$ are highly sensitive to the estimated propensity scores, which can exhibit greater dispersion under model misspecification. By contrast, the nonparametric model is less sensitive to extreme propensity score values.

Finally, Figures~\ref{fig:ss_miss_point} and~\ref{fig:ss_miss_coverage} compare, under parametric and nonparametric nuisance estimation, the empirical values of $\text{RV}_{\theta^*=0,\alpha=1}$ and $\text{XRV}_{\theta^*=0,\alpha=1}$ with their population analogs and the empirical coverage of $\text{RV}_{\theta^*=0,\alpha}$ and $\text{XRV}_{\theta^*=0,\alpha}$ with the nominal levels $1-\alpha$. Overall, the sensitivity statistics estimated from the nonparametric model are more tightly concentrated and are closer to the true values than those from the parametric model. This indicates improved finite sample stability and accuracy of the nonparametric model under model misspecification.

\subsection{Deferred derivations}

We first verify correct specification of the propensity score model by deriving the expression for the propensity score and showing that it takes the form of a logistic regression with quadratic terms.
\begin{enumerate}
\item[(a)] \textbf{Short:}
\begin{align*}
\frac{P(X=x|D=0)}{P(X=x|D=1)} 
&= \frac{\prod_{i=1}^{q_X}\frac{1}{\sqrt{2\pi\sigma_{0X}^2(i)}}\exp\left(-\frac{1}{2\sigma_{0X}^2(i)}(x_i - \mu_{0X}(i))^2\right)}{\prod_{i=1}^{q_X}\frac{1}{\sqrt{2\pi\sigma_{1X}^2(i)}}\exp\left(-\frac{1}{2\sigma_{1X}^2(i)}(x_i - \mu_{1X}(i))^2\right)} \nonumber \\
&= \sqrt{\prod_{i=1}^{q_X}\frac{\sigma_{1X}^2(i)}{\sigma_{0X}^2(i)}}\exp\left(\sum_{i=1}^{q_X}(C_{X1i}x_i^2 + C_{X2i}x_i + C_{X3i})\right)  \nonumber \\
&= \exp\left(C_{X4} + \left(\sum_{i=1}^{q_X}(C_{X1i}x_i^2 + C_{X2i}x_i)\right)\right), \nonumber 
\end{align*}
for some constants $C_{X4}, C_{X1i}, C_{X2i} \in \mathbb{R}$. Therefore, by Bayes' rule, the following holds,
        \begin{align*}
            P(D=1|X=x) &= \frac{1}{1+\frac{P(D=0)}{P(D=1)}\frac{P(X=x|D=0)}{P(X=x|D=1)}} \\
            &= \frac{1}{1 + \exp\left(-\left(C_{X5} + \left(\sum_{i=1}^{q_X}(-C_{X1i}x_i^2 - C_{X2i}x_i)\right)\right)\right)} \\
            &= \text{expit}\left(C_{X5} - \left(\sum_{i=1}^{q_X}(C_{X1i}x_i^2 + C_{X2i}x_i)\right)\right),
        \end{align*}
        for some constants $C_{X5}, C_{X1i}, C_{X2i} \in \mathbb{R}$, which aligns with the functional form of a logistic regression model with quadratic terms.
        \item[(b)] \textbf{Long:} We proceed analogously to the long case.
        \begin{align*}
            &\frac{P(X=x,U=u|D=0)}{P(X=x,U=u|D=1)} \\
            &= \frac{P(X=x|D=0)P(U=u|D=0)}{P(X=x|D=1)P(U=u|D=1)} \text{ by conditional independence,} \\
            &= \exp\left(C_{XU4} + \left(\sum_{i=1}^{q_X}(C_{X1i}x_i^2 + C_{X2i}x_i)\right) + \left(\sum_{i=1}^{q_U}(C_{U1i}u_i^2 + C_{U2i}u_i)\right)\right),
        \end{align*}
        for some constants $C_{XU4}, C_{X1i}, C_{X2i}, C_{U1i}, C_{U2i} \in \mathbb{R}$. Therefore, by Bayes' rule, 
        \begin{align*}
            &P(D=1|X=x,U=u)\\
            &= \frac{1}{1+\frac{P(D=0)}{P(D=1)}\frac{P(X=x,U=u|D=0)}{P(X=x,U=u|D=1)}} \\
            &= \text{expit}\left(C_{XU5} - \left(\sum_{i=1}^{q_X}(C_{X1i}x_i^2 + C_{X2i}x_i)+ \sum_{i=1}^{q_U}(C_{U1i}u_i^2 + C_{U2i}u_i)\right)\right),
        \end{align*}
        for some constants $C_{XU5}, C_{X1i}, C_{X2i}, C_{U1i}, C_{U2i} \in \mathbb{R}$, which aligns with the functional form of a logistic regression model with quadratic terms.
    \end{enumerate}

    Next, we derive closed-form expressions for the bias factors in our simulation setting.
    \begin{enumerate}
        \item[(a)] $\bm{C_{0\Delta Y}^2:}$
        \begin{align*}
            E[(g_0 - g_{0s})^2|D=0]  &= (\Delta \beta_{u})^2E[(U - E[U|D=0])^2|D=0] \text{ by } U\indep X\mid D=0,\nonumber \\
            &= (\Delta \beta_{u})^2\operatorname{Var}(U|D=0) = 1 \times \sigma_{0U}^2, \nonumber \\
            E[(\Delta Y - g_{0s})^2|D=0] &= E[(\Delta \beta_{u}(U - E[U|D=0]) + \Delta \tilde{\epsilon})^2|D=0] \text{ by } U\indep X\mid D=0, \nonumber \\
            &= E[(g_0 - g_{0s})^2|D=0] + E[\Delta \tilde{\epsilon}^2|D=0] + 2E[(g_0 - g_{0s})\Delta \tilde{\epsilon}|D=0] \nonumber \\
            &= E[(g_0 - g_{0s})^2|D=0] + E[\Delta \tilde{\epsilon}^2|D=0] \text{ by independence},\nonumber \\
            &= \sigma_{0U}^2 + 2, \nonumber \\
            \implies \quad C_{0\Delta Y}^2 &= \frac{E[(g_0 - g_{0s})^2|D=0]}{E[(\Delta Y - g_{0s})^2|D=0]} = \frac{\sigma_{0U}^2}{\sigma_{0U}^2 + 2}. \nonumber 
        \end{align*}
        \item[(b)] $\bm{C_{0D}^2:}$
        \begin{align*}
            C_{0D}^2 &= \frac{\chi^2(P_{X,U|1}\|P_{X,U|0}) - \chi^2(P_{X|1}\|P_{X|0})}{\chi^2(P_{X|1}\|P_{X|0}) + 1},
        \end{align*}
        which follows directly from Corollary~\ref{thm:alt-selection}. To compute this quantity in our simulation, we plug in the conditional density functions of $X$ and $U$ given $D$, and use conditional independence to simplify the calculation. Specifically,
        \begin{align*}
            &\chi^2(P_{X|1}\|P_{X|0}) \\
            &= \int \frac{\dd P_{X|1}(x)}{\dd P_{X|0}(x)}\dd P_{X|1}(x) - 1 \\
            &= \frac{\sigma_{0X}}{\sqrt{2\pi}\sigma_{1X}^2}\int \exp\left(- \frac{(x - \mu_{1X})^2}{\sigma_{1X}^2} + \frac{(x - \mu_{0X})^2}{2\sigma_{0X}^2}\right)dx - 1 \\
            &= \frac{\sigma_{0X}}{\sqrt{2\pi}\sigma_{1X}^2}\int \exp\left(-\underbrace{\left(\frac{1}{\sigma_{1X}^2} -\frac{1}{2\sigma_{0X}^2}\right)}_ax^2 + \underbrace{\left(\frac{2\mu_{1X}}{\sigma_{1X}^2} -\frac{\mu_{0X}}{\sigma_{0X}^2}\right)}_b x - \underbrace{\left(\frac{\mu_{1X}^2}{\sigma_{1X}^2} -\frac{\mu_{0X}^2}{2\sigma_{0X}^2}\right)}_c\right)dx - 1, \\
            (1) \ &\text{If $a > 0$ (i.e., $2\sigma_{0X}^2 > \sigma_{1X}^2$):} \\
           &\chi^2(P_{X|1}\|P_{X|0})  = \frac{\sigma_{0X}}{\sqrt{2\pi}\sigma_{1X}^2}\exp\left(\frac{b^2}{4a} - c\right)\int \exp\left(-\frac{1}{2}\left(\sqrt{2a}x-\frac{b}{\sqrt{2a}}\right)^2\right)dx - 1 \\
            &= \frac{\sigma_{0X}}{\sqrt{2\pi}\sigma_{1X}^2}\exp\left(\frac{b^2}{4a} - c\right)\sqrt{\frac{\pi}{a}} - 1 \\
            &= \frac{\sigma_{0X}^2}{\sigma_{1X}\sqrt{2\sigma_{0X}^2 - \sigma_{1X}^2}}\exp\left(\frac{(\mu_{1X} - \mu_{0X})^2}{2\sigma_{0X}^2 - \sigma_{1X}^2}\right) - 1; \\
            (2) \ &\text{Otherwise:}\\
            &\chi^2(P_{X|1}\|P_{X|0}) = \infty.
        \end{align*}
        
        Similarly, 
        \begin{align*}
            &\chi^2(P_{X,U|1}\|P_{X,U|0}) \\
            &= \left(\chi^2(P_{X|1}\|P_{X|0}) + 1\right)\left(\chi^2(P_{U|1}\|P_{U|0})+1\right) - 1 \text{ by conditional independence,} \\
            (1) \ &\text{If $2\sigma_{0X}^2 > \sigma_{1X}^2$ and $2\sigma_{0U}^2 > \sigma_{1U}^2$:}\\
            &\chi^2(P_{X,U|1}\|P_{X,U|0}) =  \frac{\sigma_{0X}^2\sigma_{0U}^2}{\sigma_{1X}\sigma_{1U}\sqrt{(2\sigma_{0X}^2 - \sigma_{1X}^2)(2\sigma_{0U}^2 - \sigma_{1U}^2)}} \times \\
            &
            \qquad \qquad \qquad
            \qquad \qquad \qquad
            \exp\left(\frac{(\mu_{1X} - \mu_{0X})^2}{2\sigma_{0X}^2 - \sigma_{1X}^2} + \frac{(\mu_{1U} - \mu_{0U})^2}{2\sigma_{0U}^2 - \sigma_{1U}^2}\right) - 1; \\
            (2) \ &\text{Otherwise:}\\ 
            &\chi^2(P_{X,U|1}\|P_{X,U|0}) = \infty.
        \end{align*}
        Accordingly, when $2\sigma_{0X}^2 > \sigma_{1X}^2$ and $2\sigma_{0U}^2 > \sigma_{1U}^2$, 
        \begin{align*}
            C_{0D}^2 &= \frac{\chi^2(P_{X,U|1}\|P_{X,U|0}) - \chi^2(P_{X|1}\|P_{X|0})}{\chi^2(P_{X|1}\|P_{X|0}) + 1} \\
            &= \frac{\left(\chi^2(P_{X|1}\|P_{X|0}) + 1\right)\left(\chi^2(P_{U|1}\|P_{U|0})+1\right) - 1 - \chi^2(P_{X|1}\|P_{X|0})}{\chi^2(P_{X|1}\|P_{X|0}) + 1} \\
            &= \chi^2(P_{U|1}\|P_{U|0}) \\
            &= \frac{\sigma_{0U}^2}{\sigma_{1U}\sqrt{2\sigma_{0U}^2 - \sigma_{1U}^2}}\exp\left(\frac{(\mu_{1U} - \mu_{0U})^2}{2\sigma_{0U}^2 - \sigma_{1U}^2}\right) - 1. 
        \end{align*}

        \item[(c)] $\bm{\rho_0^2:}$
        \begin{align*}
            \theta_s &= E\left[\left(E[\Delta Y|D=1,X] - E[\Delta Y|D=0,X]\right)\mid D=1\right] \nonumber \\
            &= E\left[\frac{D}{p}\left(\theta + E[U\Delta\beta_{u}|D=1] - E[U\Delta\beta_{u}|D=0]\right) \right] \text{ by } U\indep X\mid D=0, \nonumber \\
            &= \theta + (\mu_{1U} - \mu_{0U})\Delta\beta_{u} = \theta + (\mu_{1U} - \mu_{0U}). \\
        \end{align*}
        Accordingly,
        \begin{align*}
            &\operatorname{Cov}(g_0 - g_{0s}, O_{XU} - O_X|D=0) \\
            &= E[(g_0 - g_{0s})(O_{XU} - O_X)|D=0] \nonumber \\
            &= O \times (\theta_0 - \theta_{0s}) \nonumber \\
            &= O \times (\theta_s - \theta) \\
            &= O \times (\mu_{1U} - \mu_{0U}),
        \end{align*}
        where the second equality follows from the proof of Theorem~\ref{thm:main_npm}. Additionally,
        \begin{align*}
            \operatorname{Var}(O_{XU} - O_{X}|D=0) &= E[(O_{XU} - O_{X})^2|D=0] \nonumber \\
            &= E[O_{XU}^2|D=0] - E[O_{X}^2|D=0] \nonumber \\
            &= O^2 \times (\chi^2(P_{X,U|1}\|P_{X,U|0}) - \chi^2(P_{X|1}\|P_{X|0})) \text{ by Proposition~\ref{prop: odds_and_divergence} (1)}, \\
            \operatorname{Var}(g_0 - g_{0s}\mid D=0) &= E[(g_0 - g_{0s})^2|D=0] = \sigma_{0U}^2. 
        \end{align*}
        Finally, we have 
        \begin{align*}
            \rho_0^2 &= \frac{\operatorname{Cov}^2(g_0 - g_{0s}, O_{XU} - O_X|D=0)}{\operatorname{Var}(O_{XU} - O_{X}|D=0)\operatorname{Var}(g_0 - g_{0s}\mid D=0)} \\
            &= \frac{(\mu_{1U} - \mu_{0U})^2}{\sigma_{0U}^2 \times (\chi^2(P_{X,U|1}\|P_{X,U|0}) - \chi^2(P_{X|1}\|P_{X|0}))} \\
            &= \frac{(\mu_{1U} - \mu_{0U})^2}{\sigma_{0U}^2 \times (\chi^2(P_{X|1}\|P_{X|0})+1)\chi^2(P_{U|1}\|P_{U|0})}. 
        \end{align*}
    \end{enumerate}

    \subsection{Simulation: figures and tables}

    \begin{figure}[H]
            \centering
            \includegraphics[width=\linewidth]{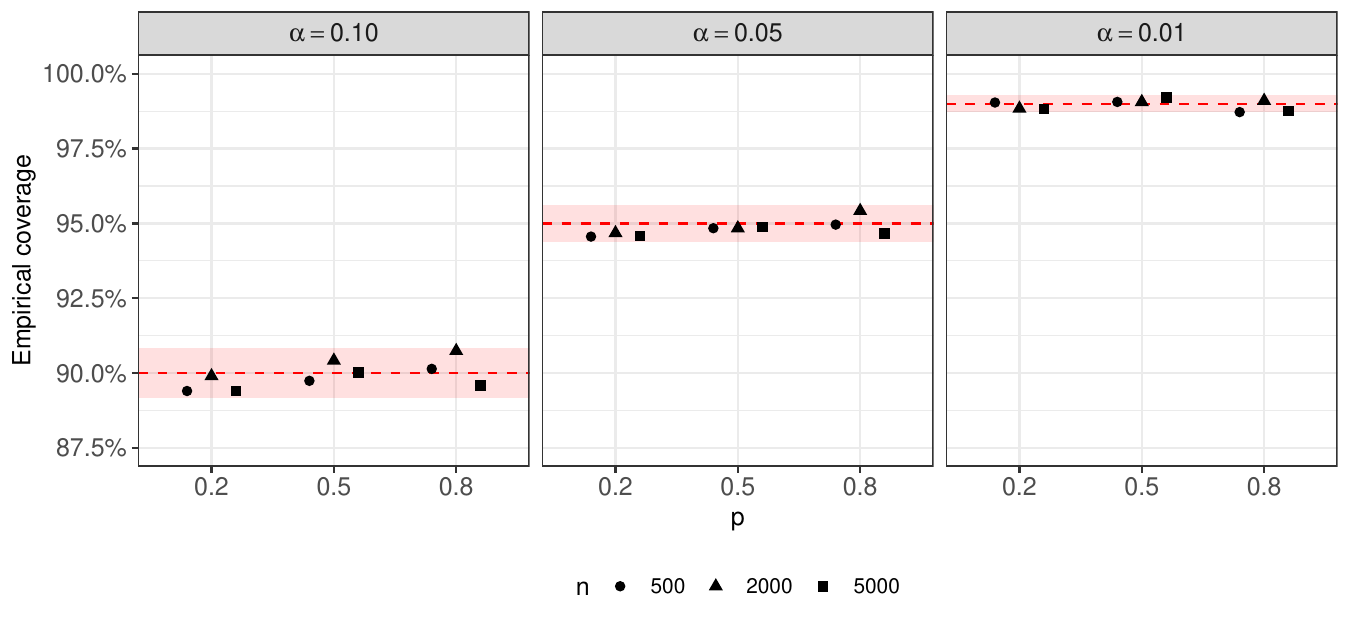}
            \caption{Empirical coverage probabilities of 90\%, 95\%, and 99\% confidence bounds across sample sizes $n$ and treatment probabilities $p$, using the true bias factors. Red dashed lines indicate nominal coverage levels; shaded regions indicate the expected (95\%) Monte Carlo error under nominal coverage.}\label{fig: sim_coverage}
        \end{figure}

\begin{table}[!ht]
\centering
\setlength{\tabcolsep}{4pt}
\renewcommand{\arraystretch}{0.9}
\small
\begin{tabular}{rr rr rr rr}
\toprule
& & \multicolumn{2}{c}{90\%} & \multicolumn{2}{c}{95\%} & \multicolumn{2}{c}{99\%} \\
\cmidrule(lr){3-4}\cmidrule(lr){5-6}\cmidrule(lr){7-8}
$p$ & $n$ & Coverage & $z$ & Coverage & $z$ & Coverage & $z$ \\
\midrule
0.2 & 500       & 0.8940 & $-1.41$ & 0.9456 & $-1.43$ & 0.9906 & $ 0.43$ \\
    & 2{,}000   & 0.8988 & $-0.28$ & 0.9466 & $-1.10$ & 0.9884 & $-1.14$ \\
    & 5{,}000   & 0.8936 & $-1.51$ & 0.9458 & $-1.36$ & 0.9882 & $-1.28$ \\
\midrule
0.5 & 500       & 0.8974 & $-0.61$ & 0.9484 & $-0.52$ & 0.9906 & $ 0.43$ \\
    & 2{,}000   & 0.9036 & $ 0.85$ & 0.9488 & $-0.39$ & 0.9906 & $ 0.43$ \\
    & 5{,}000   & 0.9000 & $ 0.00$ & 0.9486 & $-0.45$ & 0.9916 & $ 1.14$ \\
\midrule
0.8 & 500       & 0.9014 & $ 0.33$ & 0.9496 & $-0.13$ & 0.9872 & $-1.99$ \\
    & 2{,}000   & 0.9090 & $ 2.12$ & 0.9534 & $ 1.10$ & 0.9914 & $ 0.99$ \\
    & 5{,}000   & 0.8958 & $-0.99$ & 0.9470 & $-0.97$ & 0.9876 & $-1.71$ \\
\bottomrule
\end{tabular}
\caption{Empirical coverage probabilities of 90\%, 95\%, and 99\% confidence bounds across sample sizes $n$ and treatment probabilities $p$, using the true bias factors. The $z$-statistic indicates the difference between empirical and nominal coverage divided by the Monte Carlo standard error, $\sqrt{\alpha(1-\alpha)/5{,}000}$.}
\label{tab: sim_coverage}
\end{table}

\clearpage  
    \begin{figure}[H]
    \centering
    \includegraphics[width=\linewidth]{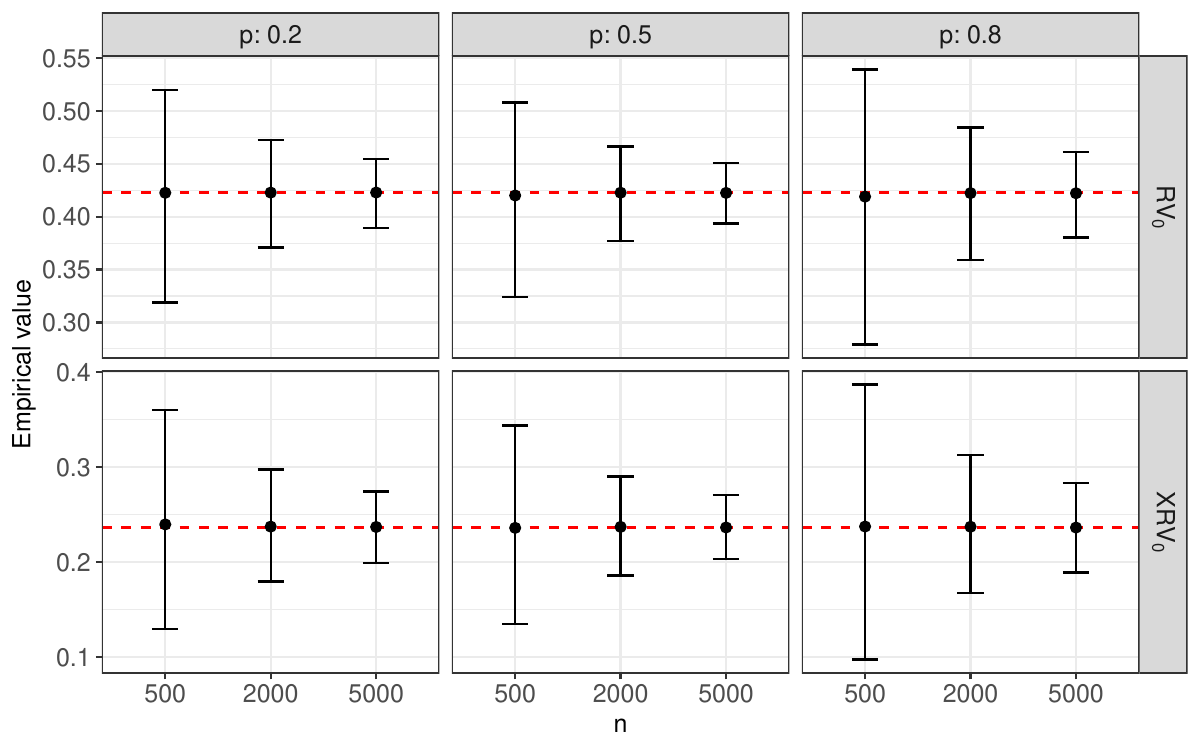}
    \caption{Finite sample behavior of $\text{RV}_{\theta^*=0,\alpha=1}$ and $\text{XRV}_{\theta^*=0,\alpha=1}$ across sample sizes and treatment probabilities ($5{,}000$ repetitions). Points denote Monte Carlo means, whiskers delimit the empirical $2.5$th and $97.5$th percentiles. The red dashed line indicates the true population value.}
    \label{fig:ss_point}
\end{figure}

    \begin{figure}[H]
    \centering
    \includegraphics[width=\linewidth]{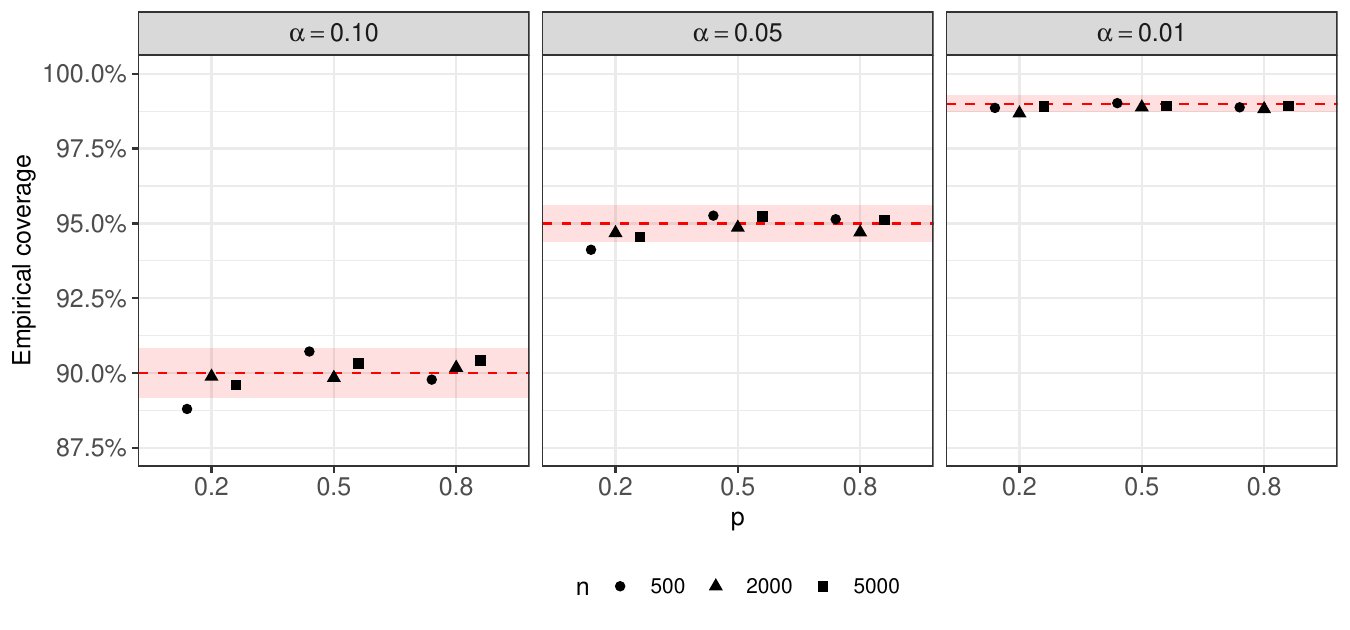}
    \caption{Empirical coverage of $\text{RV}_{\theta^*=0,\alpha}$ and $\text{XRV}_{\theta^*=0,\alpha}$ across sample sizes and treatment probabilities ($5{,}000$ repetitions). Red dashed lines indicate the nominal coverage levels $1-\alpha$; shaded regions indicate the expected (95\%) Monte Carlo error under nominal coverage.}
    \label{fig:ss_CI}
\end{figure}

\begin{table}[H]
\centering
\footnotesize                      
\setlength{\tabcolsep}{4pt}         
\renewcommand{\arraystretch}{1.0}  
\begin{tabular}{llll}
\toprule
\textbf{Learner} & \textbf{Variables} & \textbf{Package} & \textbf{Tuning grid} \\
\midrule
\makecell[l]{Parametric: \\
Linear for $g_{0s}$ \\
Logistic for $\pi$}  & $X^*,(X^*)^2$
 &
\makecell[l]{\texttt{stats} (\texttt{lm})\\
\texttt{stats} (\texttt{glm})} &
N/A\\
\hline
\makecell[l]{Ridge: \\
Linear for $g_{0s}$ \\
Logistic for $\pi$} &
$X^*,(X^*)^2,(X^*)^3$ &
\texttt{glmnet} &
\makecell[l]{$\alpha_{\text{EN}} = 0$ (ridge) \\ $\lambda$ selected from \texttt{cv.glmnet}} \\
\hline
\makecell[l]{Lasso: \\
Linear for $g_{0s}$ \\
Logistic for $\pi$} &
$X^*,(X^*)^2,(X^*)^3$ &
\texttt{glmnet} &
\makecell[l]{$\alpha_{\text{EN}} = 1$ (lasso) \\ $\lambda$ selected from \texttt{cv.glmnet}} \\
\hline
\makecell[l]{Random Forest \\ ($\texttt{num.trees} = 500$)} &
$X^*$ &
\texttt{ranger} &
\makecell[l]{\texttt{mtry} $\in \{2,3\}$ \\
             \texttt{min.node.size} $\in \{50, 100, \ldots,300\}$ \\
             \texttt{splitrule} $\in \{\texttt{variance}, \texttt{extratrees}\}$} \\
\bottomrule
\end{tabular}
\caption{Machine learning methods used for first-stage nuisance estimation. Covariates enter linearly. Models are selected by minimizing the RMSE. Across $500$ repetitions, random forest is selected in $87\%$ of replications for the propensity score $\pi$ and in $98.8\%$ of replications for the outcome evolution model $g_{0s}$.}
\label{tab:miss_learners}
\end{table}

\begin{figure}[H]
\centering
\begin{minipage}[t]{\textwidth}
    \centering
    \includegraphics[width=\linewidth]{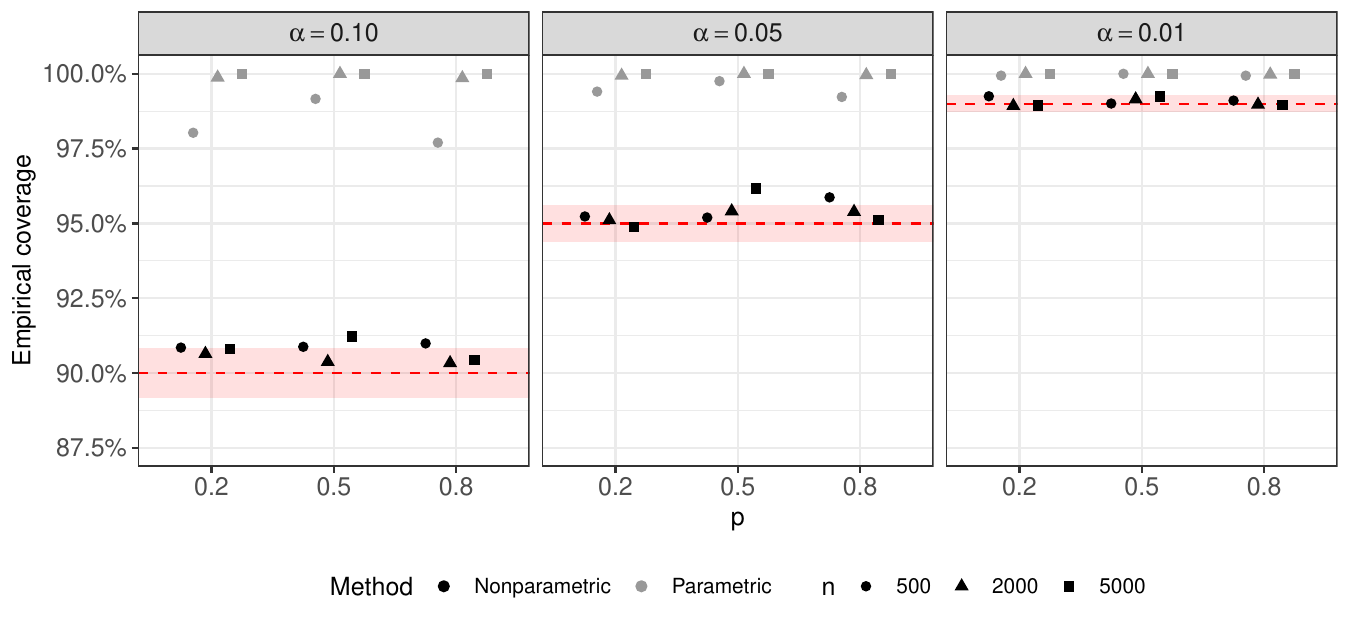}
    \captionof{figure}{Comparison of empirical coverage probabilities between parametric and nonparametric confidence bounds, using the true bias factors. Red dashed lines indicate nominal coverage levels; shaded regions indicate the expected (95\%) Monte Carlo error under nominal coverage.}
    \label{fig:miss_coverage}
\end{minipage}
\hfill

\vspace{1cm}
\begin{minipage}[t]{0.8\textwidth}
    \centering
    \includegraphics[width=\textwidth]{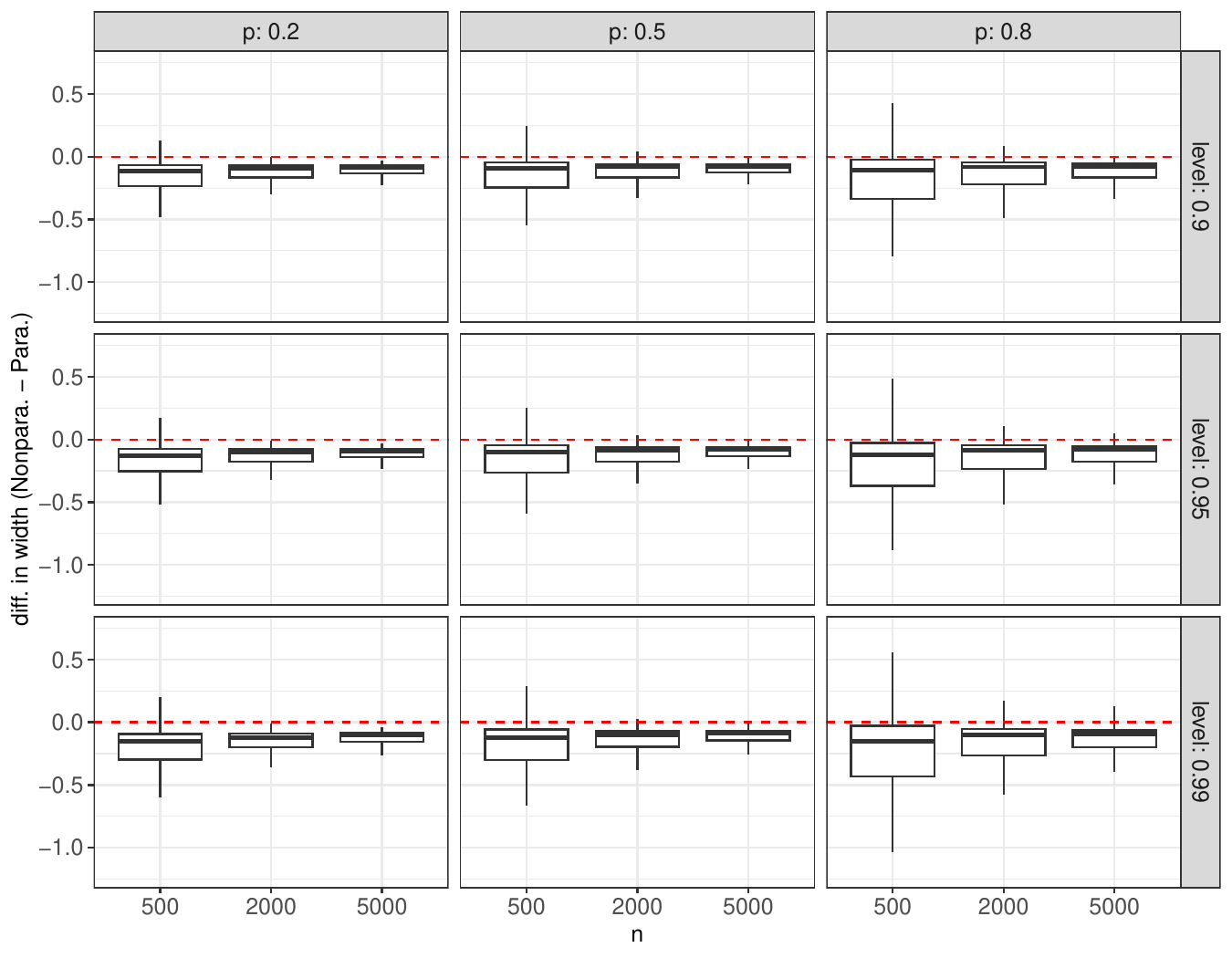}
    \captionof{figure}{Boxplots of the difference in confidence-interval widths (nonparametric minus parametric) across sample sizes and treatment probabilities; the red dashed line denotes zero difference.}
    \label{fig:width_diff_miss}
\end{minipage}
\end{figure}

\begin{figure}[H]
\centering
\begin{minipage}[t]{0.8\textwidth}
    \centering
    \includegraphics[width=\textwidth]{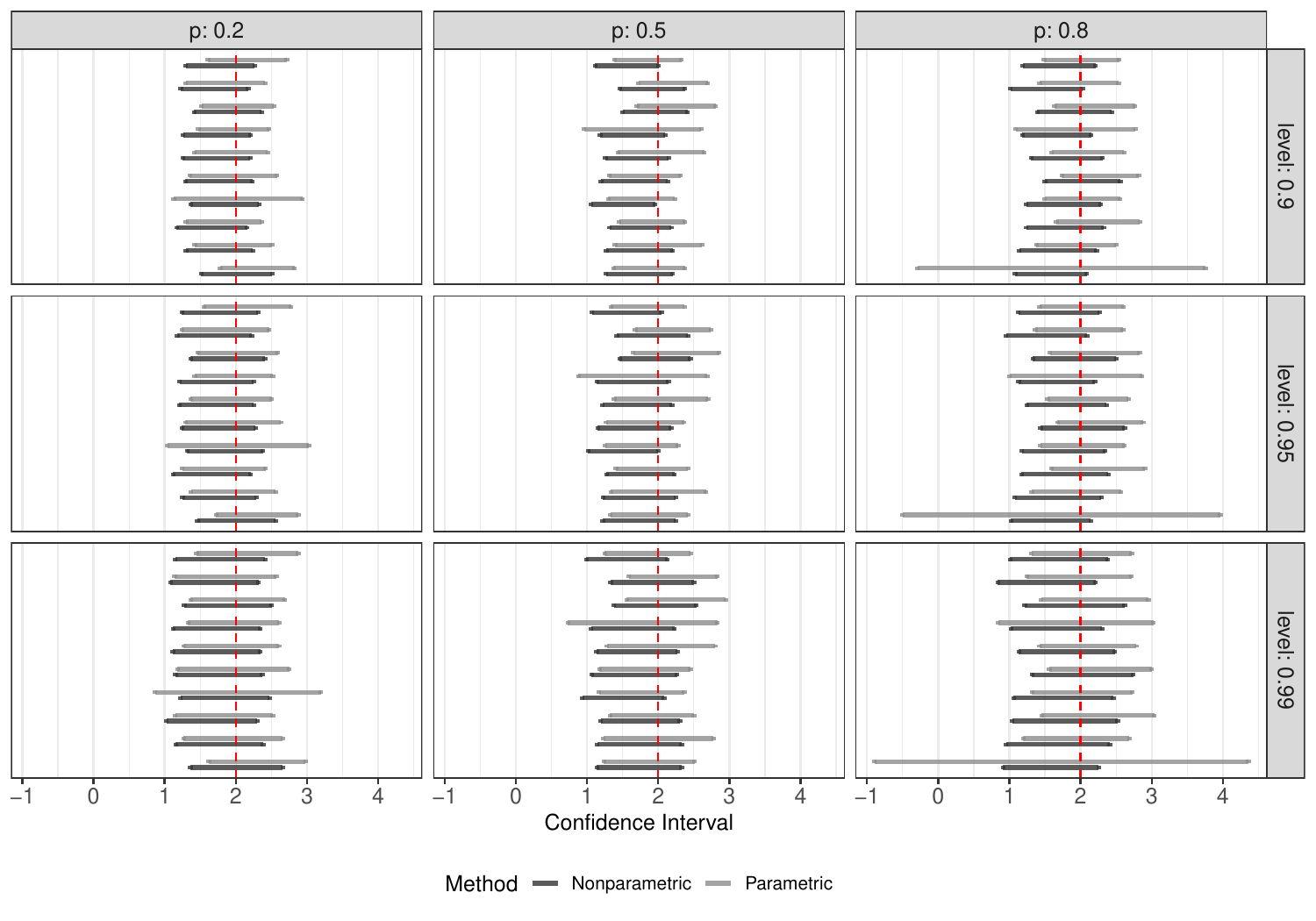}
    \captionof{figure}{Comparison of parametric and nonparametric example confidence bounds under model misspecification; red dashed lines indicate true $\theta$ value.}
    \label{fig:zoom_miss}
\end{minipage}
\hfill

\vspace{1cm}
\begin{minipage}[t]{0.8\textwidth}
    \centering
    \includegraphics[width=\textwidth]{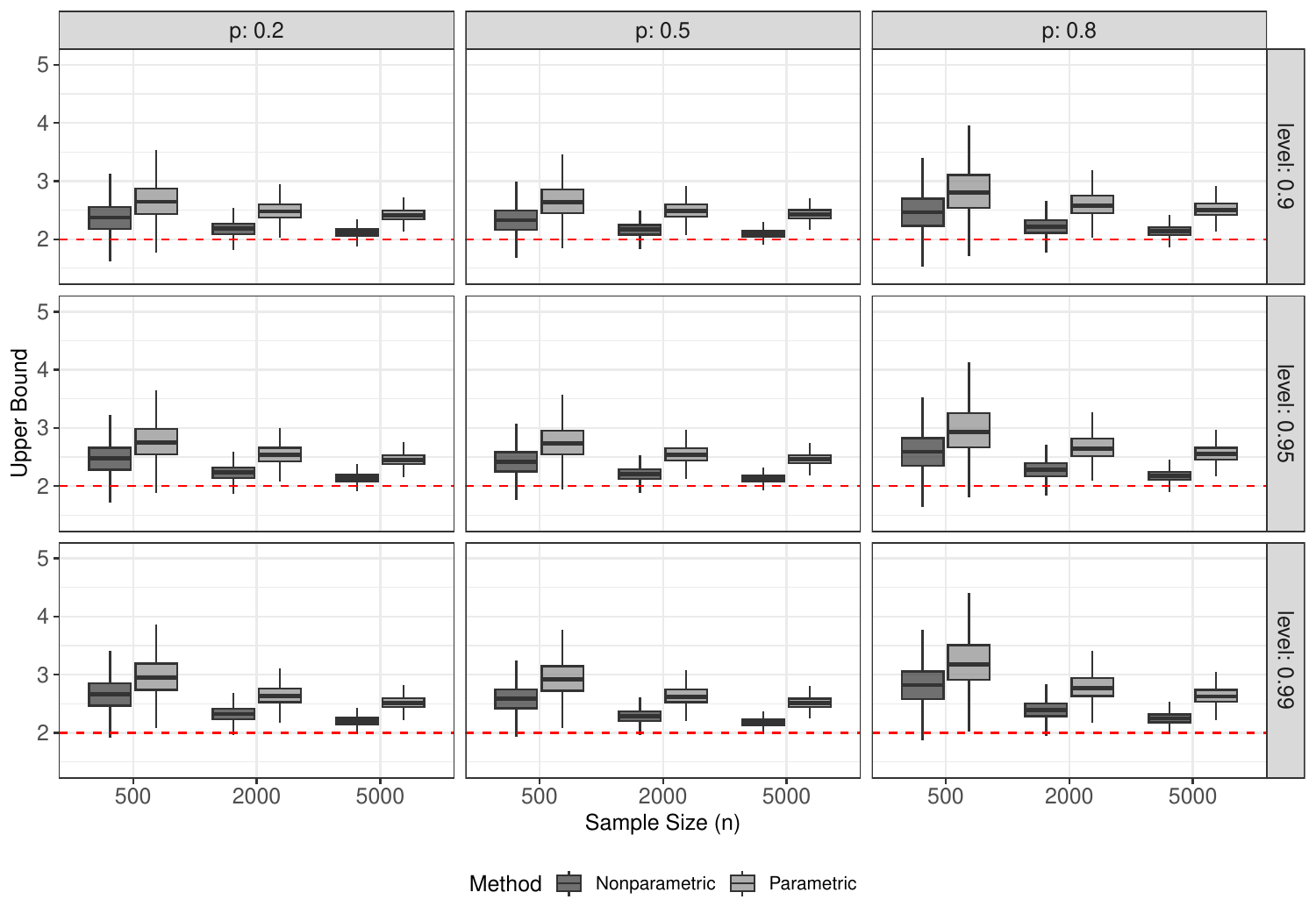}
    \captionof{figure}{Comparison of parametric and nonparametric upper confidence bounds under model misspecification. Red dashed lines indicate the true $\theta$ value.}
    \label{fig:upper_miss}
\end{minipage}
\end{figure}

\begin{figure}[H]
    \centering
    \includegraphics[width=\linewidth]{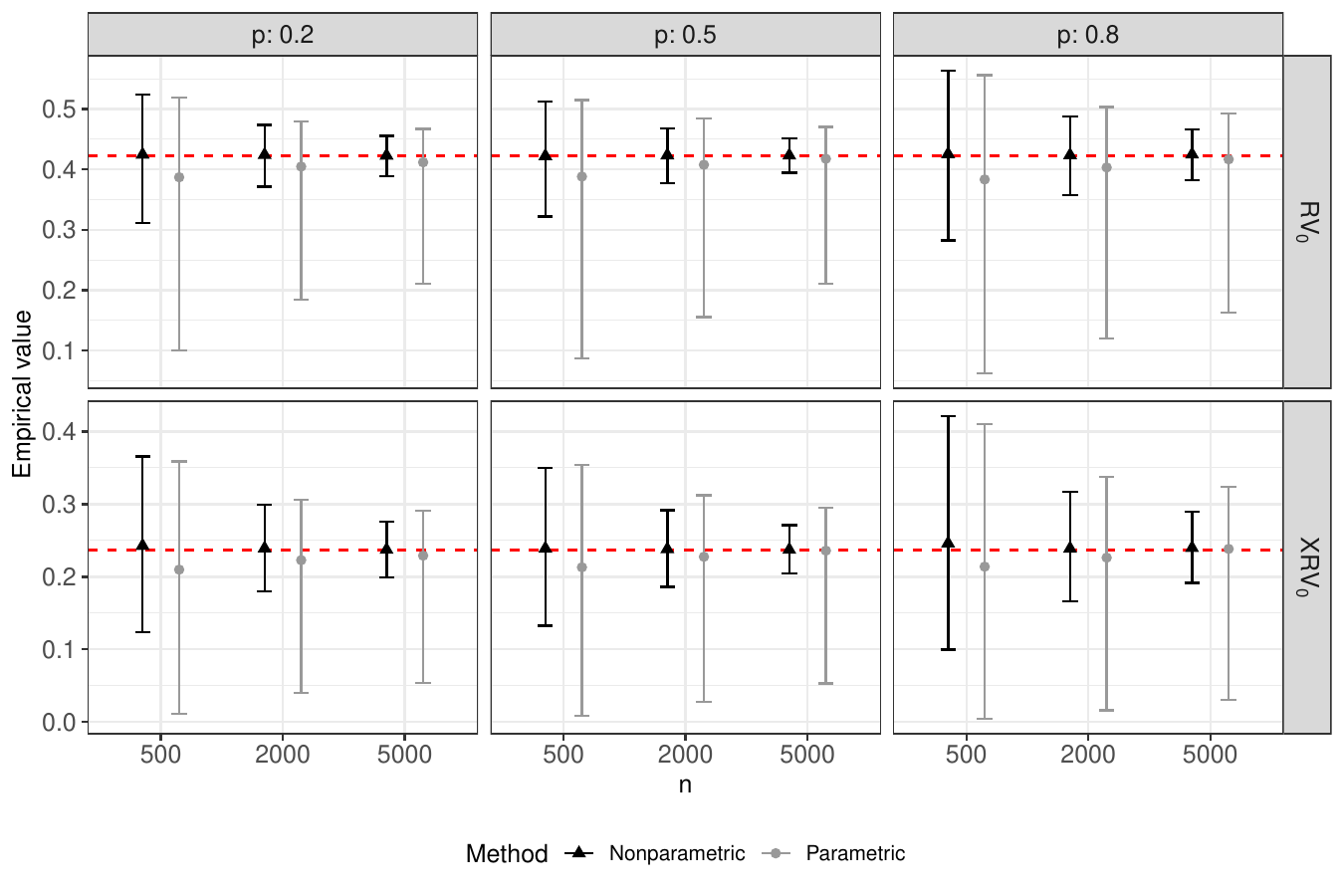}
    \caption{Finite sample behavior of $\text{RV}_{\theta^*=0,\alpha=1}$ and $\text{XRV}_{\theta^*=0,\alpha=1}$ under parametric and nonparametric nuisance estimation across sample sizes and treatment probabilities. Points denote Monte Carlo means; whiskers delimit the empirical $2.5$th and $97.5$th percentiles.}
    \label{fig:ss_miss_point}
\end{figure}

\begin{figure}[H]
    \centering
    \includegraphics[width=\linewidth]{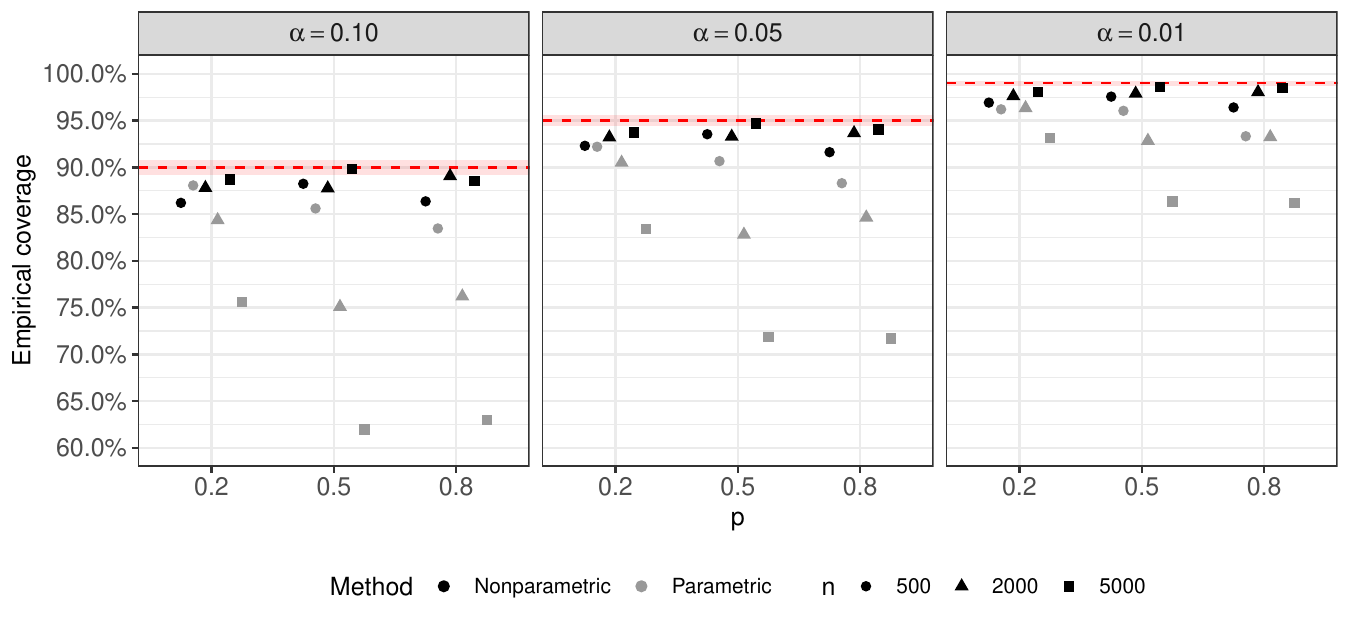}
    \caption{Empirical coverage of $\text{RV}_{\theta^*=0,\alpha}$ and $\text{XRV}_{\theta^*=0,\alpha}$ under parametric and nonparametric nuisance estimation. Red dashed lines indicate the nominal coverage levels $1-\alpha$; shaded regions indicate the expected (95\%) Monte Carlo error under nominal coverage.}
    \label{fig:ss_miss_coverage}
\end{figure}

\section{Comparison with related literature} 
\label{app:all-comparison}

\subsection{Relationship to the ``unconditional" OVB analysis}
\label{app:unc-comparison}

In this section, we compare our OVB analysis for the ATT with the results of \citet{chernozhukov2022long} as applied to the ATT, which were used by \citet{bach2025sensitivity}. We refer to that result as ``unconditional.'' Both our results and the unconditional results target the same causal estimand and characterize the same total bias. The distinction lies in how confounding strength is parameterized and interpreted.

\subsubsection{Review of the unconditional results}

We begin by briefly reviewing the results of \citet{chernozhukov2022long}. Under the same identification assumptions as those imposed in our main analysis, they start with the following parameterization of the ATT:
\begin{align*}
    &\theta = E[g(D,X,U)\alpha(D,X,U)], \text{ and } \theta_s = E[g_s(D,X)\alpha_s(D,X)] \text{ with }\\
    &g(d,X,U) := E[\Delta Y\mid D=d,X,U], \text{ and } g_s(d,X) := E[\Delta Y\mid D=d,X], \\
    &\alpha(d,X,U) := \frac{D}{p} - \frac{1-D}{1-p}\times \frac{O_{XU}}{O}, \text{ and } \alpha_s(d,X) := \frac{D}{p} - \frac{1-D}{1-p}\times \frac{O_{X}}{O}. 
\end{align*}
This leads to the following OVB result for the ATT:
\begin{align}
        \theta - \theta_s &= \rho\times \underbrace{\sqrt{R^2_{\Delta Y -g_s\sim g-g_s}}}_{=: C_{\Delta Y}} \times \underbrace{\sqrt{\frac{1 - R^2_{\alpha \sim \alpha_s}}{R^2_{\alpha \sim \alpha_s}}}}_{=:C_{D}}\times \underbrace{\sqrt{E[(\Delta Y - g_s)^2]\times E[\alpha_s^2]}}_{=: S}, \nonumber \\
        \text{where } \rho &= \operatorname{Cor}\left(g-g_s, \alpha - \alpha_s\right), \nonumber
    \end{align} 
which can be further specialized as:
\begin{align}
        \theta - \theta_s &= \rho\times \underbrace{\sqrt{\eta^2_{\Delta Y\sim U \mid X,D}}}_{=: C_{\Delta Y}} \times \underbrace{\sqrt{\frac{E[O_{XU}] - E[O_{X}]}{E[O_{X}]}}}_{=:C_{D}}\times \underbrace{\sqrt{E[\operatorname{Var}(\Delta Y \mid X,D)]\times \left(\frac{E[O_{X}]}{p^2}\right)}}_{=: S}, \nonumber \\
        \text{where } \rho &= \operatorname{Cor}\left(E[\Delta Y\mid D,X,U] - E[\Delta Y\mid D,X], -(1-D) \times (O_{XU} - O_{X})\right).\nonumber
    \end{align} 

\subsubsection{Alternative ways to quantify selection strength for the unconditional analysis}

For completeness, we now provide alternative characterizations of selection strength for the unconditional results, expressed in terms of $R^2$ measures in odds or $\chi^2$ divergence. These characterizations show that their selection strength measure increases with the treatment probability.

\begin{Proposition}[Alternative characterization of $C_D^2$]\label{prop: CB_D_equivalent}
    \begin{align*}
    C_{D}^2 &= \frac{ 1 - R^2_{O_{XU} \sim O_X|D=0}}{\frac{1}{O}R^2_{O_{XU} \sim O|D=0} + R^2_{O_{XU} \sim O_X|D=0}} \tag*{\text{(Residual $R^2$ in Selection Odds)}} \\
    &= \frac{\chi^2(P_{X,U|1}\|P_{X,U|0})-\chi^2(P_{X|1}\|P_{X|0})}{\chi^2(P_{X|1}\|P_{X|0}) + \frac{1}{p}} \tag*{\text{(Distributional Imbalance)}}. 
\end{align*}
\end{Proposition}

\subsubsection{Difference in parameterization}

For the ATT, the only unobserved component is the untreated counterfactual trend for the treated units. The treated trend for treated units is observed and therefore does not need a bias analysis. This allows us to localize the bias conditionally on the untreated group. Note that the unconditional result does not take this information into account.

We show that $C_{\Delta Y}^2$ is a weighted average of the strength of confounding on trends in the treated and control groups. In particular, as treatment becomes more likely, the weight placed on the control-group component, where the bias originates, decreases. Consequently, practitioners must rely on a pooled measure of confounding, which obscures the source of the bias and may complicate interpretation. Similarly, we show that $C_D^2$ is a weighted version of $C_{0D}^2$. In particular, $C_D^2 < C_{0D}^2$, with a weight on $C_{0D}^2$ that decreases as treatment becomes less likely. Likewise, $\rho^2$ is a downweighted version of $\rho_0^2$, with the weight decreasing as the treatment probability increases.

\begin{Proposition}[Relationship with unconditional] 
\label{prop: CB_connection} 
\par\noindent
    \begin{enumerate}
        \item[(i)] \textbf{$C_{0\Delta Y}^2$ vs $C_Y^2$.}
        \begin{align*}
        C_{\Delta Y}^2 &= W_{0\Delta Y}C_{0\Delta Y}^2 + (1 - W_{0\Delta Y})C_{1\Delta Y}^2, \\
        \text{ with } C_{0\Delta Y}^2 &= \eta^2_{\Delta Y \sim U \mid X,D=0}, \ C_{1\Delta Y}^2 = \eta^2_{\Delta Y \sim U \mid X,D=1}, \text{ and} \\
        W_{0\Delta Y} &= \frac{(1-p) \times \sigma_{0s}^2}{p \times \sigma_{1s}^2 + (1-p) \times \sigma_{0s}^2}, 
    \end{align*}
    where $\sigma_{ds}^2 = E[\operatorname{Var}(\Delta Y \mid X,D=d)\mid D=d]$ for $d \in \{0,1\}$.
    
    \item[(ii)] \textbf{$C_{0D}^2$ vs $C_{D}^2$.}
        \begin{align*}
        C_D^2 &= W_{0D}C_{0D}^2, \\
        \text{with } W_{0D} &= \frac{O\times (\chi^2(P_{X|1}\|P_{X|0}) + 1)}{O\times (\chi^2(P_{X|1}\|P_{X|0}) + 1) + 1}. 
    \end{align*}

    \item[(iii)] \textbf{$\rho_0$ vs $\rho$.} Define $g_{d} := E[\Delta Y \mid X,U,D=d]$ and $g_{ds} := E[\Delta Y \mid X,D=d]$. 
    \begin{align*}
        \rho^2 &= W_{0\rho}\rho_0^2, \\
        \text{ with } W_{0\rho} &= \frac{1}{\frac{\operatorname{Var}(g_1-g_{1s}\mid D=1)}{\operatorname{Var}(g_0-g_{0s}\mid D=0)} \times O + 1}. 
    \end{align*}
    \end{enumerate}
\end{Proposition}

The relationships in Proposition~\ref{prop: CB_connection} also yield the joint restriction emphasized in Section~\ref{sec:parameterization}:
\[
\rho^2 C_{\Delta Y}^2
=
\frac{P(D=0)\,\sigma_{0s}^2}
{E[\operatorname{Var}(\Delta Y\mid X,D)]}
\rho_0^2 C_{0\Delta Y}^2
\leq
\frac{P(D=0)\,\sigma_{0s}^2}
{E[\operatorname{Var}(\Delta Y\mid X,D)]}.
\]
Thus, the unconditional trend and alignment components cannot vary freely. Whenever the residual outcome variance among treated units is positive, their joint upper bound is strictly smaller than one. Consequently, bounding $|\rho|$ and $C_{\Delta Y}$ separately by one produces a looser worst case bias bound, whereas imposing the joint restriction recovers $C_{0D}S_0$, the bound implied by our conditional parameterization.

\subsubsection{Simulation: comparison with unconditional}

In this section, we use Monte Carlo simulations to compare the finite sample properties of our results with the unconditional results of \citet{chernozhukov2022long}. We follow the same data-generating process as in Section~\ref{sec:dgp}.

    \paragraph{Coverage and width of confidence bounds.}

    We evaluate the coverage of the confidence intervals for $\theta$ obtained from our R package \texttt{dml.sensemakr}, plugging in the true bias factors. Figure~\ref{fig: sim_coverage_uncond} reports results across 5,000 repetitions. The empirical coverage of both methods attains nominal levels across sample sizes and treatment probabilities. Figure~\ref{fig:sim_diff_width_uncond} reports the difference in the widths of the confidence intervals for $\theta$ between the two methods. As expected, this difference is centered around zero and concentrates as the sample size grows. Together, these results confirm that both the unconditional method and ours have the expected finite sample performance, producing valid inference for the bias bounds.
    
    Our approach nonetheless offers several advantages over the unconditional method. It shows that plausibility judgments need only be made on the local components that drive the bias, it yields clean characterizations of selection strength, and it delivers sharp bounds on the selection component from the marginal sensitivity model (MSM) that do not depend on the observed data (see Corollary~\ref{cor:msm-parameterization}). In addition, one of our contributions is to provide alternative ways of characterizing selection strength for the unconditional method.

    \begin{figure}[H]
            \centering
            \includegraphics[width=\linewidth]{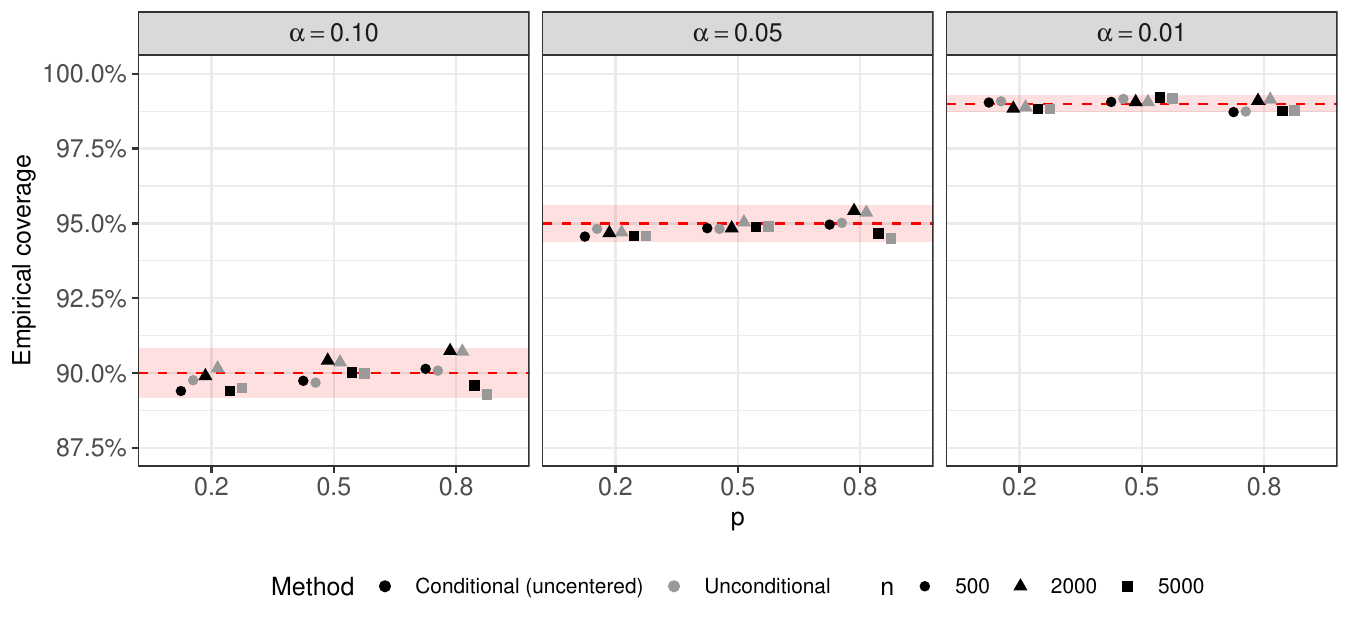}
            \caption{Empirical coverage probabilities of 90\%, 95\%, and 99\% confidence bounds across sample sizes $n$ and treatment probabilities $p$, using the true bias factors. Red dashed lines indicate nominal coverage levels; shaded regions indicate the expected (95\%) Monte Carlo error under nominal coverage.}\label{fig: sim_coverage_uncond}
        \end{figure}

    \begin{figure}[H]
        \centering
        \includegraphics[width=0.8\linewidth]{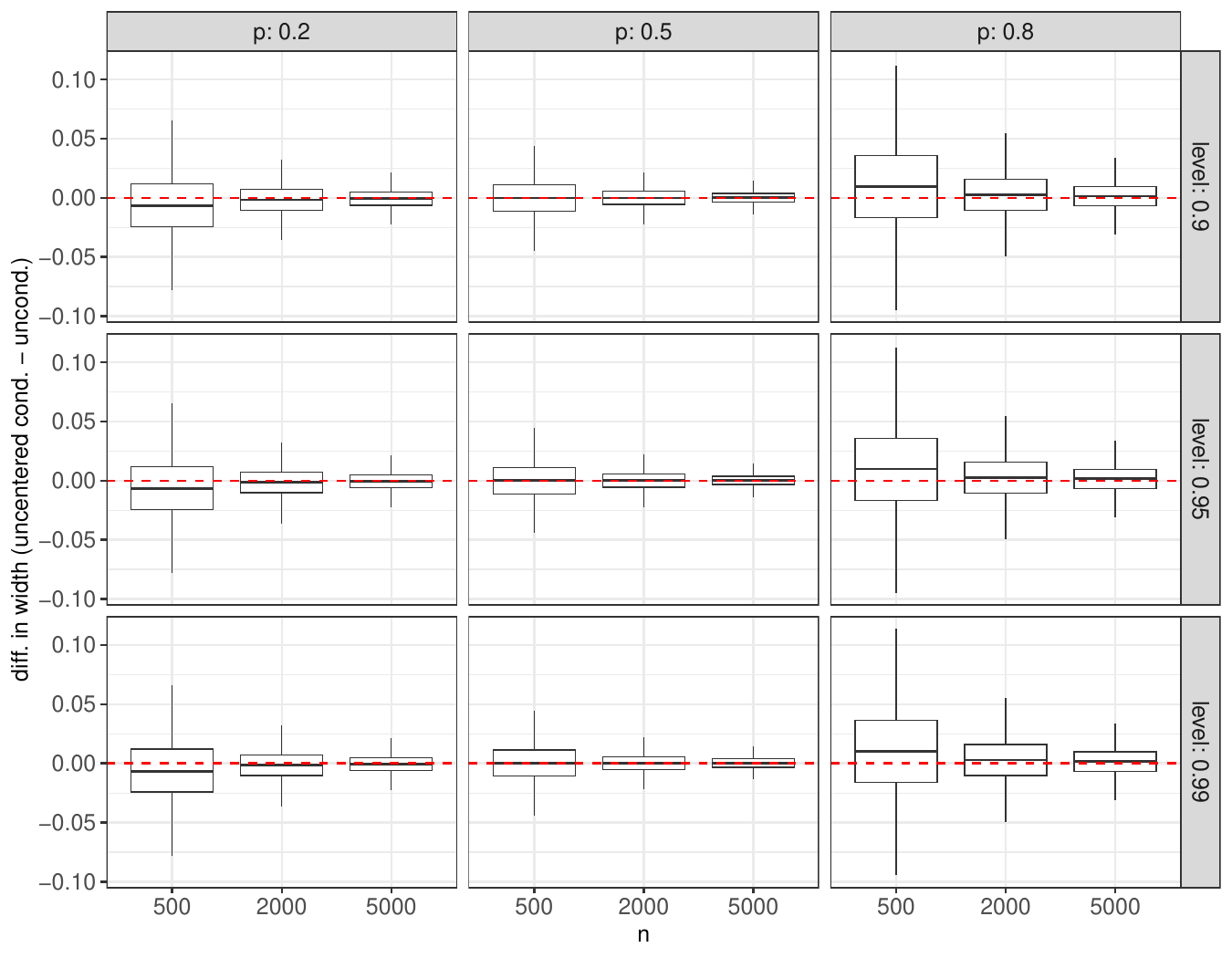}
        \caption{Boxplots of the difference in confidence-interval widths (uncentered conditional minus unconditional) across sample sizes and treatment probabilities; the red dashed line denotes zero difference.}
        \label{fig:sim_diff_width_uncond}
    \end{figure}
        
    \paragraph{Sensitivity statistics.}

    Figures~\ref{fig:ss_uncond_point} and~\ref{fig:ss_uncond_coverage} report the empirical distributions of $\text{RV}_{\theta^*=0,\alpha=1}$ and $\text{XRV}_{\theta^*=0,\alpha=1}$, as well as the empirical coverage of $\text{RV}_{\theta^*=0,\alpha}$ and $\text{XRV}_{\theta^*=0,\alpha}$. As expected, the empirical distributions concentrate around their population values as the sample size increases. Empirical coverage also remains close to nominal levels. Unlike our method, whose population values of sensitivity statistics do not vary with $p$, $\text{RV}_{\theta^*=0}$ and $\text{XRV}_{\theta^*=0}$ for the unconditional method depend on $p$. For $p \in \{0.2, 0.5, 0.8\}$, the unconditional $\text{RV}_{\theta^*=0}$ equals $\{0.2175, 0.2731, 0.2360\}$ and $\text{XRV}_{\theta^*=0}$ equals $\{0.0570, 0.0931, 0.0679\}$. We note that although they share the same name, the sensitivity statistics from our method measure confounding strength conditional on the untreated group and thus carry a different interpretation from those of the unconditional method. This explains the difference in their population values and empirical estimates.

\begin{figure}[H]
    \centering
    \includegraphics[width=\linewidth]{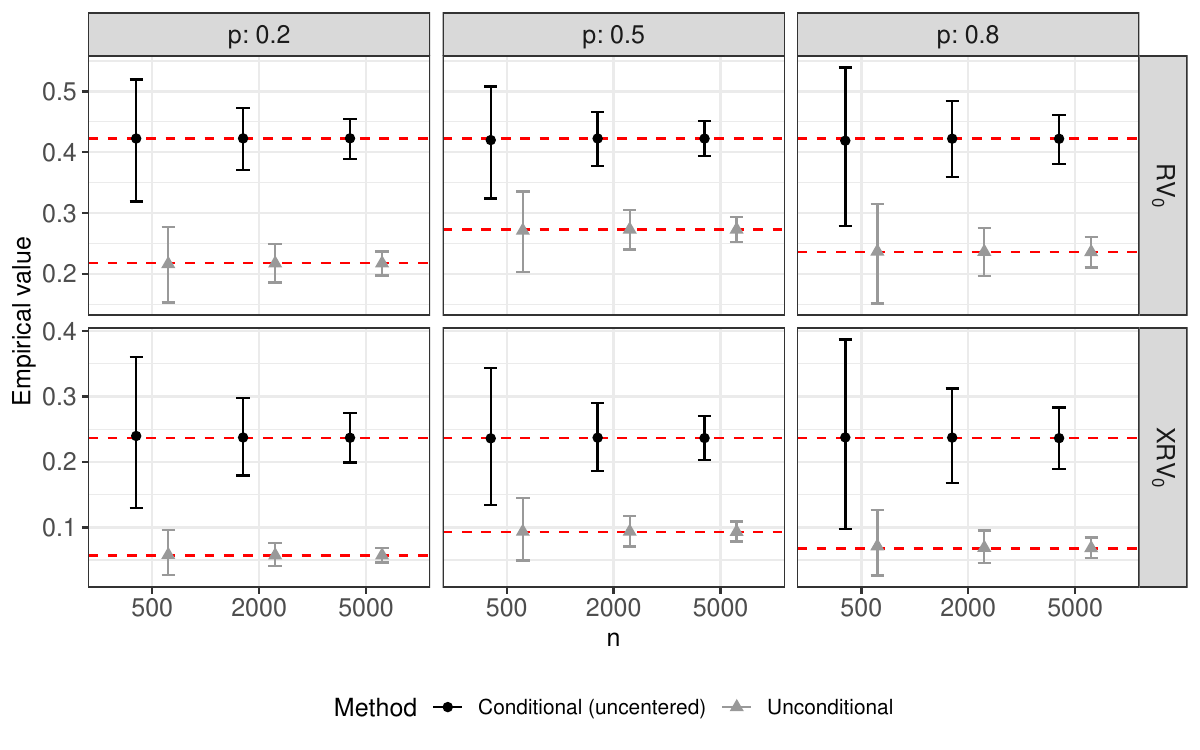}
    \caption{Finite sample behavior of $\text{RV}_{\theta^*=0,\alpha=1}$ and $\text{XRV}_{\theta^*=0,\alpha=1}$ for the conditional and unconditional approaches across sample sizes and treatment probabilities ($5{,}000$ repetitions). Points denote Monte Carlo means; whiskers delimit the empirical $2.5$th and $97.5$th percentiles.}
    \label{fig:ss_uncond_point}
\end{figure}

\begin{figure}[H]
    \centering
    \includegraphics[width=\linewidth]{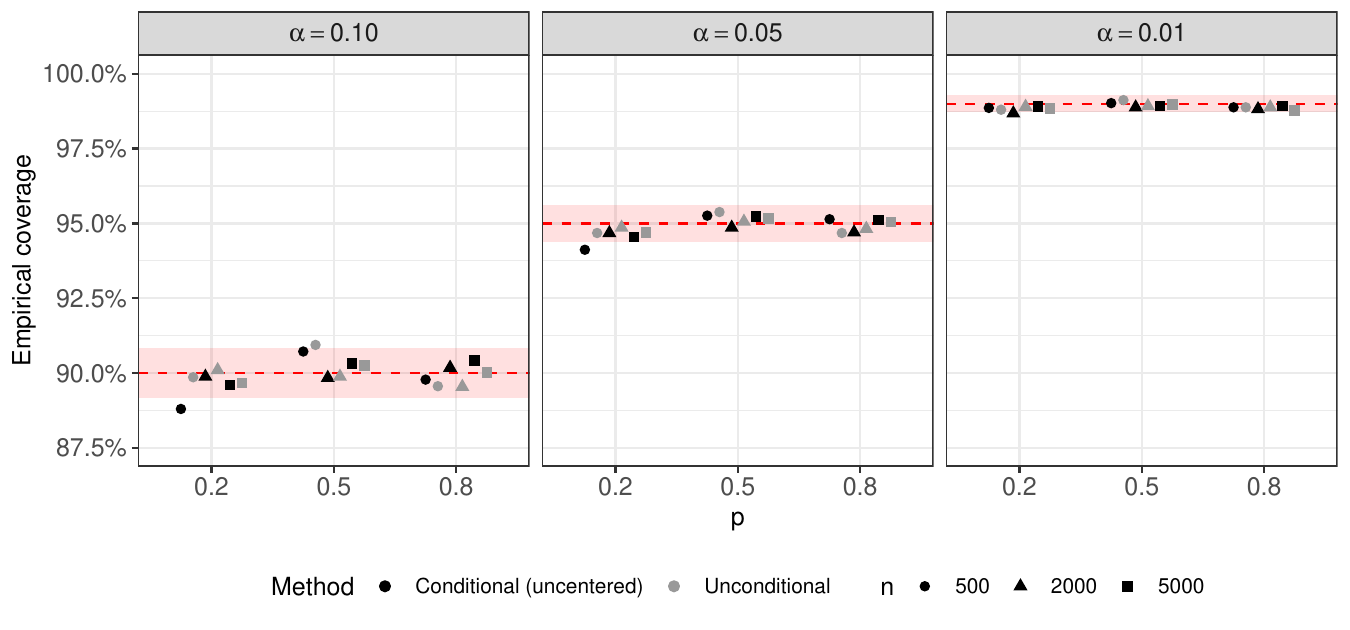}
    \caption{Empirical coverage of $\text{RV}_{\theta^*=0,\alpha}$ and $\text{XRV}_{\theta^*=0,\alpha}$ for the conditional and unconditional approaches across sample sizes and treatment probabilities ($5{,}000$ repetitions). Red dashed lines indicate the nominal coverage levels $1-\alpha$; shaded regions indicate the expected (95\%) Monte Carlo error under nominal coverage.}
    \label{fig:ss_uncond_coverage}
\end{figure}

\subsection{Relationship to variance-based sensitivity analysis for weighting estimators}
\label{app:hp-comparison}

This section compares the results of \citet{huang2025variance}, hereafter HP, with our approach. As before, both results target the same causal estimand and characterize the same total bias. They differ in (1) how confounding strength is parameterized, (2) how benchmarking is performed, and (3) how statistical inference is conducted. For reasons discussed below, we refer to HP's parameterization as the ``centered'' parameterization.

\subsubsection{Review of HP}

We begin by briefly reviewing the IPW formulation of \citet{huang2025variance} (hereafter, HP). To facilitate comparison, we express their characterization in terms of $\Delta Y$, using our notation and indexing HP's bias factors by $w$ to distinguish them from ours.

To begin with, they parameterize the long and short parameters of interest in the following way:
\begin{align*}
    &\theta_0 = E[\alpha_0(X,U)\Delta Y\mid D=0], \text{ and } \theta_{0s} = E[\alpha_{0s}(X)\Delta Y\mid D=0] \text{ with } \\
    &\alpha_0(X,U) = \frac{O_{XU}}{O}, \text{ and } \alpha_{0s}(X) = \frac{O_{X}}{O}. 
\end{align*}
This leads to the following OVB result for the ATT:
\begin{align*}
    \theta - \theta_s &= -\rho_{w0} \times \underbrace{\sqrt{\frac{1 - R^2_{\alpha_0 \sim \alpha_{0s}|1,D=0}}{R^2_{\alpha_{0} \sim \alpha_{0s}|1,D=0}}}}_{=: C_{w0D}} \times \underbrace{\sqrt{\operatorname{Var}(\Delta Y\mid D=0) \times \operatorname{Var}(\alpha_{0s}\mid D=0)}}_{=: S_{w0}}, \\
    \text{where } \rho_{w0} &= \operatorname{Cor}(\Delta Y, \alpha_{0} - \alpha_{0s}\mid D=0).
\end{align*}
We define $\sigma_{w0s}^2 := \operatorname{Var}(\Delta Y\mid D=0)$ and $\nu_{w0s}^2 := \operatorname{Var}(\alpha_{0s}\mid D=0)$. For sensitivity analysis, \citet{huang2025variance} propose bounding $|\rho_{w0}|$ by the estimable quantity $\bar\rho_{w0} := \sqrt{1 - \operatorname{Cor}^2(\alpha_{0s}, \Delta Y\mid D=0)}$.

\subsubsection{Alternative ways to quantify selection strength for HP}

HP do not explore alternative parameterizations of $C_{w0D}^2$ in terms of odds or $\chi^2$-divergence. For completeness, we provide these representations here.

\begin{Proposition}[Alternative characterizations of $C_{w0D}^2$]
    \begin{align*}
    C_{w0D}^2 &= \frac{1 - R^2_{O_{XU} \sim O_X|1,D=0}}{R^2_{O_{XU} \sim O_X|1,D=0}} \tag*{\text{(Residual $R^2$ in Selection Odds)}} \\
    &= \frac{E[O_{XU}] - E[O_X]}{E[O_X] - O}\tag*{\text{(Average Selection Odds)}} \\
    &= \frac{\chi^2(P_{X,U|1}\|P_{X,U|0})-\chi^2(P_{X|1}\|P_{X|0})}{\chi^2(P_{X|1}\|P_{X|0})} \tag*{\text{(Distributional Imbalance)}}. 
\end{align*}
\end{Proposition}
\begin{proof}
    The corresponding proof follows steps similar to those in Corollary~\ref{thm:alt-selection}. 
\end{proof}
\begin{remark}
    For the scaling factors, note that $\nu_{w0s}^2 = \chi^2(P_{X|1}\|P_{X|0}) < \chi^2(P_{X|1}\|P_{X|0}) +1 = \nu_{0s}^2$. Therefore, when the variance of observable weights is small (i.e., $\nu_{w0s}^2$ is close to zero), estimation noise, particularly from propensity score estimation, can lead to negative estimates of $\nu_{w0s}^2$. This issue is mitigated in the uncentered parameterization.
\end{remark}

\subsubsection{Difference in parameterization}

Note that our definition of $C_{0D}^2$, which measures the strength of confounding on selection, takes a form similar to HP’s $C_{w0D}^2$. The main distinction is that HP use a centered parameterization. In the odds representation, this centering subtracts the marginal treatment odds $O$ from the expected odds $E[O_X]$ in the denominator of the selection component. As a result, this creates a zero denominator issue in HP's approach when observed covariates do not induce any group imbalance, such as in the canonical $2 \times 2$ DiD design without covariates. HP recognize this issue in their paper: \textit{``Since our sensitivity analysis is based on variances of weights, it is not equipped to address
situations in which either the observable or ideal weights are all identical (with zero variance)''} \citep[Section~2.1]{huang2025variance}.

We now provide connections between HP’s result and ours.

\begin{Proposition}[Relationship with HP] \label{prop: HS_connection} 
\par\noindent
    \begin{enumerate}
        \item[(i)] \textbf{$C_{0D}^2$ vs $C_{w0D}^2$.}
        \begin{align*}
        C_{w0D}^2 = C_{0D}^2 \times \frac{E[O_{X}\mid D=1]}{E[O_{X}\mid D=1] - O} \left(= C_{0D}^2 \times \frac{1}{1 - R^2_{\alpha_{0s} \sim 1\mid D=0}} \right). 
    \end{align*}
        \item[(ii)] \textbf{$\rho_0$ vs $\rho_{w0}$.}
        \begin{align*}
        \rho_{w0}^2 &= \rho_0^2 \times C_{0\Delta Y}^2 \times (1 - R^2_{\Delta Y \sim g_{0s}|1,D=0}). 
    \end{align*}
    \end{enumerate}
\end{Proposition}

Part (ii) shows that $\rho_{w0}$ bundles the trend component $C_{0\Delta Y}$ and the alignment component $\rho_0$, together with the identified factor $1-R^2_{\Delta Y\sim g_{0s}\mid 1,D=0}$, into a single parameter. As a result, HP's parameterization does not allow researchers to make separate plausibility judgments about the strength of confounding with the untreated outcome trend and its alignment with treatment selection.

We next examine HP's proposed upper bound on $|\rho_{w0}|$. Recall that $g_{0s}(X) = E[\Delta Y\mid X,D=0]$ attains the maximal correlation with $\Delta Y$ among all square-integrable functions of $X$ within the control group. This characterization follows from classical results on maximal correlation (see, e.g., \citet{renyi1959measures} eq. (9), or \citet{renyi1959new} Theorem 1). Formally,
\begin{align}\label{eq:maximal-cor}
    g_{0s}(X) &= \arg\max_{f \in L^2(P_{X \mid D=0})} \operatorname{Cor}^2(f(X),\Delta Y \mid D=0).
\end{align}
Using this result, we obtain the following relationship.
\begin{Proposition}[Relationship with HP's upper bound on alignment] \label{prop: HP_align_bound}
    \[\rho_{w0}^2 = 1 - \operatorname{Cor}^2(\alpha_{0s}, \Delta Y\mid D=0) \quad  \text{implies} \quad \rho_0^2C_{0\Delta Y}^2 = 1.\] 
\end{Proposition}
\begin{proof}
    By Proposition~\ref{prop: HS_connection},
    \begin{align*}
        \rho_0^2 C_{0\Delta Y}^2 &= \frac{\rho_{w0}^2}{1 - R^2_{\Delta Y \sim g_{0s}|1,D=0}}.
    \end{align*}
    Therefore, when $\rho_{w0}^2$ equals $1 - \operatorname{Cor}^2(\alpha_{0s}, \Delta Y\mid D=0)$, this implies the following
    \begin{align*}
        \rho_0^2 C_{0\Delta Y}^2 &= \frac{1 - \operatorname{Cor}^2(\alpha_{0s}, \Delta Y\mid D=0)}{1 - R^2_{\Delta Y \sim g_{0s}|1,D=0}}.
    \end{align*}
    Since $g_{0s}(X) = \arg\max_{f \in L^2(P_{X \mid D=0})} \operatorname{Cor}^2(f(X),\Delta Y \mid D=0)$, 
    \begin{align*}
        &1 = \frac{1 - \operatorname{Cor}^2(g_{0s}, \Delta Y\mid D=0)}{1 - R^2_{\Delta Y \sim g_{0s}|1,D=0}} \le \frac{1 - \operatorname{Cor}^2(\alpha_{0s}, \Delta Y\mid D=0)}{1 - R^2_{\Delta Y \sim g_{0s}|1,D=0}} = \rho_0^2 C_{0\Delta Y}^2 \le 1, \\
        \implies \quad &\rho_0^2 C_{0\Delta Y}^2 = 1.
    \end{align*}
\end{proof}

Proposition~\ref{prop: HP_align_bound} shows that attaining HP's reported upper bound requires $\rho_0^2C_{0\Delta Y}^2=1$, so both the trend and alignment components must reach their worst case values. Moreover, Proposition~\ref{prop:hp-correlation} in Section~\ref{sec:parameterization} shows that HP's reported bound coincides with the sharper upper bound only when, among untreated units, the conditional outcome trend is an affine function of the treatment odds.

\subsubsection{Difference in benchmarking}

We now show how the benchmarking approach of \citet{huang2025variance} is less conservative than ours. The parameter $k$ aims to capture how much stronger the unobserved confounder is than the observed covariate $X_j$ in explaining treatment selection. For example, setting $k=1$ should express the belief that the unobserved confounder is as strong as $X_j$. Here we show that setting both approaches to the same $k$ posits a weaker confounder in their approach than in ours, and thus results in smaller bias bounds.

First, note we can write their benchmarking multiple in the following way.

\begin{Proposition}
\label{prop: HP_bench_metric}
    \begin{align*}
        k_{0D,j}^{\text{HP}} &= \frac{1 - R^2_{O_{XU} \sim O_X\mid 1, D=0}}{R^2_{O_{XU} \sim O_X\mid 1, D=0} - R^2_{O_{XU} \sim O_{X_{-j}}\mid 1, D=0}} = \frac{1 - R^2_{O_{XU} \sim O_X\mid D=0}}{R^2_{O_{XU} \sim O_X\mid D=0} - R^2_{O_{XU} \sim O_{X_{-j}}\mid D=0}}.
    \end{align*}
\end{Proposition}
\begin{proof}
    \begin{align*}
        k_{0D,j}^{\text{HP}} &= \frac{1 - R^2_{O_{XU} \sim O_X\mid 1, D=0}}{R^2_{O_{XU} \sim O_X\mid 1, D=0} - R^2_{O_{XU} \sim O_{X_{-j}}\mid 1, D=0}} \\
        &= \frac{1 - \frac{\operatorname{Var}(O_X \mid D=0)}{\operatorname{Var}(O_{XU} \mid D=0)}}{\frac{\operatorname{Var}(O_X \mid D=0)}{\operatorname{Var}(O_{XU} \mid D=0)} - \frac{\operatorname{Var}(O_{X_{-j}} \mid D=0)}{\operatorname{Var}(O_{XU} \mid D=0)}} \\
        &= \frac{\operatorname{Var}(O_{XU} \mid D=0) - \operatorname{Var}(O_X \mid D=0)}{\operatorname{Var}(O_X \mid D=0) - \operatorname{Var}(O_{X_{-j}} \mid D=0)} \\
        &= \frac{(E[O_{XU}^2 \mid D =0]-O^2) - (E[O_X^2 \mid D =0]-O^2)}{(E[O_X^2 \mid D =0]-O^2) - (E[O_{X_{-j}}^2 \mid D =0]-O^2)} \text{ by Proposition~\ref{prop: among_odds}(2),}\\
        &= \frac{E[O_{XU}^2 \mid D =0] - E[O_X^2 \mid D =0]}{E[O_X^2 \mid D =0] - E[O_{X_{-j}}^2 \mid D =0]} \\
        &= \frac{1 - R^2_{O_{XU} \sim O_X\mid D=0}}{R^2_{O_{XU} \sim O_X\mid D=0} - R^2_{O_{XU} \sim O_{X_{-j}}\mid D=0}}. 
    \end{align*}
\end{proof}

Our benchmarking metric is,
\begin{align*}
k_{0D,j} &:= \frac{R^2_{O_{X} \sim O_{X_{-j}}\mid D=0} - R^2_{O_{XU} \sim O_{X_{-j}}\mid D=0}}{1 - R^2_{O_X \sim O_{X_{-j}}\mid D=0}},
\end{align*}
while their benchmarking metric, by Proposition~\ref{prop: HP_bench_metric}, is
\begin{align*}
k_{0D,j}^{\text{HP}} &= \frac{1 - R^2_{O_{XU} \sim O_X\mid D=0}}{R^2_{O_{XU} \sim O_X\mid D=0} - R^2_{O_{XU} \sim O_{X_{-j}}\mid D=0}}.
\end{align*}
Note that
\begin{align*}
k_{0D,j} &= k_{0D,j}^{\text{HP}} \times R^2_{O_{XU} \sim O_{X_{-j}}\mid D=0},
\end{align*}
which shows that, for the same value of $k$, their benchmarking exercise will always posit a weaker confounder than ours.

To illustrate, consider the following example.

\begin{Example}
We consider an example based on \citet{cinelli_hazlett_informal_benchmarking} and generate data as follows. Let the observed covariate $X \sim \mathcal{N}(0,1)$ and let the unobserved covariate $U \sim \mathcal{N}(0,1)$ be independent of $X$. The treatment assignment is determined by
\[
D = \mathbf{1}\left\{X/2 + U/2 + \varepsilon_D > 0 \right\},
\]
where $\varepsilon_D$ is independent and standard normal.
The observed outcome evolution is given by
\[
\Delta Y = X + U + \varepsilon_Y,
\]
where $\varepsilon_Y$ is independent and standard normal. We use a sample of size $n=10^6$ to avoid concerns about sampling error.

The propensity score is correctly specified by a probit model for $D$ on $(X,U)$, and the outcome evolution model is correctly specified by a linear regression of $\Delta Y$ on $(X,U)$. Note that in this design, $\theta = 0$, and $U$ is exactly like $X$ in terms of its strength of association with the treatment selection and the outcome evolution.

Benchmarking is intended to express a researcher’s belief about the strength of unobserved confounding relative to observed covariates. Thus, a researcher who believes that $U$ and $X$ have comparable treatment selection strength might set the benchmark to $k=1$ under either approach, using $k_{0D,j}=1$ under our method and $k_{0D,j}^{\mathrm{HP}}=1$ under HP’s. Although the true values of both quantities differ from one in this design, the purpose of the example is to compare the conclusions reached when $k=1$ is used under each approach to express the common belief that $U$ is as strong as $X$. To focus on treatment selection, we use the empirically estimated true strength of $U$ on outcome evolution in both calculations.

Figure~\ref{fig:HP_contour} compares the benchmark points obtained by setting $k=1$ under the two approaches. Note that HP’s benchmark point remains far from zero. This would lead a practitioner using HP’s rule to conclude that confounding ``comparable to $X$'' cannot explain away the effect, providing a false sense of security. By contrast, our benchmark point is more conservative and crosses zero.
\begin{figure}[t]
    \centering
    \includegraphics[width=0.8\linewidth]{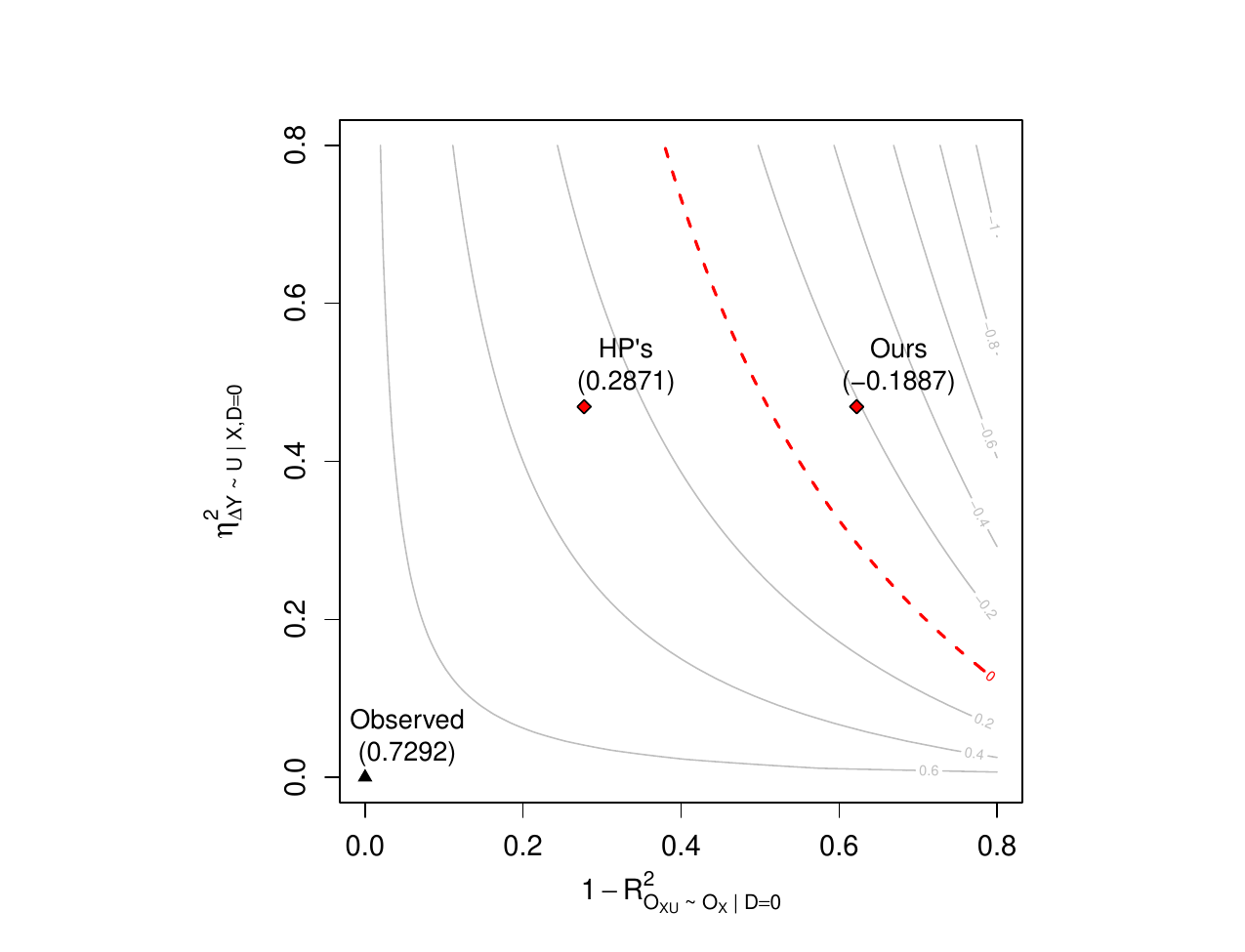}
    \caption{Sensitivity contour plots at the significance level of $\alpha = 0.05$.}
    \label{fig:HP_contour}
\end{figure}
    
\end{Example}

\subsubsection{Difference in statistical inference} \label{Appendix: Difference in Statistical Inference}

\citet{huang2025variance} focus on parametric nuisance estimators, and recommend using the percentile bootstrap to construct asymptotically valid confidence intervals for $\theta$, under pre-specified restrictions on confounding strength.  For completeness, here, we provide additional results for the centered OVB formula, which allows inference with DML. In particular,

\begingroup

\[
\theta - \theta_s = -\rho_{0}C_{0\Delta Y}C_{w0D}\sqrt{\sigma_{0s}^2\nu_{w0s}^2}, 
\]
where,
\[
\rho_0 := \operatorname{Cor}\left(g_0 - g_{0s}, O_{XU} - O_X \mid D=0\right),\quad
C^2_{0\Delta Y} := \eta^2_{\Delta Y\sim U\mid X,D=0},\quad
C^2_{w0D}:= \frac{1-R^2_{O_{XU}\sim O_{X}|1,D=0}}{R^2_{O_{XU}\sim O_{X}|1,D=0}},
\]
and  
\[
\sigma^2_{0s}:= E[\operatorname{Var}(\Delta Y \mid X,D=0)\mid D=0], \qquad 
\nu^2_{w0s}:= E\left[\left(\frac{O_X}{O}\right)^2 \middle| D=0\right] - 1.
\]
\endgroup

In this characterization, $\theta_s$ and $\sigma_{0s}^2$ can be estimated using the same scores as in \eqref{eq:psi-theta}–\eqref{eq:psi-sigma}. We derive the following debiased score for $\nu_{w0s}^2$:

\begin{align*}
    \psi_{\nu_{w0s}^2}(Z;\pi,p) 
    &= 2\frac{D}{p}\left(\frac{O_{X}}{O} - (\nu_{w0s}^2+1)\right) - \frac{1-D}{1-p}\left(\left(\frac{O_{X}}{O}\right)^2 - (\nu_{w0s}^2+1)\right).
\end{align*}

\subsubsection{Simulation: comparison with centered conditional}
 
In this section, we use Monte Carlo simulations to compare the finite sample properties of our approach with those of the ``centered'' approach introduced in Section~\ref{Appendix: Difference in Statistical Inference}, both implemented via DML. We follow the same data-generating process as in Section~\ref{sec:dgp}.

    \paragraph{Coverage and width of confidence bounds.}

    We evaluate the coverage of the confidence intervals for $\theta$ obtained from our R package \texttt{dml.sensemakr}, plugging in the true bias factors. Figure~\ref{fig: sim_coverage_centered} shows that the empirical coverage of the centered conditional approach attains the nominal levels for moderate-to-large samples, with mild under-coverage at the smallest sample size ($n=500$), most visible at the $90\%$ level. For small sample sizes (e.g., $n = 500$), estimation noise occasionally yields a negative estimate of $\nu_{w0s}^2$. Figure~\ref{fig:sim_diff_width_centered} reports the difference in the widths of the confidence intervals for $\theta$ between the uncentered and centered conditional approaches. As expected, this difference is centered around zero and concentrates as the sample size grows. 

     \begin{figure}[H]
            \centering
            \includegraphics[width=\linewidth]{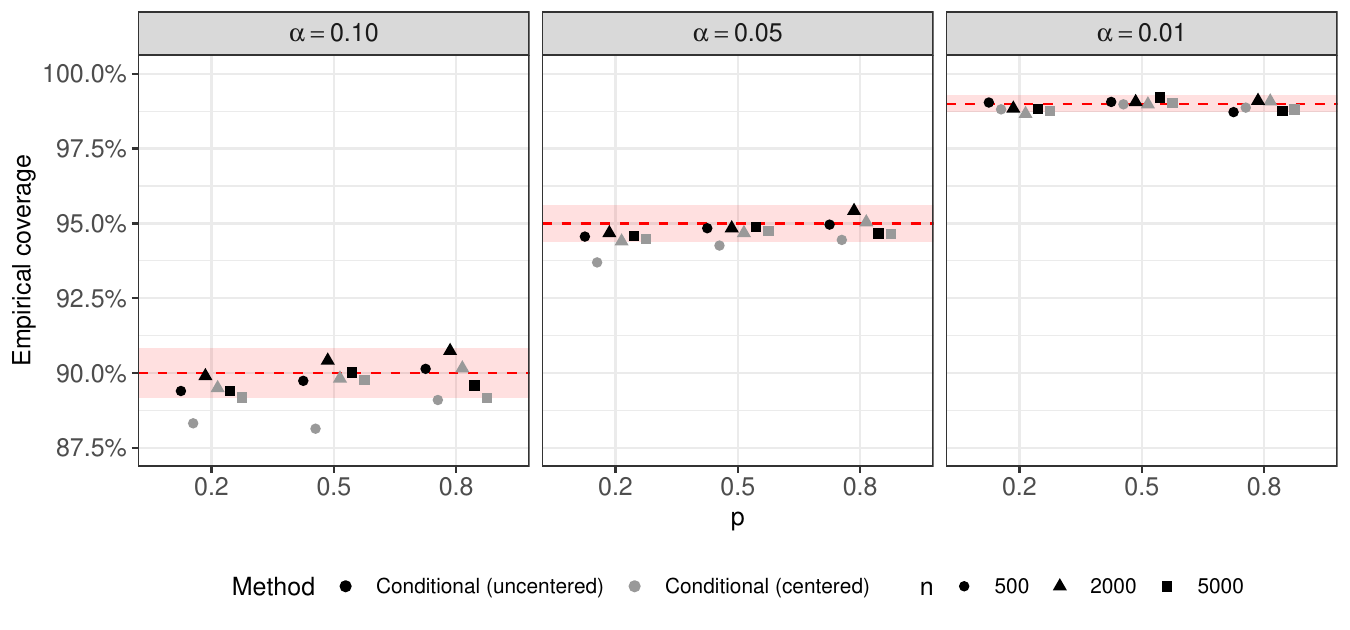}
            \caption{Empirical coverage probabilities of 90\%, 95\%, and 99\% confidence bounds across sample sizes $n$ and treatment probabilities $p$, using the true bias factors. Red dashed lines indicate nominal coverage levels; shaded regions indicate the expected (95\%) Monte Carlo error under nominal coverage.}\label{fig: sim_coverage_centered}
        \end{figure}

    \begin{figure}[H]
        \centering
        \includegraphics[width=0.8\linewidth]{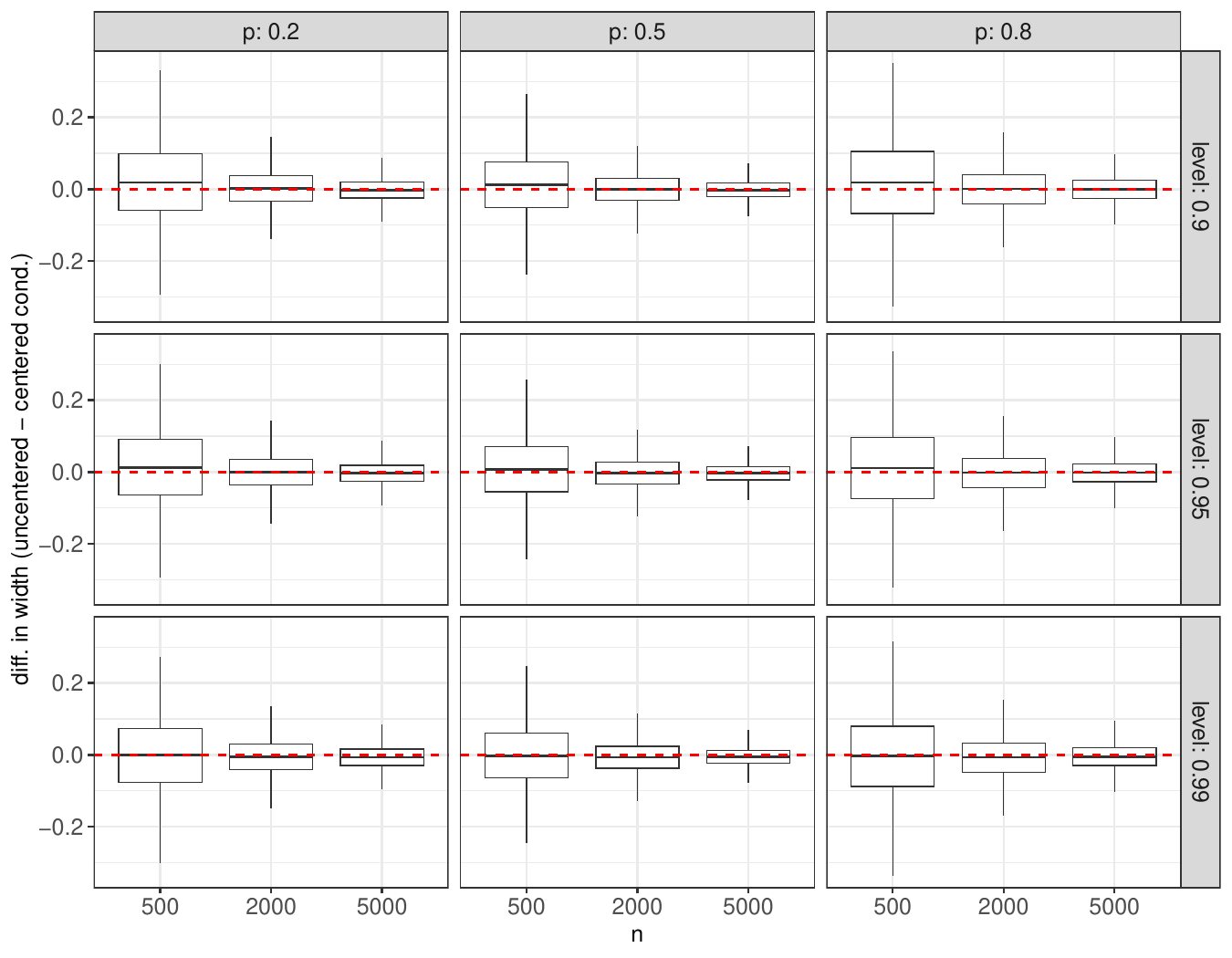}
        \caption{Boxplots of the difference in confidence-interval widths (uncentered conditional minus centered conditional) across sample sizes and treatment probabilities; the red dashed line denotes zero difference.}
        \label{fig:sim_diff_width_centered}
    \end{figure}

    \paragraph{Sensitivity statistics.}

    Figures~\ref{fig:ss_centered_point} and~\ref{fig:ss_centered_coverage} report the empirical distributions of $\text{RV}_{\theta^*=0,\alpha=1}$ and $\text{XRV}_{\theta^*=0,\alpha=1}$, as well as the empirical coverage of $\text{RV}_{\theta^*=0,\alpha}$ and $\text{XRV}_{\theta^*=0,\alpha}$. The empirical estimates of the centered $\text{RV}_{\theta^*=0,\alpha=1}$ and $\text{XRV}_{\theta^*=0,\alpha=1}$ concentrate around their respective population values, $0.7446$ and $0.6848$, with tighter concentration as the sample size increases. We also observe that empirical coverage under the centered parameterization is slightly below the nominal level in several scenarios, most notably at the $90\%$ level. The centered and uncentered statistics follow the same logic but measure robustness on different scales. The corresponding uncentered population values are only $0.4228$ and $0.2365$. The difference comes from the scaling factor alone. The centered statistics gauge confounding relative to the variance of the observable weights, $\nu_{w0s}^2 \approx 0.166$, while the uncentered statistics use their second moment, $\nu_{0s}^2 \approx 1.166$, so for the same bias the centered parameterization inflates the robustness value relative to the uncentered one. For this reason, sensitivity statistics should only be compared within the same parameterization.
    
\begin{figure}[H]
    \centering
    \includegraphics[width=\linewidth]{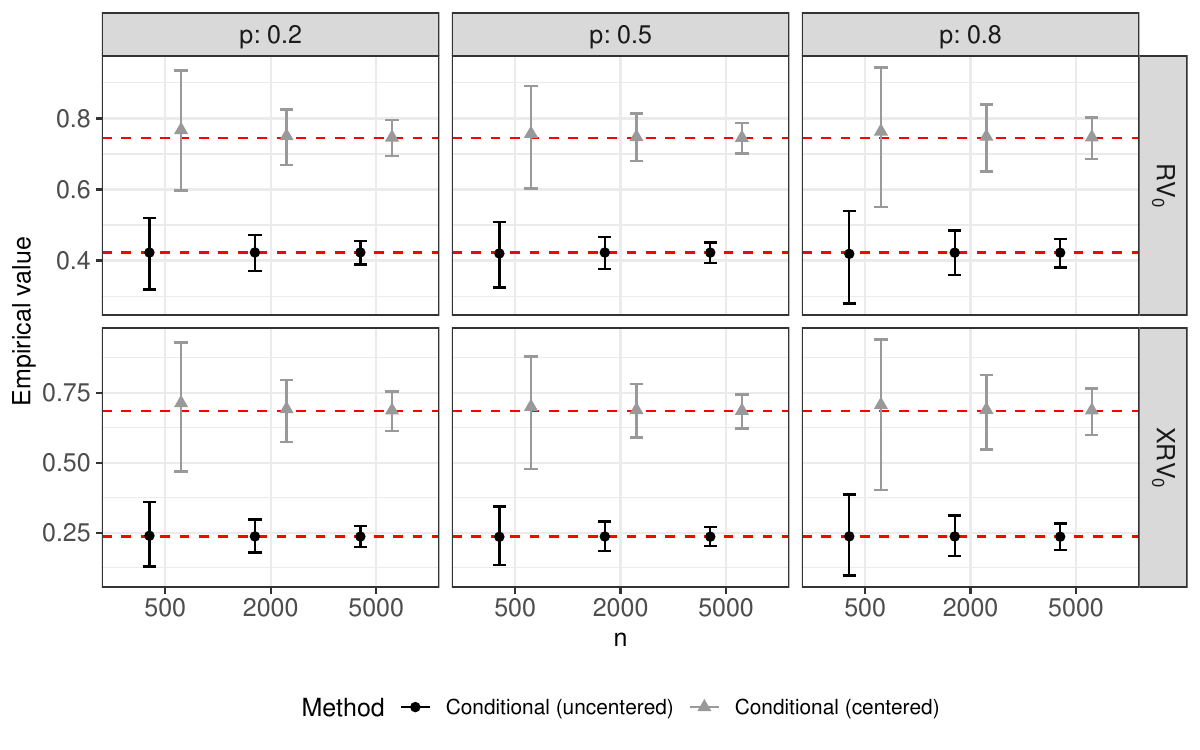}
    \caption{Finite sample behavior of $\text{RV}_{\theta^*=0,\alpha=1}$ and $\text{XRV}_{\theta^*=0,\alpha=1}$ for the centered and uncentered conditional approaches across sample sizes and treatment probabilities (up to $5{,}000$ repetitions per scenario). Points denote Monte Carlo means; whiskers delimit the empirical $2.5$th and $97.5$th percentiles.}
    \label{fig:ss_centered_point}
\end{figure}

\begin{figure}[H]
    \centering
    \includegraphics[width=\linewidth]{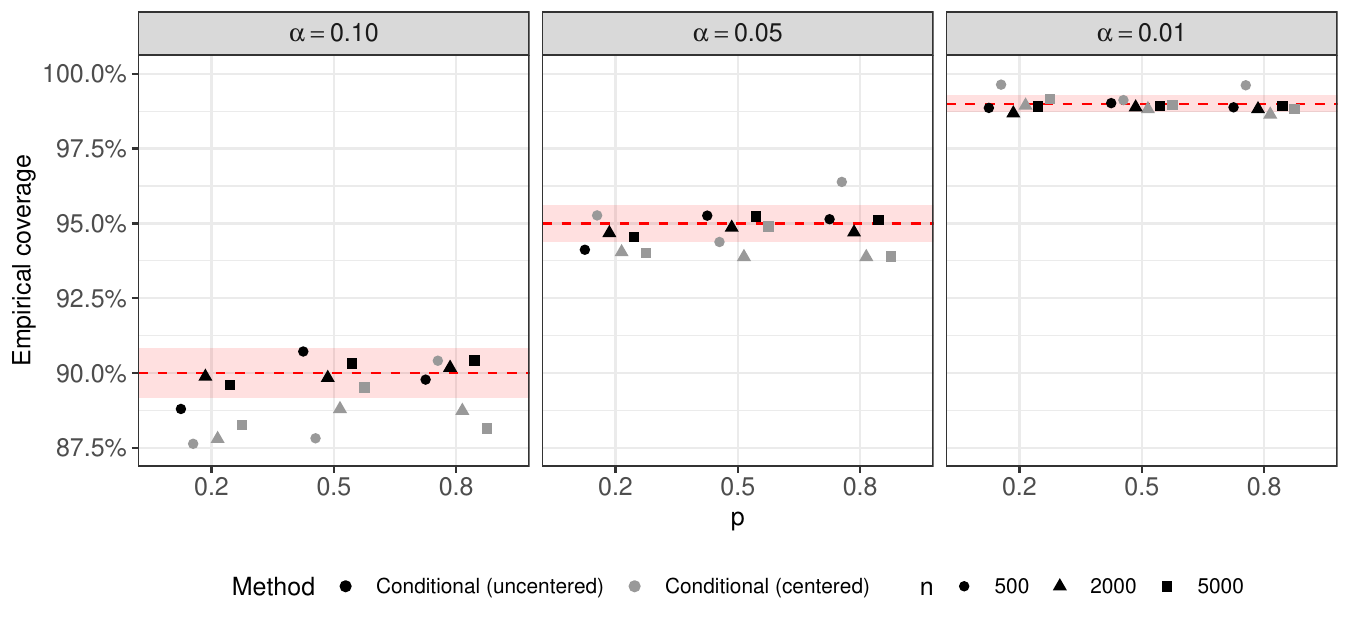}
    \caption{Empirical coverage of $\text{RV}_{\theta^*=0,\alpha}$ and $\text{XRV}_{\theta^*=0,\alpha}$ for the centered and uncentered conditional approaches across sample sizes and treatment probabilities (up to $5{,}000$ repetitions per scenario). Red dashed lines indicate the nominal coverage levels $1-\alpha$; shaded regions indicate the expected (95\%) Monte Carlo error under nominal coverage.}
    \label{fig:ss_centered_coverage}
\end{figure}

\subsection{Relationship to pre-trend extrapolation}
\label{app:pre-trend-comparison}

This section compares pre-trend extrapolation \citep{rambachan2023more} with the OVB formula. The main idea of the former is to extrapolate post-treatment deviations from parallel trends using pre-treatment information. 
Under standard DiD identification assumptions, the following result provides the bridge.

\begin{Proposition}[Parallel trends violation as omitted variable bias]\label{lemma:Parallel trends violation as omitted variable bias}~

Under no anticipation and consistency, 
\[
\text{ATT} - \theta_s = -\delta, \ \text{with} \ \delta := E[E[\Delta Y(0)\mid X, D=1] - E[\Delta Y(0)\mid X, D=0]\mid D=1],
\]
where $\delta$ is the violation of the conditional parallel trends assumption, averaged over the $X$ distribution of the treated group.
    
If, in addition, the conditional parallel trends assumption holds given $(X,U)$, then 
\[
\underbrace{E\left[E[\Delta Y(0) \mid X,D=1] - E[\Delta Y(0) \mid X,D=0]\mid D=1\right]}_{=:\delta \ (\text{deviation from parallel trends})} = \underbrace{E[g_0(X,U)\mid D=1] - E[g_{0s}(X)\mid D=1]}_{ = \theta_0-\theta_{0s} \  (\text{bias from omitting $U$})}.
\]
\end{Proposition}

\begin{proof}
By the proof of Proposition~\ref{lemma:Identification}, if no anticipation holds, we have 
\begin{align*}
    \text{ATT} := E\left[Y_2(1) - Y_2(0)\middle|D=1\right] = E[\Delta Y(1)\mid D=1]-E[\Delta Y(0) \mid D=1]. 
\end{align*}
Then, by consistency, 
\begin{align*}
    \text{ATT} - \theta_s &= E[\Delta Y(1)\mid D=1]-E[\Delta Y(0) \mid D=1] - E[\Delta Y - g_{0s} \mid D=1] \\
    &= E[\Delta Y \mid D=1] - E[\Delta Y(0) \mid D=1] - E[\Delta Y \mid D=1] + E[E[\Delta Y(0) \mid X,D=0] \mid D=1]\\
    &= -E[E[\Delta Y (0)\mid X,D=1] \mid D=1] + E[E[\Delta Y(0) \mid X,D=0] \mid D=1] \text{ by LTE,} \\
    &= - E\left[E[\Delta Y(0) \mid X,D=1] - E[\Delta Y(0) \mid X,D=0]\mid D=1\right] \\
    &= -\delta. 
\end{align*}
This matches Example 2.2.3 of \citet{rambachan2023more} where $\delta$ is the violation of conditional parallel trends, averaged over the $X$ distribution of the treated group. As a special case, if one instead uses the unconditional untreated mean trend $E[\Delta Y \mid D=0]$ (in place of $g_{0s}$) as a proxy for the treated group's counterfactual trend, the deviation from parallel trends reduces to $\delta = E[\Delta Y(0)\mid D=1] - E[\Delta Y(0)\mid D=0]$, which matches Example 2.2.1 of \citet{rambachan2023more}.

Assume further that the parallel trends assumption additionally holds conditional on $(X,U)$, 
\begin{align*}
    \delta &:= E\left[E[\Delta Y(0) \mid X,D=1] - E[\Delta Y(0) \mid X,D=0]\mid D=1\right] \\
    &= E[\Delta Y(0) \mid D=1] - E[g_{0s}(X) \mid D =1] \text{ by consistency,} \\
    &= E[\Delta Y(0) \mid D=1] - E[g_0(X,U) \mid D=1] + E[g_0(X,U) \mid D=1] - E[g_{0s}(X) \mid D =1] \\
    &= E[\Delta Y(0) \mid D=1] - E[g_0(X,U) \mid D=1] + (\theta_0 - \theta_{0s}) \\
    &= E[E[\Delta Y(0) \mid X,U,D=1]\mid D=1] - E[g_0(X,U) \mid D=1] + (\theta_0 - \theta_{0s}) \text{ by LTE,} \\
    &= E[E[\Delta Y(0) \mid X,U,D=0]\mid D=1] - E[g_0(X,U) \mid D=1] + (\theta_0 - \theta_{0s}) \text{ by conditional PTA,} \\
    &= E[g_0(X,U) \mid D=1] - E[g_0(X,U) \mid D=1] + (\theta_0 - \theta_{0s}) \text{ by consistency,} \\
    &= \theta_0 - \theta_{0s}. 
\end{align*}

\end{proof}

Note that pre-trend extrapolation is a partial identification exercise where confounding restrictions are imposed using pre-treatment information.
In contrast, we treat confounding strength as a sensitivity parameter and explicitly trace how identification and inference change as this strength varies. These two approaches are complementary. Pre-trend extrapolation can be used to motivate a range of plausible confounding strengths based on restrictions on deviations from parallel trends, while our results help interpret and translate a given confounding strength into the corresponding magnitude of deviations from parallel trends. 

We now highlight two primary limitations.
\begin{itemize}
    \item \textbf{Not applicable to the canonical two-group, two-period DiD:} Pre-trend extrapolation requires pre-treatment information to infer post-treatment deviations from parallel trends. In contrast, our results do not require pre-trend information and can serve as a general building block for sensitivity analysis of the ATT in both two-period and multi-period designs \citep{callaway2021difference}. In this way, we can see pre-trend extrapolation as one additional way of benchmarking confounding strength, as discussed in Section~\ref{sec:benchmarking}.
    
    \item \textbf{Sensitive to mismatches between pre- and post-treatment deviations:} Even when pre-treatment trend information is available, the resulting restrictions may not adequately constrain deviations from parallel trends in the post-treatment period.
    Pre-treatment trends may fail to detect post-treatment violations of parallel trends, and conversely, pre-treatment deviations may exist even when post-treatment parallel trends holds.
    Our approach provides additional ways to gauge the magnitude of the bias, such as, for example, comparing it with the bias induced by observed covariates. This provides a more complete picture of the robustness of the result, and also allows us to put the pre-trend bias itself into better context.
\end{itemize}

\subsection{Deferred proofs}

\begin{proof}[Proof of Proposition~\ref{prop: CB_D_equivalent}]
    We first construct some important relationships:
    \begin{align*}
        (1): \ R^2_{\alpha \sim \alpha_s} &:= \frac{\operatorname{Var}(\alpha_s)}{\operatorname{Var}(\alpha)} = \frac{E[\alpha_s^2]}{E[\alpha^2]} \\
        &\overset{\text{LTE}}{=} \frac{E[E[\alpha_s^2\mid D]]}{E[E[\alpha^2\mid D]]} = \frac{pE[\alpha_s^2\mid D=1] + (1-p) E[\alpha_s^2\mid D=0]}{pE[\alpha^2\mid D=1] + (1-p) E[\alpha^2\mid D=0]} \\
        &= \frac{\frac{1}{p} + \frac{1}{1-p} \times E[\alpha_{0s}^2\mid D=0]}{\frac{1}{p} + \frac{1}{1-p} \times E[\alpha_{0}^2\mid D=0]} \\
        &= \frac{R^2_{O_{XU} \sim O|D=0} + O \times R^2_{O_{XU} \sim O_X|D=0}}{R^2_{O_{XU} \sim O|D=0} + O} \\
        &= \frac{R^2_{O_{XU} \sim O|D=0}}{R^2_{O_{XU} \sim O|D=0} + O} \times 1 + \frac{O}{R^2_{O_{XU} \sim O|D=0} + O} \times R^2_{O_{XU} \sim O_X|D=0}, \\
        (2): \  R^2_{\alpha \sim \alpha_s} &=  \frac{\frac{1}{p} + \frac{1}{1-p} \times (E[\alpha_{0s}^2\mid D=0] -1) +  \frac{1}{1-p}}{\frac{1}{p} + \frac{1}{1-p} \times (E[\alpha_{0}^2\mid D=0] -1) +  \frac{1}{1-p}} \\ 
        &= \frac{\frac{1}{p} +\operatorname{Var}(\alpha_{0s}|D=0)}{\frac{1}{p} + \operatorname{Var}(\alpha_0|D=0)} \\
        &= \frac{\frac{1}{p} + \frac{(1-p)^2}{p^2}\operatorname{Var}(O_X \mid D=0)}{\frac{1}{p} + \frac{(1-p)^2}{p^2}\operatorname{Var}(O_{XU} \mid D=0)}\\
        &=\frac{p + (1-p)^2\operatorname{Var}(O_X \mid D=0)}{p + (1-p)^2\operatorname{Var}(O_{XU} \mid D=0)}\\
        (3): \ R^2_{\alpha \sim \alpha_s} &=  \frac{\frac{1}{p} + \frac{1}{1-p} \times E[\alpha_{0s}^2\mid D=0]}{\frac{1}{p} + \frac{1}{1-p} \times E[\alpha_{0}^2\mid D=0]}  \\
        &\overset{\text{def}}{=} \frac{O + E[O_X^2\mid D=0]}{O + E[O_{XU}^2\mid D=0]} \\
        &= \frac{O + \frac{1}{1-p}(E[O_X] - p)}{O + \frac{1}{1-p}(E[O_{XU}] - p)} \text{ by Proposition \ref{prop: among_odds} (4)}, \\
        &= \frac{E[O_X]}{E[O_{XU}]}. 
    \end{align*}

    Accordingly, 
    \begin{align*}
        \text{By (1):} \ C_{D}^2 &:= \frac{1 - R^2_{\alpha \sim \alpha_s}}{R^2_{\alpha \sim \alpha_s}} \\
        &= \frac{ 1 - R^2_{O_{XU} \sim O_X|D=0}}{\frac{1}{O}R^2_{O_{XU} \sim O|D=0} + R^2_{O_{XU} \sim O_X|D=0}}. \\
        \text{Additionally, } &1 - R^2_{\alpha \sim \alpha_s} = \frac{1 - R^2_{O_{XU} \sim O_X|D=0}}{\frac{1}{O}R^2_{O_{XU} \sim O|D=0} + 1}. \\
        \text{By (2):} \ C_{D}^2 &:= \frac{1 - R^2_{\alpha \sim \alpha_s}}{R^2_{\alpha \sim \alpha_s}} \\
        &= \frac{ 1 - \frac{p + (1-p)^2\operatorname{Var}(O_X \mid D=0)}{p + (1-p)^2\operatorname{Var}(O_{XU} \mid D=0)}}{\frac{p + (1-p)^2\operatorname{Var}(O_X \mid D=0)}{p + (1-p)^2\operatorname{Var}(O_{XU} \mid D=0)}} \\
        &= \frac{\operatorname{Var}(O_{XU} \mid D=0) - \operatorname{Var}(O_X \mid D=0)}{\frac{1}{p}\times O^2 + \operatorname{Var}(O_X \mid D=0)}\\
        &= \frac{1 - R^2_{O_{XU} \sim O_X \mid 1,D=0}}{\frac{1}{p \times CV^2_{O_{XU}|0}} + R^2_{O_{XU} \sim O_X \mid 1,D=0}}. \\
        \text{Additionally, } &1 - R^2_{\alpha \sim \alpha_s} = \frac{1 - R^2_{O_{XU} \sim O_X \mid 1,D=0}}{\frac{1}{p \times CV^2_{O_{XU}|0}} + 1}. \\
        \text{By (3):} \ C_{D}^2 &:= \frac{1 - R^2_{\alpha \sim \alpha_s}}{R^2_{\alpha \sim \alpha_s}} \\
        &= \frac{1 - \frac{E[O_X]}{E[O_{XU}]}}{\frac{E[O_X]}{E[O_{XU}]}} = \frac{E[O_{XU}] - E[O_{X}]}{E[O_{X}]} \\
         &= \frac{E[O_{XU}\mid D=1] - E[O_{X}\mid D=1]}{E[O_{X}\mid D=1] + 1} \text{ by Proposition~\ref{prop: among_odds} (4),} \\
         &= \frac{\chi^2(P_{X,U|1}\|P_{X,U|0}) - \chi^2(P_{X|1}\|P_{X|0})}{\chi^2(P_{X|1}\|P_{X|0}) + \frac{1}{p}} \text{ by Proposition~\ref{prop: odds_and_divergence} (2).} \\
         \text{Additionally, } &1 - R^2_{\alpha \sim \alpha_s} = \frac{\chi^2(P_{X,U|1}\|P_{X,U|0}) - \chi^2(P_{X|1}\|P_{X|0})}{\chi^2(P_{X,U|1}\|P_{X,U|0}) + \frac{1}{p}}.
    \end{align*}
\end{proof}

\begin{proof}[Proof of Proposition~\ref{prop: CB_connection}]
    \par\noindent
    \begin{enumerate}
        \item[(i)]
            Define $g_{d} := E[\Delta Y \mid X,U,D=d]$ and $g_{ds} := E[\Delta Y \mid X,D=d]$. Note that $g_{d}(X,U) = g(D=d,X,U)$ and $g_{ds}(X) = g_s(D=d,X)$, for $d \in \{0,1\}$. 
    \begin{align*}
        C_{\Delta Y}^2 &:= \frac{\operatorname{Var}(g-g_s)}{\operatorname{Var}(\Delta Y - g_s)}  \\
        &\overset{\text{LTV}}{=} \frac{E\left[\operatorname{Var}(g-g_s\mid D)\right] + \operatorname{Var}(E[g-g_s\mid D])}{E\left[\operatorname{Var}(\Delta Y-g_s\mid D)\right] + \operatorname{Var}(E[\Delta Y-g_s\mid D])} \\
        &\overset{\text{LTE}}{=} \frac{E\left[\operatorname{Var}(g-g_s\mid D)\right]}{E\left[\operatorname{Var}(\Delta Y-g_s\mid D)\right]} \\
        &= \frac{p \times \operatorname{Var}(g-g_s\mid D=1) + (1-p) \times \operatorname{Var}(g-g_s\mid D=0)}{p \times \operatorname{Var}(\Delta Y-g_s\mid D=1) + (1-p) \times \operatorname{Var}(\Delta Y-g_s\mid D=0)} \\
        &= \frac{p \times \operatorname{Var}(g_1-g_{1s}\mid D=1) + (1-p) \times \operatorname{Var}(g_0-g_{0s}\mid D=0)}{p \times \operatorname{Var}(\Delta Y-g_{1s}\mid D=1) + (1-p) \times \operatorname{Var}(\Delta Y-g_{0s}\mid D=0)} \\
        &= \frac{p \times \operatorname{Var}(g_1-g_{1s}\mid D=1)}{p \times \sigma_{1s}^2 + (1-p) \times \sigma_{0s}^2} + \frac{(1-p) \times \operatorname{Var}(g_0-g_{0s}\mid D=0)}{p \times \sigma_{1s}^2 + (1-p) \times \sigma_{0s}^2} \\
        &= \frac{p \times \operatorname{Var}(\Delta Y-g_{1s}\mid D=1)}{p \times \sigma_{1s}^2 + (1-p) \times \sigma_{0s}^2} \times \frac{\operatorname{Var}(g_1-g_{1s}\mid D=1)}{\operatorname{Var}(\Delta Y-g_{1s}\mid D=1)} + \\ \nonumber 
        &\phantom{=} \hspace{2cm}
        \frac{(1-p) \times \operatorname{Var}(\Delta Y-g_{0s}\mid D=0)}{p \times \sigma_{1s}^2 + (1-p) \times \sigma_{0s}^2} \times \frac{\operatorname{Var}(g_0-g_{0s}\mid D=0)}{\operatorname{Var}(\Delta Y-g_{0s}\mid D=0)} \\
        &= \frac{p \times \sigma_{1s}^2}{p \times \sigma_{1s}^2 + (1-p) \times \sigma_{0s}^2} \times C_{1\Delta Y}^2 + \frac{(1-p) \times \sigma_{0s}^2}{p \times \sigma_{1s}^2 + (1-p) \times \sigma_{0s}^2} \times C_{0\Delta Y}^2 \\
        &= W_{0\Delta Y}C_{0\Delta Y}^2 + (1 - W_{0\Delta Y})C_{1\Delta Y}^2. 
    \end{align*}
    The 6th and 8th equalities are established by observing that for $d \in \{0,1\}$, 
    \begin{align*}
        &\operatorname{Var}(\Delta Y - g_{ds}\mid D=d) \\
        &\overset{\text{LTV}}{=} E[\operatorname{Var}(\Delta Y \mid X,D=d)\mid D=d] + \operatorname{Var}(E[\Delta Y - g_{ds} \mid X,D=d]\mid D=d) \\
        &\overset{\text{LTE}}{=} E[\operatorname{Var}(\Delta Y \mid X,D=d)\mid D=d] \\
        &\overset{\text{def.}}{=} \sigma_{ds}^2. 
    \end{align*}
    
        \item[(ii)]
        \begin{align*}
        C_D^2 &:= \frac{1  - R^2_{\alpha \sim \alpha_s}}{R^2_{\alpha \sim \alpha_s}} = \frac{E[O_{XU}] - E[O_{X}]}{E[O_{X}]} \\
        &= \frac{p \times \left(E[O_{XU}\mid D=1] - E[O_{X}\mid D=1]\right)}{p \times (E[O_{X}\mid D=1] + 1)} \text{ by Proposition~\ref{prop: among_odds} (4),} \\
        &= \frac{E[O_{XU}\mid D=1] - E[O_{X}\mid D=1]}{E[O_{X}\mid D=1] + 1} \\
        &= \frac{E[O_{XU}\mid D=1] - E[O_{X}\mid D=1]}{E[O_{X}\mid D=1]} \times \frac{E[O_{X}\mid D=1]}{E[O_{X}\mid D=1] + 1} \\
        &\overset{\text{def.}}{=} C_{0D}^2 \times \frac{E[O_{X}\mid D=1]}{E[O_{X}\mid D=1] + 1} \\
        &= C_{0D}^2 \times \frac{O\times \nu_{0s}^2}{O\times \nu_{0s}^2 + 1}, 
    \end{align*}
    where the last equality uses $\nu_{0s}^2 = \frac{E[O_{X}\mid D=1]}{O}$, which was established in Corollary~\ref{thm:alt-selection}. Furthermore, to better separate the confounding strength of interest from the sampling mechanism, we use $\nu_{0s}^2 = \chi^2(P_{X|1}\|P_{X|0}) + 1$, which was established in Corollary~\ref{thm:alt-selection}. Then, the following relationship holds:
    \begin{align*}
        C_D^2 &= C_{0D}^2 \times \frac{O\times (\chi^2(P_{X|1}\|P_{X|0}) + 1)}{O\times (\chi^2(P_{X|1}\|P_{X|0}) + 1) + 1}. 
    \end{align*}
    
        \item[(iii)]
        We begin by constructing several key relationships between the covariances in the conditional and unconditional results.
    \begin{align*}
        \operatorname{Cov}(g-g_s,\alpha - \alpha_s) &= E[(g-g_s)(\alpha - \alpha_s)] \overset{\text{LTE}}{=} E\left[E[(g-g_s)(\alpha - \alpha_s)\mid D]\right] \nonumber \\
        &= pE[(g-g_s)(\alpha - \alpha_s)\mid D=1] + (1-p)E[(g-g_s)(\alpha - \alpha_s)\mid D=0]\nonumber \\
        &= (1-p) E\left[(g_0 - g_{0s}) \times (-O_{XU} + O_{X}) \times \frac{1}{p} \mid D=0\right]\nonumber \\
        &= \frac{1}{O} \times E\left[(g_0 - g_{0s})\left(-O_{XU} + O_{X}\right) \mid D=0\right]\nonumber \\
        &= -E[(g_0 - g_{0s})(\alpha_0 - \alpha_{0s})\mid D=0] \\
        &= -\operatorname{Cov}(g_0 - g_{0s}, \alpha_0 - \alpha_{0s}\mid D=0), \\
        \operatorname{Var}(g-g_s) &\overset{\text{LTV}}{=} E[\operatorname{Var}(g-g_s\mid D)] + \operatorname{Var}(E[g-g_s\mid D]) \\
        &\overset{\text{LTE}}{=}  E[\operatorname{Var}(g-g_s\mid D)] \nonumber \\
        &= \operatorname{Var}(g-g_s\mid D=1)p + \operatorname{Var}(g-g_s\mid D=0) (1-p) \nonumber \\
        &= \operatorname{Var}(g_1-g_{1s}\mid D=1)p + \operatorname{Var}(g_0-g_{0s}\mid D=0) (1-p), \nonumber \\
        \operatorname{Var}(\alpha - \alpha_s) &\overset{\text{LTV}}{=} E[\operatorname{Var}(\alpha-\alpha_s\mid D)] + \operatorname{Var}(E[\alpha - \alpha_s\mid D]) \nonumber \\
        &\overset{\text{LTE}}{=} E[\operatorname{Var}(\alpha-\alpha_s\mid D)] \nonumber \\
        &=  0 \times p + \operatorname{Var}(\alpha-\alpha_s\mid D=0) \times (1-p) \nonumber \\
        &= \operatorname{Var}\left(\frac{-O_{XU} + O_{X}}{p} \mid D=0\right) \times (1-p) \nonumber \\
        &=\operatorname{Var}\left(\frac{O_{XU} - O_{X}}{O} \mid D=0\right) \times \frac{1}{1-p} \nonumber \\
        &\overset{\text{def.}}{=} \frac{1}{1-p}\operatorname{Var}(\alpha_0 - \alpha_{0s}\mid D=0),
    \end{align*}
    
    Accordingly, the following holds:
    \begin{align}
        \rho^2 &:= \frac{\operatorname{Cov}^2(g-g_s, \alpha - \alpha_s)}{\operatorname{Var}(g-g_s)\operatorname{Var}(\alpha -\alpha_s)} \nonumber \\
        &= \frac{\operatorname{Cov}^2(g_0 - g_{0s}, \alpha_0 - \alpha_{0s}\mid D=0)}{\left(\operatorname{Var}(g_1-g_{1s}\mid D=1)p + \operatorname{Var}(g_0-g_{0s}\mid D=0) (1-p)\right) \times \left(\frac{1}{1-p}\operatorname{Var}(\alpha_0 - \alpha_{0s}\mid D=0)\right)} \nonumber \\
        &\overset{\text{def.}}{=} \frac{\rho_0^2}{\frac{\operatorname{Var}(g_1-g_{1s}\mid D=1)}{\operatorname{Var}(g_0-g_{0s}\mid D=0)} \times O + 1} \left(= \frac{\rho_0^2}{\frac{C_{1\Delta Y}^2\sigma_{1s}^2}{C_{0\Delta Y}^2\sigma_{0s}^2} \times O + 1}\right). \label{eq:rho-unc-explicit}
    \end{align}
    We implicitly assume that $\operatorname{Var}(g_0-g_{0s}\mid D=0), \operatorname{Var}(\alpha_0 - \alpha_{0s}\mid D=0) > 0$; otherwise, the bias is zero.

        \item[(iv)] \textbf{Supplement:} Finally, combining the results above, we show that the unconditional result reduces to the conditional result. We begin by deriving two important relationships for the scaling factors:
    \begin{align*}
        \text{(1): }\operatorname{Var}(\Delta Y - g_s) &= E[(\Delta Y - g_s)^2] \overset{\text{LTE}}{=} E\left[E[(\Delta Y - g_s)^2\mid D]\right]\\
        &=pE[(\Delta Y - g_{1s})^2\mid D=1] + (1-p)E[(\Delta Y - g_{0s})^2\mid D=0]\nonumber \\
        &= p\operatorname{Var}(\Delta Y - g_{1s}\mid D=1) + (1-p)\operatorname{Var}(\Delta Y - g_{0s}\mid D=0), \\
        \text{(2): }\operatorname{Var}(\alpha_s) &= E[\alpha_s^2] \overset{\text{LTE}}{=} E\left[E[\alpha_s^2\mid D]\right]\nonumber \\
        &= p \times E[\alpha_s^2\mid D=1] + (1-p) \times E[\alpha_s^2\mid D=0] \nonumber \\
        &\overset{\text{def.}}{=} \frac{1}{p} + \frac{1}{1-p} \times E[\alpha_{0s}^2\mid D=0]. 
    \end{align*}
    The second equality implies that $\operatorname{Var}(\alpha_s) = E[\alpha_s^2] > E[\alpha_{0s}^2\mid D=0]$, with the gap increasing as $p$ approaches zero or one. Accordingly, we have the following:
    \begin{align}
        &\rho^2C_{\Delta Y}^2C_D^2S^2 \nonumber \\
        &= \frac{\rho_0^2}{\frac{\operatorname{Var}(g_1-g_{1s}\mid D=1)O}{\operatorname{Var}(g_0-g_{0s}\mid D=0)} + 1} \times ((1-W_{0\Delta Y})C_{1\Delta Y}^2 + W_{0\Delta Y}C_{0\Delta Y}^2) \times (W_{0D}C_{0D}^2) \times \operatorname{Var}(\Delta Y - g_s)\operatorname{Var}(\alpha_s) \nonumber \\
        &= \rho_0^2 \times \operatorname{Var}(g_0-g_{0s}\mid D=0) \times \frac{O \times \operatorname{Var}(g_1 - g_{1s}\mid D=1) + \sigma_{0s}^2C_{0\Delta Y}^2}{O \times \operatorname{Var}(g_1 - g_{1s}\mid D=1) + \operatorname{Var}(g_0-g_{0s}\mid D=0)} \times C_{0D}^2 \times \nu_{0s}^2 \nonumber \\
        &= \rho_0^2 \times \sigma_{0s}^2C_{0\Delta Y}^2 \times \frac{O\operatorname{Var}(g_1 - g_{1s}\mid D=1) +\sigma_{0s}^2C_{0\Delta Y}^2}{O\operatorname{Var}(g_1 - g_{1s}\mid D=1) + \sigma_{0s}^2C_{0\Delta Y}^2} \times C_{0D}^2 \times \nu_{0s}^2 \nonumber \\
        &= \rho_0^2 C_{0\Delta Y}^2 C_{0D}^2 \sigma_{0s}^2 \nu_{0s}^2 = \rho_0^2 C_{0\Delta Y}^2 C_{0D}^2 S_0^2, \label{eq:unc-cond-product}
    \end{align}
    where the first equality follows from (i),(ii),(iii), the second equality follows from (1) and (2), and the third equality uses $\operatorname{Var}(g_0-g_{0s}\mid D=0) = \sigma_{0s}^2C_{0\Delta Y}^2$. 

   An alternative way to establish this equivalence is to use $\operatorname{Cov}(g-g_s,\alpha - \alpha_s) = -\operatorname{Cov}(g_0 - g_{0s}, \alpha_0 - \alpha_{0s}\mid D=0)$ directly, which is shown as part of the proof of (iii). 
    \end{enumerate}
\end{proof}

\begin{proof}[Proof of Proposition~\ref{prop: HS_connection}]
    \par\noindent
    \begin{enumerate}
        \item[(i)] Note that 
\begin{align*}
    R^2_{\alpha_{0s} \sim 1\mid D=0} &= 1 - \frac{E[(\alpha_{0s} - 1)^2\mid D=0]}{E[\alpha_{0s}^2\mid D=0]} \\
    &= \frac{1}{E[\alpha_{0s}^2\mid D=0]} = \frac{O}{E[O_{X}\mid D=1]}. 
\end{align*}

Accordingly, we have
    \begin{align*}
        C_{w0D}^2 &:= \frac{1 - R^2_{\alpha_0 \sim \alpha_{0s}|1,D=0}}{R^2_{\alpha_0 \sim \alpha_{0s}|1,D=0}} = \frac{\operatorname{Var}(\alpha_0\mid D=0) - \operatorname{Var}(\alpha_{0s}\mid D=0)}{\operatorname{Var}(\alpha_{0s}\mid D=0)} \nonumber \\
        &= \frac{E[\alpha_0^2\mid D=0] - E[\alpha_{0s}^2\mid D=0]}{E[\alpha_{0s}^2\mid D=0] - 1} \\
        &= \frac{1 - R^2_{\alpha_0 \sim \alpha_{0s}\mid D=0}}{R^2_{\alpha_0 \sim \alpha_{0s}\mid D=0} - \frac{1}{E[\alpha_0^2\mid D=0]}} \\
        &\overset{\text{def.}}{=} C_{0D}^2 \times \frac{R^2_{\alpha_0 \sim \alpha_{0s}\mid D=0}}{R^2_{\alpha_0 \sim \alpha_{0s}\mid D=0} - \frac{1}{E[\alpha_0^2\mid D=0]}} \\
        &= C_{0D}^2 \times \frac{\frac{E[O_{X}\mid D=1]}{E[O_{XU}\mid D=1]}}{\frac{E[O_{X}\mid D=1]}{E[O_{XU}\mid D=1]} - \frac{O}{E[O_{XU}\mid D=1]}} \\
        &= C_{0D}^2 \times \frac{E[O_{X}\mid D=1]}{E[O_{X}\mid D=1] - O} \left(= C_{0D}^2 \times \frac{1}{1 - R^2_{\alpha_{0s} \sim 1\mid D=0}} \right). 
    \end{align*}
       These equalities follow from the results established in the proof of Corollary~\ref{thm:alt-selection}.

        \item[(ii)]  We begin by constructing several important relationships:
    \begin{align*}
        \text{(1): } &\operatorname{Cov}(g_0 - g_{0s}, \alpha_0 - \alpha_{0s}\mid D=0) = E[(g_0 - g_{0s})(\alpha_{0} - \alpha_{0s})\mid D=0] \nonumber \\
        &= E[g_0(X,U)\alpha_{0}(X,U)\mid D=0] - E[g_{0s}(X)\alpha_{0s}(X)\mid D=0] \nonumber \\
        &= E[E\left[\Delta Y \mid X, U, D=0\right]\alpha_{0}(X,U)\mid D=0] - E[E\left[\Delta Y \mid X, D=0\right]\alpha_{0s}(X)\mid D=0] \nonumber \\
        &= E[E\left[\alpha_{0}(X,U)\Delta Y \mid X, U, D=0\right]\mid D=0] - E[E\left[\alpha_{0s}(X)\Delta Y \mid X, D=0\right]\mid D=0] \nonumber \\
        &= E\left[\alpha_{0}(X,U)\Delta Y\mid D=0\right] - E\left[\alpha_{0s}(X)\Delta Y\mid D=0\right] \nonumber \\
        &= E[(\alpha_{0}- \alpha_{0s})\Delta Y\mid D=0] - \underbrace{E[(\alpha_{0} - \alpha_{0s})\mid D=0]E[\Delta Y\mid D=0]}_{=0} \nonumber \\
        &= \operatorname{Cov}((\alpha_{0} - \alpha_{0s}), \Delta Y\mid D=0),
    \end{align*}
    where the second equality follows from the proof of Theorem~\ref{thm:main_npm}. Moreover, we showed in the proof of Theorem~\ref{thm:main_npm} that $\sigma_{0s}^2 = \operatorname{Var}(\Delta Y -g_{0s}\mid D=0) = E[\operatorname{Var}(\Delta Y \mid X,D=0)\mid D=0]$, so we have 
    \begin{align*}
        \text{(2): }\sigma^2_{w0s}&:=\operatorname{Var}(\Delta Y\mid D=0) \\
        &\overset{\text{LTV}}{=} E[\operatorname{Var}(\Delta Y \mid X,D=0)\mid D=0] + \operatorname{Var}(E[\Delta Y \mid X,D=0]\mid D=0) \\
        &= \sigma_{0s}^2 + \operatorname{Var}(g_{0s}\mid D=0),
    \end{align*}
    which suggests that $\sigma^2_{w0s} \ge \sigma_{0s}^2$. Additionally, we have 
    \begin{align*}
        R^2_{g_0 \sim g_{0s}|1,D=0} &:= 1 - \frac{\operatorname{Var}(g_0 - g_{0s}\mid D=0)}{\operatorname{Var}(g_0\mid D=0)} \\
        &= \frac{-\operatorname{Var}(g_{0s}\mid D=0) + 2(E[g_0g_{0s}\mid D=0]-E[g_0\mid D=0]E[g_{0s}\mid D=0])}{\operatorname{Var}(g_0\mid D=0)}  \\
        &\overset{\text{LTE}}{=} \frac{\operatorname{Var}(g_{0s}\mid D=0)}{\operatorname{Var}(g_0\mid D=0)}. 
    \end{align*}
    Then, we define the corresponding Cohen's partial $f^2$ as: 
    \begin{align*}
        f^2_{g_0 \sim g_{0s}|1,D=0} &:= \frac{R^2_{g_0 \sim g_{0s}|1,D=0}}{1-R^2_{g_0 \sim g_{0s}|1,D=0}} \\
        &= \frac{\operatorname{Var}(g_{0s}\mid D=0)}{\operatorname{Var}(g_{0}\mid D=0) - \operatorname{Var}(g_{0s}\mid D=0)}. 
    \end{align*} 
    Similarly, we have that
    \begin{align*}
        R^2_{\Delta Y \sim g_{0s}|1,D=0} &:= 1 - \frac{\operatorname{Var}(\Delta Y - g_{0s}\mid D=0)}{\operatorname{Var}(\Delta Y\mid D=0)} \\
        &= \frac{-\operatorname{Var}(g_{0s}\mid D=0) + 2(E[\Delta Yg_{0s}\mid D=0] - E[\Delta Y\mid D=0]E[g_{0s}\mid D=0])}{\operatorname{Var}(\Delta Y\mid D=0)} \\
        &\overset{\text{LTE}}{=} \frac{\operatorname{Var}(g_{0s}\mid D=0)}{\operatorname{Var}(\Delta Y\mid D=0)}.
    \end{align*}
    Combined with (2), this yields
    \begin{align}\label{eq:R2-DY-g0s}
        1 - R^2_{\Delta Y \sim g_{0s}|1,D=0} = \frac{\sigma_{0s}^2}{\operatorname{Var}(\Delta Y\mid D=0)}.
    \end{align}
    
    Accordingly, we have the following:
    \begin{align*}
        \rho_{w0}^2 &= \frac{\operatorname{Cov}^2(\Delta Y, \alpha_{0} - \alpha_{0s}\mid D=0)}{\operatorname{Var}(\Delta Y\mid D=0)\operatorname{Var}(\alpha_{0} - \alpha_{0s}\mid D=0)} \\
        &\overset{\text{(1)}}{=} \frac{\operatorname{Cov}^2(g_0 - g_{0s}, \alpha_{0} - \alpha_{0s}\mid D=0)}{\operatorname{Var}(g_0 - g_{0s}\mid D=0)\operatorname{Var}(\alpha_{0} - \alpha_{0s}\mid D=0)} \times \frac{\operatorname{Var}(g_0 - g_{0s}\mid D=0)}{\operatorname{Var}(\Delta Y\mid D=0)} \\
        &\overset{\text{(2)}}{=} \rho_0^2 \times \frac{\operatorname{Var}(g_0 - g_{0s}\mid D=0)}{\sigma_{0s}^2 + \operatorname{Var}(g_{0s}\mid D=0)} \\
        &\overset{\text{def.}}{=} \rho_0^2 \times \frac{1}{\frac{1}{C_{0\Delta Y}^2} + \frac{\operatorname{Var}(g_{0s}\mid D=0)}{\operatorname{Var}(g_0 - g_{0s}\mid D=0)}} = \rho_0^2 \times \frac{1}{\frac{1}{C_{0\Delta Y}^2} + \frac{\operatorname{Var}(g_{0s}\mid D=0)}{\operatorname{Var}(g_0\mid D=0) - \operatorname{Var}(g_{0s}\mid D=0)}} \\
        &= \rho_0^2 \times \frac{1}{\frac{1}{C_{0\Delta Y}^2} + f^2_{g_0 \sim g_{0s}|1,D=0}} \left( = \rho_0^2 \times \frac{1}{\frac{1}{C_{0\Delta Y}^2} + \frac{1}{1-R^2_{g_0 \sim g_{0s}|1,D=0}} - 1}\right). 
    \end{align*}
    Note that $\frac{1}{C_{0\Delta Y}^2} + \frac{1}{1-R^2_{g_0 \sim g_{0s}|1,D=0}} - 1 \ge 1$, so we have $\rho_{w0}^2 \le \rho_0^2$. Alternatively, we have that 
    \begin{align*}
        \rho_{w0}^2 &\overset{\text{(1)}}{=} \rho_0^2 \times \frac{\operatorname{Var}(g_0 - g_{0s}\mid D=0)}{\operatorname{Var}(\Delta Y - g_{0s}\mid D=0)} \times \frac{\operatorname{Var}(\Delta Y - g_{0s}\mid D=0)}{\operatorname{Var}(\Delta Y\mid D=0)} \\
        &\overset{\text{def.}}{=} \rho_0^2 \times C_{0\Delta Y}^2 \times (1 - R^2_{\Delta Y \sim g_{0s}|1,D=0}). 
    \end{align*}
    \end{enumerate}
\end{proof}

\section{Supplementary results for the empirical application}
\label{app:additional-results}

Our empirical application studies the effect of the increase in minimum wage on teen employment. This section presents additional results, including those for the main analysis, and results that allow for one year anticipation.

\subsection{Nuisance learners: specification and RMSE}

Table~\ref{tab:mw_learners} presents the specifications of the machine learning methods used to estimate the first-stage nuisance functions. We consider four learners---a parametric model, ridge regression, lasso, and random forest---using the covariates from \citet{callaway2021difference}. Table~\ref{tab:mw_RMSE} reports the corresponding cross-fitted root mean squared errors (RMSE) for predicting $D$ and $\Delta Y$. Random Forest attains the lowest RMSE for both, and we therefore use it as the learner for our main analysis.

\begin{table}[H]
\centering
\footnotesize
\renewcommand{\arraystretch}{0.9}
\setlength{\tabcolsep}{4pt}
\resizebox{\textwidth}{!}{
\begin{tabular}{llll}
\toprule
\textbf{Learner} & \textbf{Variables} & \textbf{Package} & \textbf{Tuning grid} \\
\midrule
\makecell[l]{Parametric: \\
Linear for $g_{0s}$ \\
Logistic for $\pi$}  & \makecell[l]{\textit{region, white, hs, pov,} \\
\textit{lpop, $\textit{lpop}^2$} \\ \textit{lmedinc, $\textit{lmedinc}^2$}}
 &
\makecell[l]{\texttt{stats} (\texttt{lm})\\
\texttt{stats} (\texttt{glm})} &
N/A\\
\hline
\makecell[l]{Ridge: \\
Linear for $g_{0s}$ \\
Logistic for $\pi$} &
\makecell[l]{\textit{region}; cubic polynomials \\ in
\textit{lpop, white, pov, hs, lmedinc}; \\
\textit{region}-polynomial interactions} &
\texttt{glmnet} &
\makecell[l]{$\alpha_{\text{EN}} = 0$ (ridge) \\ $\lambda$ selected from \texttt{cv.glmnet}} \\
\hline
\makecell[l]{Lasso: \\
Linear for $g_{0s}$ \\
Logistic for $\pi$} &
\makecell[l]{\textit{region}; cubic polynomials \\ in
\textit{lpop, white, pov, hs, lmedinc}; \\
\textit{region}-polynomial interactions} &
\texttt{glmnet} &
\makecell[l]{$\alpha_{\text{EN}} = 1$ (lasso) \\ $\lambda$ selected from \texttt{cv.glmnet}} \\
\hline
\makecell[l]{Random Forest \\ ($\texttt{num.trees}$ \\ $ = 1,000$)} &
\makecell[l]{\textit{region, white, hs, pov,} \\
\textit{lpop, lmedinc}}  &
\texttt{ranger} &
\makecell[l]{\texttt{mtry} $\in \{2,4,6\}$ \\
             \texttt{min.node.size} $\in \{10, 15, 25, 50, 75,\dots,150\}$ \\
             \texttt{splitrule} $\in \{\texttt{variance}, \texttt{extratrees}\}$} \\
\bottomrule
\end{tabular}}                    
\caption{Machine learning methods used for first-stage nuisance estimation in the minimum wage example. Covariates enter linearly. Models are selected by minimizing the RMSE.}
\label{tab:mw_learners}
\end{table}

\begin{table}[ht]
\centering
\begin{tabular}{lll}
  \toprule
Methods & RMSE($D$) & RMSE($\Delta Y$)  \\ 
  \hline
  Parametric & 0.420 & 0.1548  \\ 
  Lasso & 0.3803 & 0.1540 \\ 
  Ridge & 0.3859 & 0.1539  \\ 
  \textbf{Random Forest} & \textbf{0.3719} & \textbf{0.1529} \\  
  \hline
  Best & 0.3719 & 0.1529  \\ 
  \bottomrule
\end{tabular}
\caption{Cross-fitted RMSEs for predicting $\Delta Y$ and $D$ in the minimum wage application. Random Forest achieves the lowest RMSE for both the outcome evolution and the propensity score and is used as the learner in our main analysis.}
\label{tab:mw_RMSE}
\end{table}

\subsection{Multiplier bootstrap for uniform confidence bands}

Algorithm~\ref{alg:Mboot} presents the multiplier bootstrap procedure of \citet{callaway2021difference}, which we use to construct the uniform $95\%$ confidence bands shown in Figure~\ref{fig:nt_ATT_multi_main}. The procedure is computationally fast and, by reweighting rather than resampling observations, guarantees that every bootstrap iteration retains units from both treated and control groups.

\begin{algorithm}[H]
\caption{Multiplier Bootstrap: MBoot($\bm{\varphi}^0_{\bm{\theta_s}}$)}
\begin{algorithmic}
\State \textbf{Input:} An $n \times T$ influence-function matrix $\bm{\varphi}^0_{\bm{\theta}_s}$, with columns corresponding to time periods ($T=1$ in the canonical setup). Significance level $\alpha$, and number of bootstrap draws $B$. 
    \For{$b = 1,\dots,B$}
        \State Draw weights $\{V_i\}_{i=1}^n$ from Mammen's distribution: let $\kappa = (\sqrt{5}+1)/2$, 

        \vspace{-0.5em}
        \[
V_i \overset{\text{i.i.d.}}{\sim}
\begin{cases}
1-\kappa & \text{with probability } \kappa/\sqrt{5},\\
\kappa   & \text{with probability } 1-\kappa/\sqrt{5}.
\end{cases}
\]
        \State Compute $\widehat{R}^* = \frac{1}{\sqrt{n}}\sum_{i=1}^n\left(V_i \times \bm{\varphi}^0_{\bm{\theta_s}}(Z_i)\right)$. Denote the $t$-th element by $\widehat{R}^*(t)$ for $t = 1,\dots,T$. 
    \EndFor
    \State Estimate bootstrap standard deviations: $\widehat{\sigma}_{\bm{\varphi}^0_{\bm{\theta_s},t}} = \text{IQR}(\widehat{R}^*(t))/\left(\Phi^{-1}(0.75)-\Phi^{-1}(0.25)\right)$. 
    \State For each bootstrap draw compute $\text{t-test}_b = \max_t\left(\left|\widehat{R}^*(t)\right|\big/\widehat{\sigma}_{\bm{\varphi}^0_{\bm{\theta_s},t}}\right)$. 
    \State Calculate the critical value as $\hat{c}_{1-\alpha}= \text{Quantile}_{1-\alpha}(\{\text{t-test}_b\}_{b=1}^B)$.
\State \textbf{Return:} Construct confidence band for $\bm{\theta_s}(t)$ as $\widehat{\text{CI}}_{1-\alpha}(t) = \left[\widehat{\bm{\theta_s}}(t) \pm \left(\hat{c}_{1-\alpha}\widehat{\sigma}_{\bm{\varphi}^0_{\bm{\theta_s},t}}\Big/\sqrt{n}\right)\right]$. 
\end{algorithmic}
\label{alg:Mboot}
\end{algorithm}

\subsection{Covariate balance check}

Table~\ref{tab:balance_main} reports the means and standard deviations of each covariate in the treated and control groups, together with conventional balance diagnostics---the standardized mean difference and the variance-ratio deviation. Table~\ref{tab:obs_strength} reports the bias decomposition obtained by adjusting for each observed covariate separately. This provides a natural diagnostic for covariate imbalance, showing directly how each covariate affects the results.

\begin{table}[H]
\centering
\renewcommand{\arraystretch}{0.65}
\setlength{\tabcolsep}{4pt}
\resizebox{0.85\textwidth}{!}{
\begin{tabular}{lcccccc}
\toprule
&
\multicolumn{2}{c}{Untreated}
&
\multicolumn{2}{c}{Treated}
&
Std. Mean Diff.
&
Var. Ratio Dev.
\\
\cmidrule(l{3pt}r{3pt}){2-3}
\cmidrule(l{3pt}r{3pt}){4-5}
\cmidrule(l{3pt}r{3pt}){6-6}
\cmidrule(l{3pt}r{3pt}){7-7}
&
$\mu_0$
&
$\sigma_0^2$
&
$\mu_1$
&
$\sigma_1^2$
&
$\frac{\mu_1-\mu_0}{\sqrt{(\sigma_1^2 + \sigma_0^2)/2}}$
&
$\frac{\sigma_1^2}{\sigma_0^2} - 1$
\\
\midrule
Midwest & 0.336 & 0.223 & 0.483 & 0.250 & 0.303 & 0.120 \\ 
South   & 0.593 & 0.241 & 0.301 & 0.212 & -0.613 & -0.128 \\
West    & 0.072 & 0.067 & 0.216 & 0.169 & 0.419 &  1.536\\
\midrule
Population (1000s) & 53.425 & 23,438 & 84.142 & 32,850 & 0.183 & 0.402 \\
log(Population)    & 3.018  & 1.584 & 3.467  & 1.781 & 0.346 & 0.125 \\
\midrule
Median Inc. (1000s) & 31.889 & 54.565  & 33.080 & 66.654 & 0.153 & 0.222 \\
log(Median Inc.)    & 3.438  & 0.048 & 3.471  &  0.053 & 0.150 & 0.111 \\
\midrule
White         & 0.826 & 0.026 & 0.885 & 0.018 & 0.397 & -0.331 \\
HS Graduates  & 0.553 & 0.006 & 0.577 & 0.004 & 0.336 & -0.382 \\
Poverty Rate  & 0.157 & 0.004 & 0.138 & 0.003 & -0.324 & -0.378 \\
\bottomrule
\end{tabular}
}
\caption{
{\footnotesize Covariate balance between $584$ treated counties ($29.8\%$) in states that raised the minimum wage above the federal level in 2007 and $1,377$ never-treated control counties. Columns report, from left to right, means, variances, standardized mean differences, and deviations of the variance ratio from one.}}
\label{tab:balance_main}
\end{table}

\begin{table}[H]
\centering
\footnotesize                       
\setlength{\tabcolsep}{5pt}         
\renewcommand{\arraystretch}{0.8}  
\resizebox{0.9\textwidth}{!}{%
\begin{tabular}[t]{lcccc}
\toprule
\multicolumn{1}{c}{ } & \multicolumn{1}{c}{Bias (from $X_j$)} & \multicolumn{1}{c}{Alignment ($\rho_{0,j}$)} & \multicolumn{1}{c}{Trend ($C_{0\Delta Y,j}$)} & \multicolumn{1}{c}{Imbalance ($C_{0D,j}$)} \\
\cmidrule(l{3pt}r{3pt}){2-2} \cmidrule(l{3pt}r{3pt}){3-3} \cmidrule(l{3pt}r{3pt}){4-4} \cmidrule(l{3pt}r{3pt}){5-5}
& $\theta_{s,j} - \theta_{s,\emptyset}$ & $\operatorname{Cor}\left(g_{0s,j}, O_{X_j}\mid D=0\right)$ & $\sqrt{\eta^2_{\Delta Y \sim X_j\mid D=0}}$ &
$\sqrt{\chi^2(P_{X_j\mid1}||P_{X_j\mid0})}$ \\
\midrule
Region (Overall) & -0.0010 & 0.291 & 0.036 & 0.589\\
 & (0.0015) & (0.425) & (0.023) & (0.036)\\
\hline
Region (Midwest) & 0.0016 & -1.000 & 0.032 & 0.263\\
 & (0.0006) & (0.498) & (0.021) & (0.030)\\
Region (South) & 0.0013 & -1.000 & 0.015 & 0.504\\
 & (0.0012) & (1.921) & (0.030) & (0.029)\\
Region (West) & -0.0025 & 1.000 & 0.033 & 0.432\\
 & (0.0011) & (0.611) & (0.028) & (0.039)\\
\hline
Population (1000s) & -0.0016 & 0.765 & 0.045 & 0.295\\
 & (0.0007) & (0.373) & (0.010) & (0.032)\\
log(Population) & -0.0017 & 1.000 & 0.023 & 0.286\\
 & (0.0007) & (1.834) & (0.027) & (0.032)\\
\hline
Median Inc. (1000s) & -0.0002 & 0.051 & 0.110 & 0.179\\
 & (0.0006) & (0.186) & (0.018) & (0.034)\\
log(Median Inc.) & -0.0002 & 0.058 & 0.110 & 0.162\\
 & (0.0006) & (0.210) & (0.019) & (0.037)\\
 \hline
White & 0.0011 & -0.307 & 0.060 & 0.406\\
 & (0.0010) & (0.261) & (0.026) & (0.029)\\
HS degree & -0.0026 & 0.607 & 0.072 & 0.389\\
 & (0.0010) & (0.246) & (0.016) & (0.026)\\
Poverty rate & 0.0001 & -0.029 & 0.076 & 0.251\\
 & (0.0008) & (0.283) & (0.026) & (0.026)\\
\midrule
$X_{\text{all}}$ & -0.0089 &0.194 & 0.117 & 2.548 \\
& (0.0079) &(0.166) & (0.045) & (0.170)\\
\bottomrule
\multicolumn{5}{l}{\rule{0pt}{1.5em}\textit{Note:}
$S_{0,\emptyset} = \sqrt{\sigma_{0s,\emptyset}^2\nu_{0s,\emptyset}^2} = 0.154$
with an SE of $0.006$.} \\
\multicolumn{5}{l}{
$\widehat{\sigma}_{0s,\emptyset}^2 = 0.024$ with an SE of $0.002$,
and $\widehat{\nu}_{0s,\emptyset}^2 = 1$ because there are no covariates.}\\
\end{tabular}%
}
\caption{{\footnotesize Decomposition of the observed confounding strength. For each covariate, the first row reports the estimates and the second row reports the corresponding standard errors. The decomposition follows from $\theta_{s,j} - \theta_{s,\emptyset} = -\rho_{0,j}C_{0\Delta Y,j}C_{0D,j}S_{0,\emptyset}$. If the plug-in estimate of $\rho_{0,j}$ exceeds one in absolute value, it is capped at $\pm 1$, in which case the product of the displayed factors may not match the estimated bias exactly.}}
\label{tab:obs_strength}
\end{table}
    
\subsection{Additional sensitivity results}

Table~\ref{tab:gain_metrics} reports additional empirical benchmarking results for the minimum wage application. The gain metrics $G_{0\Delta Y,j}$ and $G_{0D,j}$, defined in Section~\ref{sec:Statistical Inference of benchmark components}, measure the explanatory power of each observed covariate for the outcome evolution and the treatment indicator, respectively; these are the gain metrics used in the contour plots in the main text. The correlations $\rho_{0,j}$, defined in Section~\ref{sec:rho-benchmark} for benchmarking $\rho_{0}$, are all bounded above by $0.3$ in absolute value.

\begin{table}[ht]
\centering
\begin{tabular}{lcccc}
\toprule
& \multicolumn{2}{c}{Gain Metrics} & Correlation & Change in estimate \\
\cmidrule(lr){2-5}
Observed covariate 
& $G_{0\Delta Y,j}$ 
& $G_{0D,j}$ 
& $\rho_{0,j}$ 
& $\theta_{s} - \theta_{s,-j}$ \\
\midrule
lmedinc & 0.0051 & 0.3701 &  -0.2296 &  0.0036\\
region &  0.0049 &  0.7389 &  0.1327 & -0.0025\\
white   &  0.0021 & 0.5190 & -0.0544 & 0.0006 \\
pov & 0.0047 &  0.2392 &  0.0834 &  -0.0011 \\
\bottomrule
\end{tabular}
\caption{Explanatory power of observed covariates, minimum wage example. All estimates are debiased and cross-fitted.}
\label{tab:gain_metrics}
\end{table}

\begin{table}[htbp]
\centering
\begin{tabular}{lrrrrrr}
\toprule
& \multicolumn{2}{c}{$|\rho_0| = |\rho_{0,j}|$} & \multicolumn{2}{c}{$|\rho_0| = 0.3$} & \multicolumn{2}{c}{$|\rho_0| = 1$} \\
\cmidrule(lr){2-3} \cmidrule(lr){4-5} \cmidrule(lr){6-7}
Covariate & Lower & Upper & Lower & Upper & Lower & Upper \\
\midrule
lmedinc & $-0.0622$ & $-0.0147$ & $-0.0613$ & $-0.0112$ & $-0.0840$ & $0.0122$ \\
white   & $-0.0554$ & $-0.0161$ & $-0.0648$ & $-0.0106$ & $-0.1021$ & $0.0253$ \\
pov     & $-0.0541$ & $-0.0156$ & $-0.0595$ & $-0.0144$ & $-0.0751$ & $-0.0002$ \\
region  & $-0.0838$ & $-0.0098$ & $-0.0841$ & $0.0011$ & $-0.1679$ & $0.0828$ \\
\bottomrule
\end{tabular}
\caption{One-sided 95\% confidence intervals obtained by benchmarking against observed covariates, minimum wage example. The intervals account for sampling uncertainty in the estimated benchmark components.}
\label{tab:benchmark_conf_bound}
\end{table}

\begin{table}[H]
\centering
\begin{tabular}{lrr}
\toprule
Pre-trend extrapolation method & Lower & Upper \\
\midrule
Bias magnitude & $-0.0922$ & $0.0272$ \\
Bias factors   & $-0.0917$ & $0.0270$ \\
\bottomrule
\end{tabular}
\caption{One-sided 95\% confidence intervals obtained by benchmarking against the $2006$ pre-trend deviation, minimum wage example. The intervals account for sampling uncertainty in the estimated benchmark components.}
\label{tab:benchmark_pretrend_bounds}
\end{table}

Table~\ref{tab:benchmark_conf_bound} reports one-sided 95\% confidence bounds implied by an unobserved confounder comparable to each observed covariate. In the first pair of columns, the confounder is set comparable in both strength and alignment to the observed covariate. In the remaining two pairs, we instead fix the alignment at $|\rho_0|=0.3$ (a moderate scenario) and $|\rho_0|=1$ (an adversarial scenario), while keeping its strength in the outcome trend and treatment selection comparable to the covariate. All bounds account for the uncertainty in estimating the benchmark components, using the results of Section~\ref{sec:est-inf-bench}. 

Table~\ref{tab:benchmark_pretrend_bounds} reports the one-sided 95\% confidence bounds implied by the 2006 pre-trend deviation, under the two extrapolation strategies introduced in Section~\ref{sec:bench-pre-trend}, with sampling uncertainty accounted for using the results of Section~\ref{sec:est-inf-bench}. Since the estimated effect is negative, we focus on the upper bound, which corresponds to the direction of bias relevant for overturning the original conclusion. The two strategies yield similar upper bounds, both indicating that a confounder generating bias comparable to the 2006 pre-trend deviation would be able to explain away the negative effect. 

\subsection{Results allowing for anticipation}
\label{app:anticipation}

In the main text, we show that an unobserved confounder would have to be stronger or more adversarial than any observed covariate to overturn the result. Explaining the 2006 pre-trend deviation would be more demanding still. Once we rule out such confounders, a natural explanation is that the pre-trend reflects anticipation, meaning that treated units began responding to the policy before it formally took effect. In this case, the anticipation effect is part of the total treatment effect. This section therefore reports results that allow for one year of treatment anticipation.

Here we ascribe the full decrease in teen employment from 2005 to 2007 to the minimum wage increase in 2007. As shown in Figure~\ref{fig:nt_ATT_multi_anticipation}, the point estimate of the 2007 treatment effect then changes from a $3.66\%$ decrease in teen employment (without anticipation) to a $7.4\%$ decrease (with anticipation). Table~\ref{tab: MW-results-robustness_anticipation} reports the corresponding minimal sensitivity reporting. The robustness value $\text{RV}_{\theta^*=0,\,\alpha=0.05}$ increases from $4.38\%$ (without anticipation) to $10.46\%$ (with anticipation), and the extreme robustness value $\text{XRV}_{\theta^*=0,\,\alpha=0.05}$ increases from $0.2\%$ (without anticipation) to $1.21\%$ (with anticipation). As explained in Section~\ref{sec:statistics}, these statistics provide a quick summary of how robust the estimated effect is to the presence of omitted variables. Allowing for one year of anticipation therefore makes the estimated negative effect in 2007 more robust to omitted confounding, since a larger amount of confounding would now be required to overturn it.

\begin{figure}[H]
        \centering
        \includegraphics[width=0.66\textwidth]{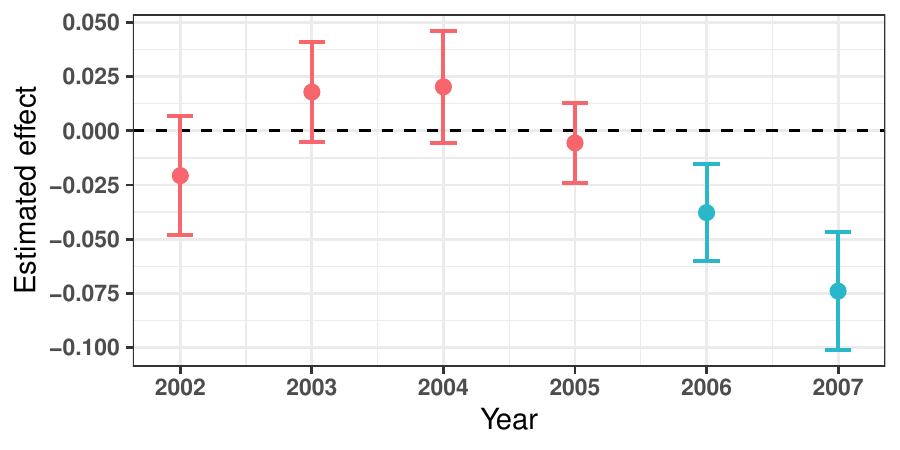}
    \caption{{\footnotesize The effect of the minimum wage increase on teen employment under the conditional parallel trends assumption, using the \textbf{never-treated} units as the control group and allowing for \textbf{one year of anticipation}. Estimates use DML with a canonical DiD design. For the pre-treatment periods, the preceding year is used as the base period; for 2006 and 2007, 2005 is used as the base period. Red and blue lines give point estimates with uniform 95\% confidence bands for the pre-treatment periods and the treatment effects, respectively.}}
    \label{fig:nt_ATT_multi_anticipation}
\end{figure}

\begin{table}[H]
\centering
{\footnotesize
\begin{tabular}{ccccc}
\toprule
\multicolumn{3}{c}{\textbf{Results Under Conditional Parallel Trends}} 
& \multicolumn{2}{c}{\textbf{Robustness Values}} \\
\cmidrule(lr){1-3} \cmidrule(lr){4-5}
\textbf{Short Estimate ($\widehat{\theta}_{s,2007}$)} 
& \textbf{Std. Error} 
& \textbf{Confidence Interval} 
& $\text{RV}_{\theta^*=0,\ \alpha=0.05}$ 
& $\text{XRV}_{\theta^*=0,\ \alpha=0.05}$\\
\midrule
-0.074
& 0.0109
& [-0.0954;\;-0.0526] 
& 10.46\% 
& 1.21\% \\
\bottomrule
\end{tabular}
}
\caption{{\footnotesize Minimal Sensitivity Reporting: Increase in Minimum Wage on Teen Employment allowing for one year anticipation. We use random forests for outcome evolution and propensity scores, with parameters chosen by cross-validation.}}
\label{tab: MW-results-robustness_anticipation}
\end{table}

We now turn to the benchmarking analysis. Table~\ref{tab:benchmark_conf_bound_anticipation} is the counterpart of Table~\ref{tab:benchmark_conf_bound} under one year anticipation. It shows that a latent confounder comparable to any of the three observed covariates---or comparable in strength with alignment fixed at $|\rho_0|=0.3$---is not strong enough to explain away the estimated effect. Figure~\ref{fig:mw_contour_all_anticipation} presents the contour plots. The red diamonds mark the scenarios implied by the observed covariates, positioned by their gain metrics (Table~\ref{tab:gain_metrics_anticipation}).

\begin{table}[htbp]
\centering
\begin{tabular}{lrrrrrr}
\toprule
& \multicolumn{2}{c}{$|\rho_0| = |\rho_{0,j}|$} & \multicolumn{2}{c}{$|\rho_0| = 0.3$} & \multicolumn{2}{c}{$|\rho_0| = 1$} \\
\cmidrule(lr){2-3} \cmidrule(lr){4-5} \cmidrule(lr){6-7}
Covariate & Lower & Upper & Lower & Upper & Lower & Upper \\
\midrule
lmedinc & $-0.1057$ & $-0.0464$ & $-0.1068$ & $-0.0425$ & $-0.1548$ & $0.0047$ \\
white   & $-0.0978$ & $-0.0500$ & $-0.1086$ & $-0.0389$ & $-0.1600$ & $0.0127$ \\
region  & $-0.1360$ & $-0.0255$ & $-0.1382$ & $-0.0106$ & $-0.2716$ & $0.1226$ \\
\bottomrule
\end{tabular}
\caption{One-sided 95\% confidence intervals obtained by benchmarking against observed covariates, minimum wage example allowing for anticipation. The intervals account for sampling uncertainty in the estimated benchmark components.}
\label{tab:benchmark_conf_bound_anticipation}
\end{table}

\begin{figure}[H]
    \centering
    \begin{subfigure}[t]{0.49\textwidth}
        \centering
        \makebox[\linewidth][c]{
        \includegraphics[width=1.4\linewidth]{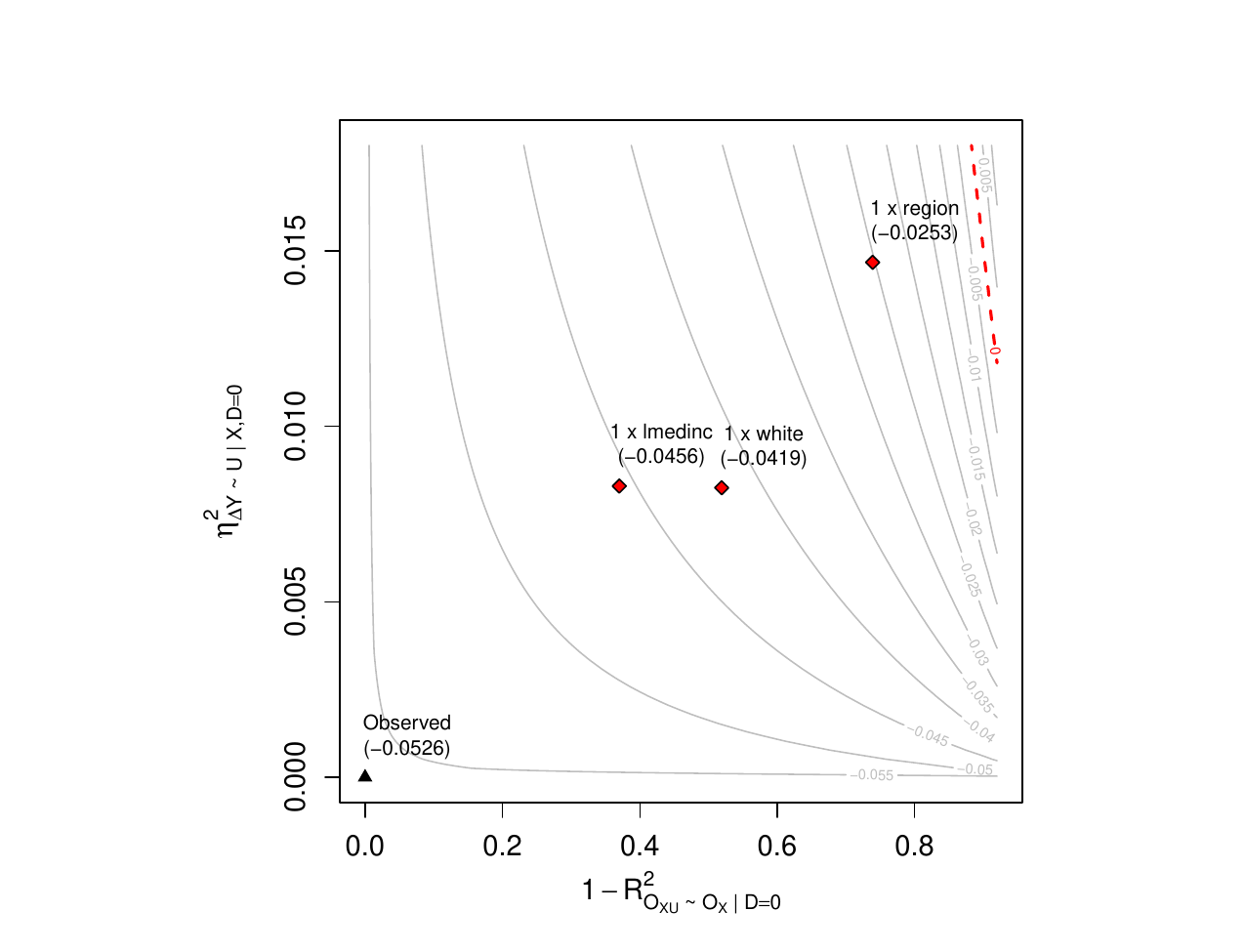}}
        \caption{\scriptsize{Upper limit conf. bound, $|\rho_{0}|=0.3$.}}
    \end{subfigure}\hfill
    \begin{subfigure}[t]{0.49\textwidth}
        \centering
        \makebox[\linewidth][c]{
        \includegraphics[width=1.4\linewidth]{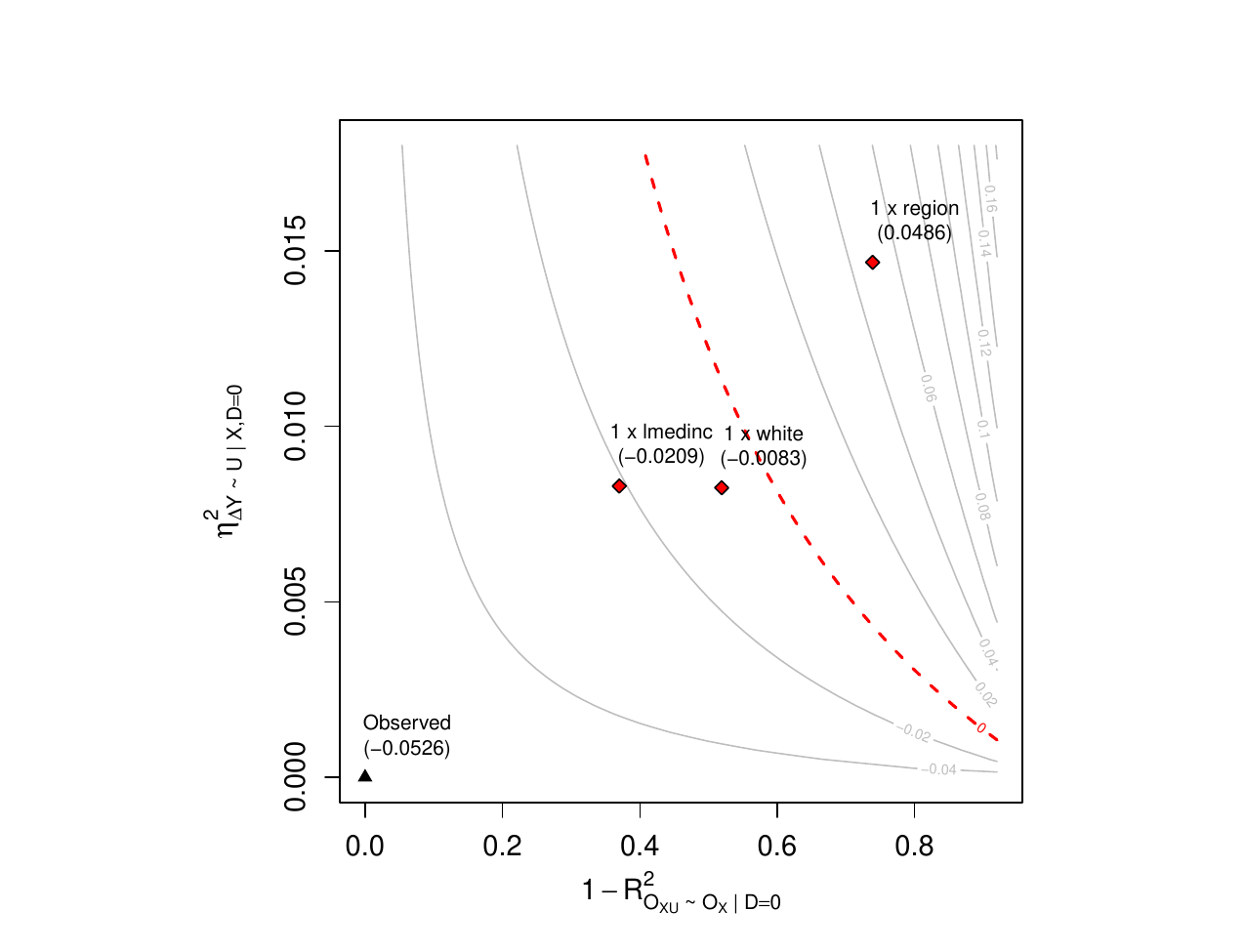}}
        \caption{\scriptsize{Upper limit conf. bound, $|\rho_{0}|=1$.}}
    \end{subfigure}
    \caption{{\footnotesize Sensitivity contour plots for the minimum wage example at a significance level of $\alpha = 0.05$.}}
    \label{fig:mw_contour_all_anticipation}
\end{figure} 

\begin{table}[H]
\centering
\begin{tabular}{lcccc}
\toprule
& \multicolumn{2}{c}{Gain Metrics} & Correlation & Change in estimate \\
\cmidrule(lr){2-5}
Observed covariate 
& $G_{0\Delta Y,j}$ 
& $G_{0D,j}$ 
& $\rho_{0,j}$ 
& $\theta_{s} - \theta_{s,-j}$ \\
\midrule
lmedinc & 0.0083 & 0.3701 &  -0.2915 &  0.0068\\
region &  0.0147 &  0.7389 &  0.2156 & -0.0084\\
white   &  0.0082 & 0.5190 & -0.0801 & 0.0021 \\
\bottomrule
\end{tabular}
\caption{Explanatory power of observed covariates, minimum wage example allowing for anticipation. All estimates are debiased and cross-fitted.}
\label{tab:gain_metrics_anticipation}
\end{table}

Overall, the analysis shows that allowing for anticipation yields a larger estimated negative effect that is more robust to omitted confounding.

\clearpage

\putbib
\end{bibunit}

\end{document}